%% file: paper_ext.tex
\documentclass[prodmode,acmtops]{acmsmall}

\usepackage{times}
\usepackage{xcolor}
\usepackage{datetime}
\usepackage{url}
\usepackage{hyperref}
\usepackage{verbatim}

\usepackage{amssymb,amsmath}
\usepackage{amsfonts,color}
\usepackage[english]{babel}
\usepackage{comment}
\usepackage[mathcal]{euscript}
\usepackage{graphicx}
\usepackage{enumerate}
\usepackage[caption=false,font=footnotesize]{subfig}
\usepackage{array}
\usepackage{tikz}
\usepackage{colortbl}
\usepackage{multirow}
\usepackage{framed}
\usepackage{paralist}
\newcolumntype{P}[1]{>{\centering\arraybackslash}p{#1}}
\newcolumntype{M}[1]{>{\centering\arraybackslash}m{#1}}
\newcolumntype{C}[1]{>{\centering\arraybackslash}c{#1}}
\definecolor{cellgray}{RGB}{220,220,220}

\newcommand{\pick} [1] {\stackrel{#1}{\leftarrow}}

\usepackage[]{algorithm}
\usepackage[noend]{algpseudocode}

\usepackage{blindtext}
\usepackage{todonotes}

\newcommand\NoThen{\renewcommand\algorithmicthen{}}

\hypersetup{
  colorlinks,
  linkcolor={red!50!black},
  citecolor={blue!50!black},
  urlcolor={blue!80!black}
}

\newcommand{\ignore}[1]{}

\newcommand{\la}{\lambda}

\newcommand{\A}{\mathcal{A}}

\newcommand{\negl}{\mathrm{negl}}

\newcommand{\Cl}{\mathsf{Clock}}

\newcommand{\Com}{\mathsf{Com}}

\newcommand{\rec}{\mathsf{Rec}}

\newcommand{\set}[1]  {\left\{#1\right\}}

\newcommand{\mr}{\mathrm}
\newcommand{\mc}{\mathcal}
\newcommand{\tc}[1]{\tilde{\mathcal{#1}}}
\newcommand{\mbf}{\mathbf}
\newcommand{\Ch}{\mathcal{C}}

\newcommand{\EA}{\textsl{EA}}
\newcommand{\BB}{\textsl{BB}}
\newcommand{\VC}{\textsl{VC}}
\newcommand{\trustees}{\textsl{trustees}}
\newcommand{\Trustees}{\textsl{Trustees}}
\newcommand{\Trustee}{\textsl{Trustee}}
\newcommand{\trustee}{\textsl{trustee}}
\newcommand{\VOTE}{\texttt{VOTE}}
\newcommand{\VOTEP}{\texttt{VOTE\_P}}
\newcommand{\ENDORSE}{\texttt{ENDORSE}}
\newcommand{\ENDORSEMENT}{\texttt{ENDORSEMENT}}
\newcommand{\ANNOUNCE}{\texttt{ANNOUNCE}}
\newcommand{\RECREQ}{\texttt{RECOVER}-\texttt{REQUEST}}
\newcommand{\RECRES}{\texttt{RECOVER}-\texttt{RESPONSE}}
\newcommand{\UCERT}{$\mathsf{UCERT}$}

\newcounter{claimcounter}
\numberwithin{claimcounter}{theorem}
\newenvironment{myclaim}{\refstepcounter{claimcounter}{\par\noindent \textsc{\small{Claim \theclaimcounter}}:}}{}
\newenvironment{claimproof}{\par\noindent\textit{Proof of Claim \theclaimcounter:}}{\hspace*{\fill} \emph{(End of Claim \theclaimcounter)} $\dashv$}

\newenvironment{boxfig}[2]{%
     \begin{figure}[h!]
     \newcommand{\FigCaption}{#1}
     \newcommand{\FigLabel}{#2}
     \begin{center}
       \begin{small}
         \begin{tabular}{@{}|@{~~}l@{~~}|@{}}
           \hline
           \rule[-1.5ex]{0pt}{1ex}\begin{minipage}[b]{.95\linewidth}
             \vspace{1ex}
             \smallskip
             }{%
           \end{minipage}\\
           \hline
         \end{tabular}
       \end{small}
      \caption{\FigCaption}
     \end{center}
     \vspace{-0.5cm}
   \end{figure}
}

\newif\ifextended
\newif\ifjournal
\journaltrue

\acmYear{2016}
\acmMonth{8}

\doi{0000001.0000001}

\issn{1234-56789}
\newif\ifextended
\extendedtrue

\begin{document}
\markboth{N. Chondros et al.}{Distributed, End-to-end Verifiable, and Privacy-Preserving Internet Voting Systems}

\title{Distributed, End-to-end Verifiable, and Privacy-Preserving Internet Voting Systems} 
\author{
    Nikos Chondros
    \affil{University of Athens}
    Bingsheng Zhang
    \affil{Lancaster University}
    Thomas Zacharias
    \affil{University of Athens}
    Panos Diamantopoulos
    \affil{University of Athens}
    Stathis Maneas
    \affil{University of Toronto}
    Christos Patsonakis
    \affil{University of Athens}
    Alex Delis
    \affil{University of Athens}
    Aggelos Kiayias
    \affil{University of Edinburgh}
    Mema Roussopoulos
    \affil{University of Athens}
}
\begin{abstract}
  \input{abstract}
\end{abstract}

%
%

%
%


\begin{bottomstuff}
This work is supported in part by ERC Starting Grants \#  279237 and \#  259152 funded
by the European Research Council, and the FINER Project funded by the General
Secretariat for Research and Technology ARISTEIA Program.

Author's addresses: N. Chondros, T. Zacharias, P. Diamantopoulos, C. Patsonakis, A. Delis, {and} M. Roussopoulos, 
Department of Informatics and Telecommunications, 
University of Athens,
Panepistiomiopolis, Ilisia, 157 84, Athens, Greece,
emails: \{n.chondros, thzacharias, panosd, c.patswnakis, ad, aggelos, mema\}@di.uoa.gr; 
B. Zhang,
School of Computing and Communications, Lancaster University, InfoLab21, Bailrigg, Lancaster LA1 4WA, UK,
email: b.zhang2@lancaster.ac.uk;
S. Maneas, Department of Computer Science, University of Toronto, 40 St. George Street, Toronto, ON, M5S2E4, Canada,
email: smaneas@cs.toronto.edu;
Aggelos Kiayias,
School of Informatics, University of Edinburgh, Office 5.16, 10 Crichton St., Edinburgh EH8 9AB, UK,
email: Aggelos.Kiayias@ed.ac.uk.
\end{bottomstuff}

\maketitle

\input{intro}

\input{background}
\input{related}
\input{system}
\input{attacks}
\input{evaluation_ext}

\input{conclusion}
\input{acknowledgement}

\appendix
\input{security}

\bibliographystyle{ACM-Reference-Format-Journals}
\bibliography{MyBibFile}

\end{document}

%% file: abstract.tex
E-voting systems are a powerful technology for improving democracy by reducing election cost, increasing voter participation, and even allowing voters to directly verify the entire election procedure.
Unfortunately, prior internet voting systems have single points of failure, which may result in the compromise of availability, voter secrecy, or integrity of the election results.

In this paper, we present the design, implementation, security analysis, and evaluation of the D-DEMOS suite of distributed, privacy-preserving, and end-to-end verifiable e-voting systems.
We present two systems: one completely asynchronous and one with minimal timing assumptions but better performance.
Our systems include a distributed vote collection subsystem that provides immediate assurance to the voter her vote was recorded as cast, without requiring cryptographic operations on behalf of the voter. 
We also include a distributed, replicated and fault-tolerant Bulletin Board component, that stores all necessary election-related information, and allows any party to read and verify the complete election process.
Finally, we also incorporate trustees, i.e., individuals who control election result production while guaranteeing privacy and end-to-end-verifiability as long as their strong majority is honest.
 
Our suite of e-voting systems are the first whose voting operation is human verifiable, i.e., a voter can vote over the web, even when her web client stack is potentially unsafe, without sacrificing her privacy, and still be assured her vote was recorded as cast.
Additionally, a voter can outsource election auditing to third parties, still without sacrificing privacy. 
Finally, as the number of auditors increases, the probability of election fraud going undetected is diminished exponentially.

We provide a model and security analysis of the systems.
We implement prototypes of the complete systems, we measure their performance experimentally, and we demonstrate their ability to handle large-scale elections. 
Finally, we demonstrate the performance trade-offs between the two versions of the system.
A preliminary version of our system was used to conduct exit-polls at three voting sites for two national-level elections and is being adopted for use by the largest civil union of workers in Greece, consisting of over a half million members.

%% file: intro.tex
\section{Introduction}
E-voting systems are a powerful technology to improve the election process. 
Kiosk-based e-voting systems, e.g., \cite{chaum2001surevote,chaum-esorics-2005,fisher-wote-2006,chaum2008scantegrity,benaloh2013starvote,culnane2014peered}, allow the tally to be produced faster, but  require the voter's physical presence at the booth. 
Internet e-voting systems, e.g., \cite{CGS-eurocrypt-1997,adida-helios-2008,clarkson2008civitas,kutylowski2010SCV,gjosteen2013norway,zagorski2013remotegrity,chaum2001surevote,chaum2008scantegrity,zagorski2013remotegrity,DEMOS}, however, allow voters to cast their votes remotely. 
Internet voting systems have the potential to enhance the democratic process by reducing election costs and by increasing voter participation for social groups that face considerable physical barriers and overseas voters. 
In addition, several internet voting systems~\cite{adida-helios-2008,kutylowski2010SCV,zagorski2013remotegrity,DEMOS} allow voters and auditors to directly verify the integrity of the entire election process, providing \emph{end-to-end verifiability}.
This is a highly desired property that has emerged in the last decade, where voters can be assured that no entities, even the election authorities, have manipulated  the election result.
Despite their potential, existing internet voting systems suffer from single points of failure, which may result in the compromise of voter secrecy, service availability, or integrity of the result~\cite{chaum2001surevote,chaum-esorics-2005,fisher-wote-2006,chaum2008scantegrity,benaloh2013starvote,CGS-eurocrypt-1997,adida-helios-2008,clarkson2008civitas,kutylowski2010SCV,gjosteen2013norway,zagorski2013remotegrity,DEMOS}.

In this paper, we present the design and prototype implementation of the \emph{D-DEMOS} suite of distributed, end-to-end verifiable internet voting systems, with no single point of failure during the election process (that is, besides setup). 
We set out to overcome two major limitations in existing internet voting systems. 
The first, is their dependency on centralized components. 
The second is their requirement for the voter to run special software on their devices, which processes cryptographic operations. 
Overcoming the latter allows votes to be cast with a greater variety of client devices, such as feature phones using SMS, or untrusted public web terminals.
Our design is inspired by the novel approach proposed in~\cite{DEMOS}, where the voters are used as a source of randomness to challenge the zero-knowledge proof protocols~\cite{feige1988zero}. 
We use the latter to enable end-to-end verifiability.

We design a distributed \emph{Vote Collection} (\VC{}) subsystem that is able to collect votes from voters and assure them their vote was recorded as cast, without requiring any cryptographic operation from the client device. 
This allows voters to vote via SMS, a simple console client over a telnet session, or a public web terminal, while preserving their privacy.
At election end time, \VC{} nodes agree on a single set of votes.
We introduce two versions of D-DEMOS that differ in how they achieve agreement on the set of cast votes.
The D-DEMOS/Async version is completely asynchronous, while D-DEMOS/IC makes minimal synchrony assumptions but is more efficient than the alternative.
Once agreement has been achieved, \VC{} nodes upload the set of cast votes to a second distributed component, the \emph{Bulletin Board} (\BB{}). 
This is a replicated service that publishes its data immediately and makes it available to the public forever.
Finally, our \trustees{} subsystem, comprises a set of persons entrusted with secret keys which can unlock information stored in the \BB{}. 
We share these secret keys among the \trustees{}, making sure only an honest majority can uncover information from the \BB{}.
\Trustees{} interact with the \BB{} once the votes are uploaded to the latter, to produce and publish the final election tally.

The resulting voting systems are end-to-end verifiable, by the voters themselves and third-party auditors, while preserving voter privacy. 
To delegate auditing, a voter provides an auditor specific information from her ballot. 
The auditor, in turn, reads from the distributed \BB{} and verifies the complete election process, including the correctness of the election setup by election authorities. 
Additionally, as the number of auditors increases, the probability of election fraud going undetected diminishes exponentially.

Finally, we implement prototypes of both D-DEMOS voting system versions. 
We measure their performance experimentally, under a variety of election settings, demonstrating their ability to handle thousands of concurrent connections, and thus manage large-scale elections.
We also compare the two systems and emphasize the trade-offs between them, regarding security and performance.

To summarize, we make the following contributions:
\begin{itemize}
 \item We present the world's first suite of state-of-the-art, end-to-end verifiable, distributed voting systems with no single point of failure besides setup. 
 \item Both systems allow voters to verify their vote was tallied-as-intended without the assistance of special software or trusted devices, and allow external auditors to verify the correctness of the election process. 
 Additionally, both systems allow voters to delegate auditing to a third party auditor, without sacrificing their privacy.
 \item We provide a model and a security analysis of D-DEMOS/IC.
\ifextended
\else
    Due to lack of space, we omit the corresponding sections of D-DEMOS/Async, and refer the interested reader to the extended version of this paper~\cite{extended}.
\fi
 \item We implement prototypes of the systems, measure their performance and demonstrate their ability to handle large-scale elections. Finally, we demonstrate the performance trade-offs between the two versions of the system.
\end{itemize}
Note that, a preliminary version of one of our systems was used to conduct exit-polls at three voting sites for two national-level elections and is being adopted for use by the largest civil union of workers in Greece, consisting of over a half million members.

\ifextended
The remainder of this paper is organized as follows. 
Section~\ref{section:background} introduces required background knowledge we reference throughout the paper, while Section~\ref{section:related} presents related work. 
Section~\ref{section:system} gives an overview of the system components, defines the system and threat model, and describes each system component in detail.
Section~\ref{section:attacks} goes over some interesting attack vectors, which help to clarify our design choices.
Section~\ref{section:impl_and_eval} describes our prototype implementations and their evaluation, and Section~\ref{section:conclusion} concludes the main body of the paper.
Finally, Appendix~\ref{sec:security_full} provides, for the interested reader, the full proofs of liveness, safety, privacy and end-to-end verifiability of both our systems.
\fi

%% file: background.tex
\section{Background}\label{section:background}
\ifextended
In this section we provide basic background knowledge required to comprehend the system description in the next
section. This includes some voting systems terminology, a quick overview of Interactive Consistency, and a series of cryptographic tools we use to design our systems.
These tools include additively homomorphic commitment schemes and zero-knowledge proofs, which are used in the System Description (Section~\ref{section:system}), and are needed to understand the system design. 
Additionally, we provide details about collision resistant hash functions, IND-CPA symmetric encryption schemes, and digital signatures, which we use as building blocks for our security proofs in Appendix~\ref{sec:security_full}.
\fi

\ifextended
\subsection{Voting Systems requirements}\label{section:voteback}
An ideal electronic voting system would address a specific list of requirements (see \cite{neumann-ncsc-1993,IPI-2001} for an extensive description).
Our system addresses the following requirements:
\begin{itemize}
\item\textbf{End-to-end verifiability:} the voters can verify that their votes were counted as they intended and any party can verify that the election procedure was executed correctly.
\item\textbf{Privacy}: a party that does not monitor voters during the voting phase of the election, cannot extract information about the voters' ballots. In addition, a voter cannot prove how she voted to any party that did not monitor her during the voting phase of the election\footnote{In~\cite{DEMOS}, this property is referred as \emph{receipt-freeness}.}.
\item\textbf{Fault tolerance:} the voting system should be resilient to the faulty behavior of up to a number of components or parts, and be both live and safe.
\end{itemize}
\fi

\subsection{Interactive Consistency}
\label{section:ic}
Interactive consistency (IC), first introduced and studied by Pease et al.~\cite{lamport.ic},
is the problem in which $n$ nodes, where up to $t$ may be byzantine, each with its own private value, run an algorithm that allows all non-faulty nodes to infer the values of each other.
In our D-DEMOS/IC system, we use the \emph{IC,BC-RBB} algorithm from~\cite{IC@ICPADS2015}, which achieves IC using a single synchronous round. 
This algorithm uses two phases to complete. 
The synchronous \emph{Value Dissemination Phase} comes first, aiming to disperse the values across nodes.
Consequently, an asynchronous \emph{Result Consensus Phase} starts, which results in each honest node holding a vector with every honest node's slot filled with the corresponding value.

\subsection{Cryptographic tools}\label{subsec:sec_tools}

\subsubsection{Additively homomorphic commitments}\label{subsubsec:commit}
To achieve integrity against a malicious election authority, our D-DEMOS utilizes lifted ElGamal~\cite{elgamal-crypto-1985} over elliptic curves as
a \emph{non-interactive commitment scheme} that achives the following properties:
\begin{enumerate}
 \item \emph{Perfectly binding}: no adversary can open a commitment $\Com(m)$ of $m$ to a value other than $m$.
 \item \emph{Hiding}: there exists a constant $c<1$ s.t. the probability that a commitment $\Com(m)$ to $m$ leaks information about $m$ to an adversary running in $O(2^{\lambda^c})$ steps is no more than $\mathsf{negl}(\lambda)$.
 \item \emph{Additively homomorphic}: $\forall m_1,m_2$, we have that $\Com(m_1) \cdot \Com(m_2) = \Com(m_1+m_2)\enspace.$
\end{enumerate}

\subsubsection{Zero-knowledge Proofs}\label{subsubsec:zkback}
D-DEMOS's security requires the election authority to show the correctness of the election setup to the public without compromising privacy. We enable this kind of verification with the use of zero-knowledge proofs. 
In a zero-knowledge proof, the prover is trying to convince the verifier that a statement is true, without revealing any information about the statement apart from the fact that it is true~\cite{quisquater1990explain}. More specifically, we say an interactive proof system has the \emph{honest-verifier zero-knowledge (HVZK)} property if there exists a probabilistic polynomial time simulator $\mathcal{S}$ that , for any given challenge, can output an accepting proof transcript that is distributed indistinguishable to the real transcript between an honest prover and an honest verifier. Here, we adopt  Chaum-Pedersen  zero-knowledge proofs~\cite{CP}, which belong in the special class of $\Sigma$ protocols (i.e., 3-move public-coin special HVZK proofs), allowing the Election Authority to show that the content inside each commitment is a valid option encoding. 

\subsubsection{Collision resistant hash functions}\label{subsubsec:crhf}
 Given the security parameter $\lambda\in\mathbb{N}$, we say that a hash function $h:\{0,1\}^{*}\mapsto \{0,1\}^{\ell(\lambda)}$, where $\ell(\lambda)$  is polynomial in $\lambda$, is $(t,\epsilon)$-\emph{collision resistant} if for every adversary $\A$ running in time at most $t$, the probability of $\A$ finding two distinct preimages $m_1\neq m_2$ such that $h(m_1) = h(m_2)$ is less than $\epsilon$. By the birthday attack, in order for $h$ to be $(t,\epsilon)$-collision resistant, we necessitate that $t^2/2^{\ell(\lambda)}<\epsilon$. In this work, we use SHA-256 as the instantiation of a $(t,t^2\cdot2^{-256})$-collision resistant hash function.
 %
\subsubsection{IND-CPA symmetric encryption schemes}\label{subsubsec:aes}
We say that a symmetric encryption scheme $\mc{SE}$ achieves $(t,q,\epsilon)$-\emph{indistinguishability against chosen plaintext attacks (IND-CPA)}, if for every adversary $\A$ that (i) runs in time at most $t$, (ii) makes at most $q$ encryption queries that are pairs of messages $(m_{0,1},m_{1,1}),\ldots,(m_{0,q},m_{1,q})$ and (iii) for every encryption query $(m_{0,i},m_{1,i})$, it receives the encryption of $m_{b,i}$, where $b$ is the outcome of a coin-flip, it holds that
\begin{equation*}
 \begin{split}
  &\mathbf{Adv}_\mc{SE}^{\mathsf{IND-CPA}}(\A):=\big|\Pr[\A\mbox{ outputs }1\mid b=1]-Pr[\A\mbox{ outputs }1\mid b=0]\big|<\epsilon\hspace{2pt},
 \end{split}
\end{equation*}
where by $\mathbf{Adv}_\mc{SE}^{\mathsf{IND-CPA}}(\A)$ we denote the \emph{advantage} of $\A$. 
D-DEMOS applies AES-128-CBC\$ encryption, for which a known
safe conjecture is that $\mathbf{Adv}_\mathsf{AES-128}^{\mathsf{PRF}}(\mathcal{B})\leq (t+129\cdot q+q^2)\cdot2^{-128}$,
so in our proofs we assume that AES-128-CBC\$ is $(t,q,\allowbreak(2t+258\cdot q+3q^2)\cdot2^{-128})$-IND-CPA secure. For further details, we refer the reader to~\cite[Chapters 3 \& 4]{bellare-rogaway-notes}.
\ifextended
\subsubsection{Digital Signature Schemes}\label{subsubsec:ds}
A digital signature system is said to be secure if it is \emph{existentially unforgeable
under a chosen-message attack (EUF-CMA)}. Roughly speaking, this means that an adversary running in  polynomial time and adaptively querying signatures for (polynomially many) messages has no more than $\mathsf{negl}(\lambda)$ probability to forge a valid signature for a new message. D-DEMOS/Async utilizes the standard the RSA signature scheme, which is EUF-CMA secure under the factoring assumption.
\fi

%% file: related.tex
\section{Related work}\label{section:related}
\subsection{Voting systems}
Several end-to-end verifiable e-voting systems have been introduced, e.g. the kiosk-based systems~\cite{chaum-esorics-2005,fisher-wote-2006,chaum2008scantegrity,benaloh2013starvote,Moran:2010:SVE:1698750.1698756} and the internet voting systems~\cite{adida-helios-2008,kutylowski2010SCV,zagorski2013remotegrity,DEMOS}. 
In all these works, the Bulletin Board (\BB{}) is a single point of failure and has to be trusted.

Dini presents a distributed e-voting system, which however is not end-to-end verifiable~\cite{dini2003secure}.  
In~\cite{culnane2014peered}, there is a distributed \BB{} implementation, also handling vote collection, according to the design of the vVote end-to-end verifiable e-voting system~\cite{culnane2014vVote}, which in turn is an adaptation of the Pr{\^{e}}t {\`{a}} Voter e-voting system~\cite{chaum-esorics-2005}.
In~\cite{culnane2014peered}, the proper operation of the \BB{} during ballot casting requires a trusted device for signature verification. 
In contrast, our vote collection subsystem is done so that correct execution of ballot casting can be ``human verifiable'', i.e., by simply checking the validity of the obtained receipt. 
Additionally, our vote collection subsystem in D-DEMOS/Async is fully asynchronous, always deciding with exactly $n-f$ inputs, while in~\cite{culnane2014peered}, the system uses a synchronous approach based on the FloodSet algorithm from~\cite{Lynch:1996:DA} to agree on a single version of the state.

DEMOS~\cite{DEMOS} is an end-to-end verifiable e-voting system, which introduces the novel idea 
of extracting the challenge of the zero-knowledge proof protocols from the voters' random choices; we leverage this idea in our system too.
However, DEMOS uses a centralized Election Authority (EA), which maintains all secrets throughout the entire election procedure, collects votes, produces the result and commits to verification data in the \BB{}.
Hence, the EA is a single point of failure, and because it knows the voters' votes, it is also a critical privacy vulnerability.
In this work, we address these issues by introducing distributed components for vote collection and result tabulation, and we do not assume any trusted component during election.  
Additionally, DEMOS does not provide any recorded-as-cast feedback to the voter, whereas our system includes such a mechanism.

Besides, DEMOS encodes the $i$-th option to $N^{i-1}$, where $N$ is greater than the total number of voters, and this option encoding has to fit in the message space of commitments. 
Therefore, the size of the underlying elliptic curve grows linearly with the number of options, which makes DEMOS not scalable with respect to the number of options. 
In this work, we overcome this problem by using a different scheme for option encoding commitments. 
Moreover, the zero-knowledge proofs in DEMOS have a big soundness error, and it decreases the effectiveness of zero-knowledge application; whereas, in our work, we obtain nearly optimal overall zero-knowledge soundness.

Furthermore, none of the above works provide any performance evaluation results.
Finally,~\cite{Appel:2011:SSV:2019599.2019603} outlines the difficulties in managing seals for kiosks and ballot boxes, supporting our position towards the use of internet voting. 

\subsection{State Machine Replication}
Castro et al.~\cite{castro2002practical} introduce a practical Byzantine Fault Tolerant replicated state machine protocol. 
In the last several years, several protocols for Byzantine Fault Tolerant state machine replication have been introduced to improve performance (\cite{cowling2006hq,kotla2007zyzzyva}), robustness (\cite{aublin2013rbft,clement2009making}), or both (\cite{clement2009upright,next700tocs}). 
Our system does not use the state machine replication approach to handle vote collection, as it would be inevitably more costly. 
Each of our vote collection nodes can validate a voter's requests on its own. 
In addition, we are able to process multiple different voters' requests concurrently, without enforcing the total ordering inherent in replicated state machines. 
Finally, we do not wish voters to use special client-side software to access our system.

%% file: system.tex
\section{System description}\label{section:system}
\subsection{Problem Definition and Goals}
We consider an \emph{election} with a single \emph{question} and $m$ \emph{options}, for a voter population of size $n$, where voting takes place between a certain \emph{begin} and \emph{end} time (the \emph{voting hours}), and each voter may select a single \emph{option}.

Our major goals in designing our voting system are three. 
\begin{inparaenum}[1)]
\item It has to be end-to-end verifiable, so that anyone can verify the complete election process. 
Additionally, voters should be able to outsource auditing to third parties, without revealing their voting choice.
\item It has to be fault-tolerant, so that an attack on system availability and correctness is hard. 
\item Voters should not have to trust the terminals they use to vote, as such devices may be malicious. Instead, voters should be assured their vote was recorded, without disclosing any information on how they voted to the malicious entity controlling their device.
\end{inparaenum}

%
%

\subsection{System overview}
\label{sec:sysoverview}
We employ an election setup component in our system, which we call the Election Authority (\EA{}), to alleviate the voter from employing any cryptographic operations. 
The \EA{} initializes all other system components, and then gets immediately destroyed to preserve privacy.
The \emph{Vote Collection} (\VC{}) subsystem collects the votes from the voters during election hours, and assures them their vote was \emph{recorded-as-cast}.
Our \emph{Bulletin Board} (\BB{}) subsystem, which is a public repository of all election-related information, is used to hold all ballots, votes, and the result, either in encrypted or plain form, allowing any party to read from the \BB{} and verify the complete election process.
The \VC{} subsystem uploads all votes to the \BB{} at election end time.
Finally, our design includes \trustees{}, who are persons entrusted with managing all actions needed until result tabulation and publication, including all actions supporting end-to-end verifiability.
\Trustees{} hold the keys to uncover any information hidden in the \BB{}, and we use threshold cryptography to make sure a malicious minority cannot uncover any secrets or corrupt the process.

\ifextended
\else
We outline the interactions between these subsystems and the actors in Figure~\ref{figure:ddemos}. In the following paragraphs, we explain these interactions in more detail.

\begin{figure}[ht!]
  \centering
  \includegraphics[width=0.95\textwidth]{ddemos}
  \caption{High-level diagram of interactions between subsystems and actors. Subsystems are distributed systems of their own, but are depicted as a unified entity in this diagram.}
  \label{figure:ddemos}
\end{figure}
\fi



Our system starts with the \EA{} generating initialization data for every component of our system. 
The \EA{} encodes each election option, and \emph{commits} to it using a commitment scheme, as described below.
It encodes the $i$-th option as $\vec{e}_i$, a unit vector where the $i$-th element is $1$ and the remaining elements are $0$. 
The commitment of an option encoding is a vector of (lifted) ElGamal ciphertexts~\cite{ElGamal} over elliptic curve, that element-wise encrypts a unit vector. Note that this commitment scheme is also additively homomorphic, i.e., the commitment of $e_a + e_b$ can be computed by component-wise multiplying the corresponding commitments of $e_a$ and $e_b$.
The \EA{} then creates a $\mathsf{vote code}$ and a $\mathsf{receipt}$ for each option.
Subsequently, the \EA{} prepares one ballot for each voter, with two functionally equivalent parts. 
Each part contains a list of options, along with their corresponding vote codes and receipts.
We consider ballot distribution to be outside the scope of this paper, but we do assume ballots, after being produced by the \EA{}, are distributed in a secure manner to each voter; thus only each voter knows the vote codes listed in her ballot.
We make sure vote codes are not stored in clear form anywhere besides the voter's ballot.
\ifextended
We depict this interaction in Figure~\ref{figure:ddemos-ea}.

\begin{figure}
  \centering
  \includegraphics[width=0.95\textwidth]{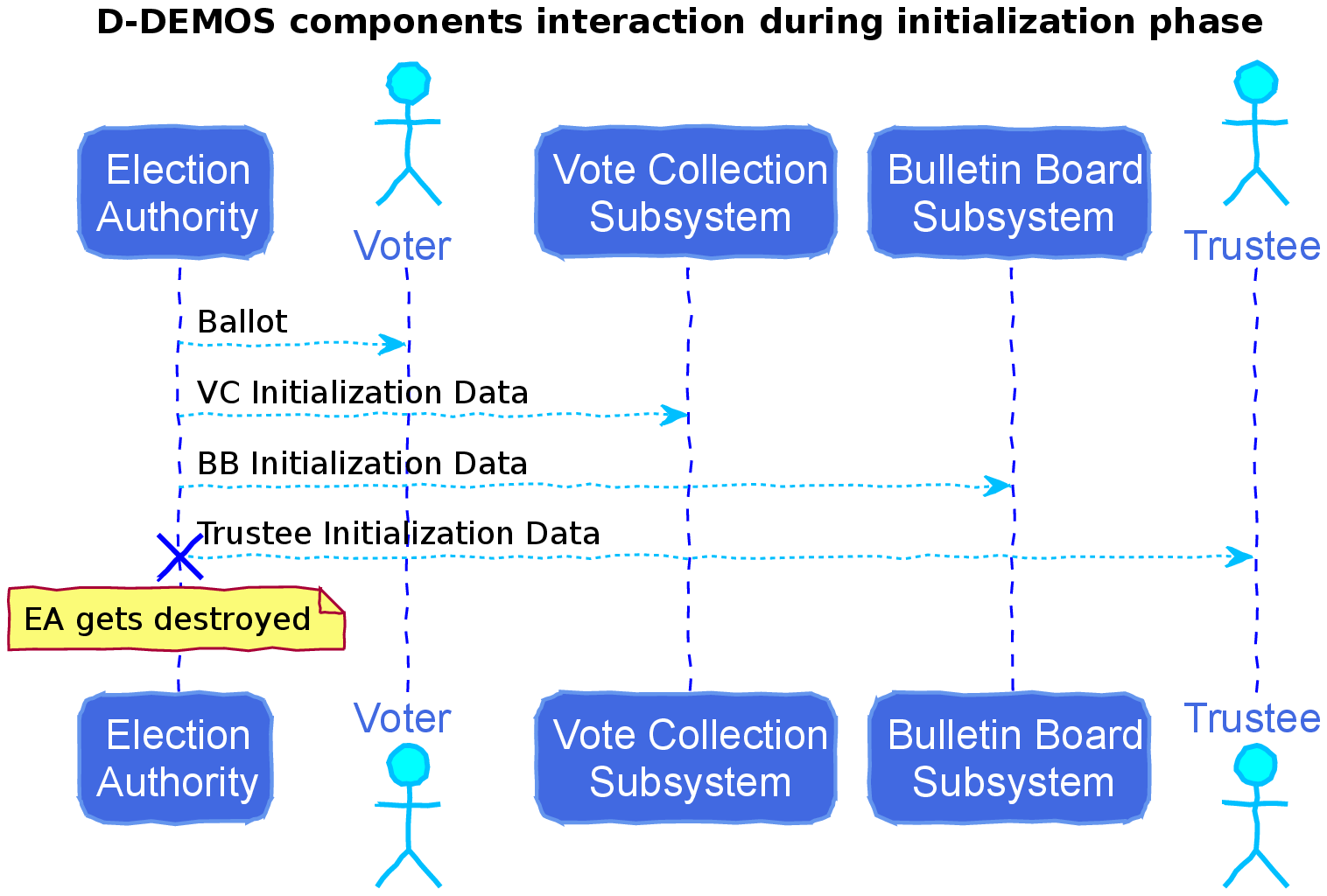}
  \caption{High-level diagram of component interactions during system initialization. Each subsystem is a distributed system of its own, but is depicted as a unified entity in this diagram for brevity.}
  \label{figure:ddemos-ea}
\end{figure}

\begin{figure}
  \centering
  \includegraphics[width=0.95\textwidth]{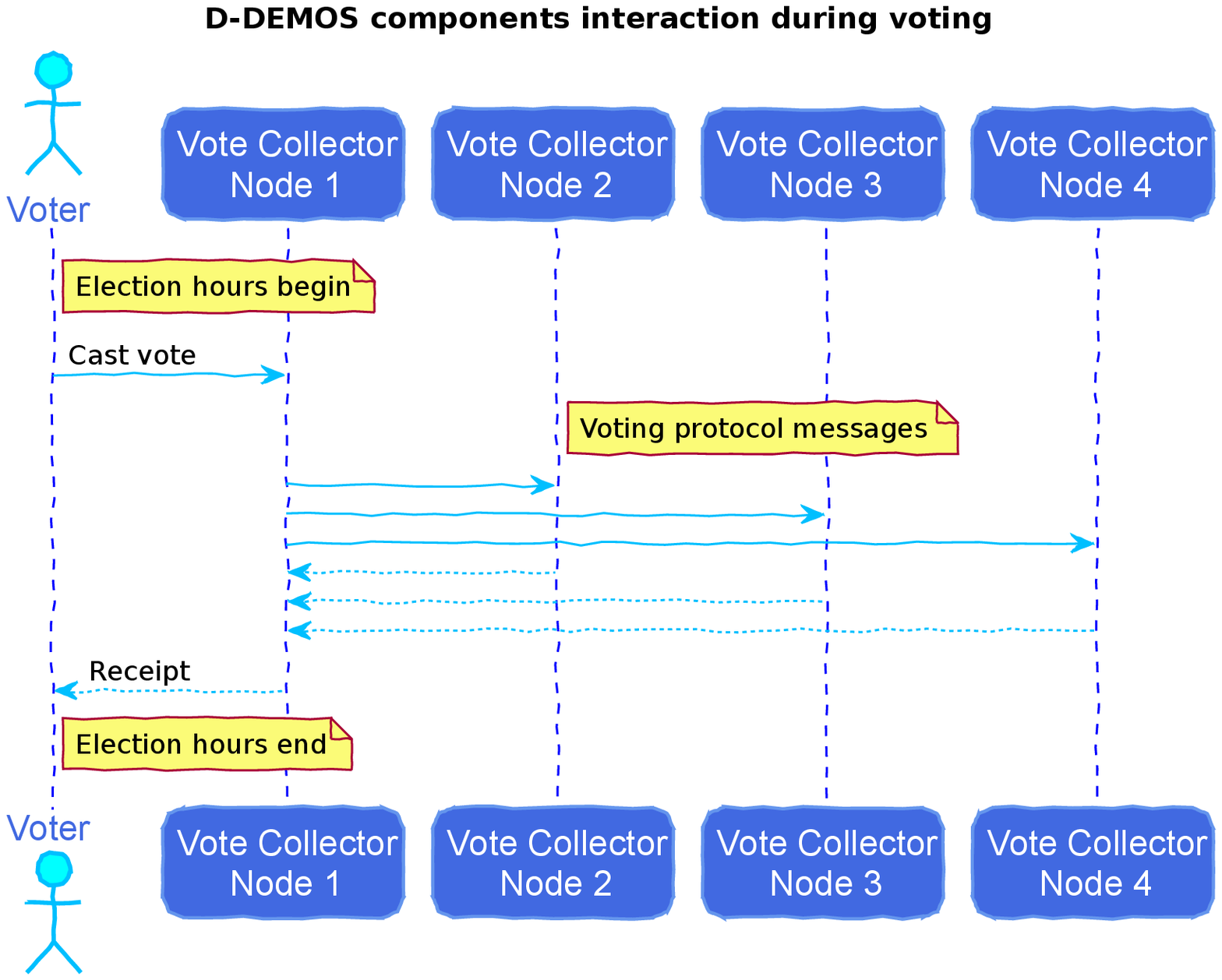}
  \caption{High-level diagram of component interactions during the voting phase. Message exchanges between \VC{} nodes are simplified for this diagram. In this diagram, there are $4$ \VC{} nodes, tolerating up to $1$ fault.}
  \label{figure:ddemos-voting}
\end{figure}
\fi

Our \VC{} subsystem collects the votes from the voters during election hours, by accepting up to one vote code from each voter
\ifextended
(see Figure~\ref{figure:ddemos-voting}). 
\else
.
\fi
The \EA{} initializes each \VC{} node with the vote codes and the receipts of the voters' ballots. 
However, it hides the vote codes, using a simple commitment scheme based on symmetric encryption of the plaintext along with a random salt value.
This way, each \VC{} node can verify if a vote code is indeed part of a specific ballot, but cannot recover any vote code until the voter actually chooses to disclose it. 
Additionally, we secret-share each receipt across all \VC{} nodes using an $(N-f, N)$-VSS (verifiable secret-sharing) scheme with trusted dealer~\cite{schneier1996applied}, making sure that a receipt can be recovered and posted back to the voter only when a strong majority of \VC{} nodes participates successfully in our voting protocol.
With this design, our system adheres to the following contract with the voters: \emph{Any honest voter who receives a valid receipt from a Vote Collector node, is assured her vote will be published on the \BB{}, and thus it will be included in the election tally}.

The voter selects one part of her ballot at random, and posts her selected vote code to one of the \VC{} nodes. 
When she receives a receipt, she compares it with the one on her ballot corresponding to the selected vote code.
If it matches, she is assured her vote was correctly recorded and will be included in the election tally.
The other part of her ballot, the one not used for voting, will be used for auditing purposes.
This design is essential for verifiability, in the sense that the \EA{} cannot predict which part a voter may use, and the unused part will betray a malicious \EA{} with $1 \over 2$ probability per audited ballot.

Our second distributed subsystem is the \BB{}, which is a replicated service of isolated nodes.
Each \BB{} node is initialized from the \EA{} with vote codes and associated option encodings in committed form (again, for vote code secrecy), and each \BB{} node provides public access to its stored information.
At election end time, \VC{} nodes run our Vote Set Consensus protocol, which guarantees all \VC{} nodes agree on a single set of voted vote codes.
After agreement, each \VC{} node uploads this set to every \BB{} node, which in turn publishes this set once it receives the same copy from enough \VC{} nodes
\ifextended
(see Figure~\ref{figure:ddemos-vcs}).

\begin{figure}
  \centering
  \includegraphics[width=0.95\textwidth]{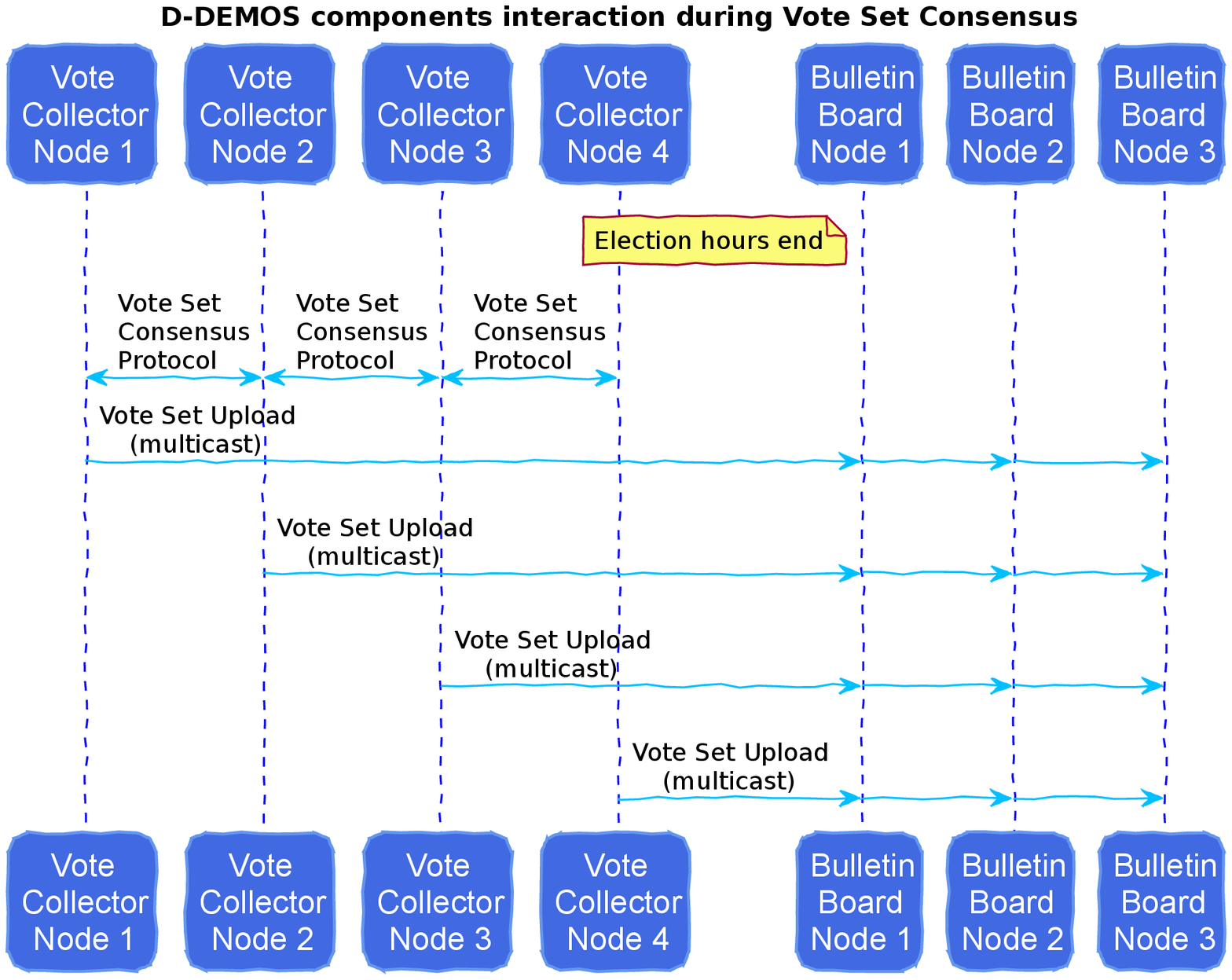}
  \caption{High-level diagram of component interactions during the vote set consensus phase. $4$ \VC{} nodes and $3$ \BB{} nodes are shown, where each subsystem tolerates $1$ fault. ``VSC'' stands for ``Vote Set Consensus''. After agreeing on a single Vote Set $S$, each \VC{} node uploads $S$ to every \BB{} node. Messages are simplified for this diagram. }
  \label{figure:ddemos-vcs}
\end{figure}
\else
.
\fi

Our third distributed subsystem is a set of \trustees{}, who are persons entrusted with managing all actions needed after vote collection, until result tabulation and publication; this includes all actions supporting end-to-end verifiability.
Secrets that may uncover information in the \BB{} are shared across \trustees{}, making sure malicious \trustees{} under a certain threshold cannot uncover and disclose sensitive information.
We use Pedersen's Verifiable linear Secret Sharing (VSS)~\cite{vss} to split the election data among the \trustees{}. 
In a $(k,n)$-VSS, at least $k$ shares are required to reconstruct the original data, and any collection of less than $k$ shares leaks no information about the original data. 
Moreover, Pedersen's VSS is  additively homomorphic, i.e., one can compute the share of $a+b$ by adding the share of $a$ and the share of $b$ respectively.
This approach allows \trustees{} to perform homomorphic ``addition'' on the option-encodings of cast vote codes, and contribute back a share of the opening of the homomorphic ``total''. 
Once enough \trustees{} upload their shares of the ``total'', the election tally is uncovered and published at each \BB{} node
\ifextended
(see Figure~\ref{figure:ddemos-trustees}).

\begin{figure}
  \centering
  \includegraphics[height=8cm]{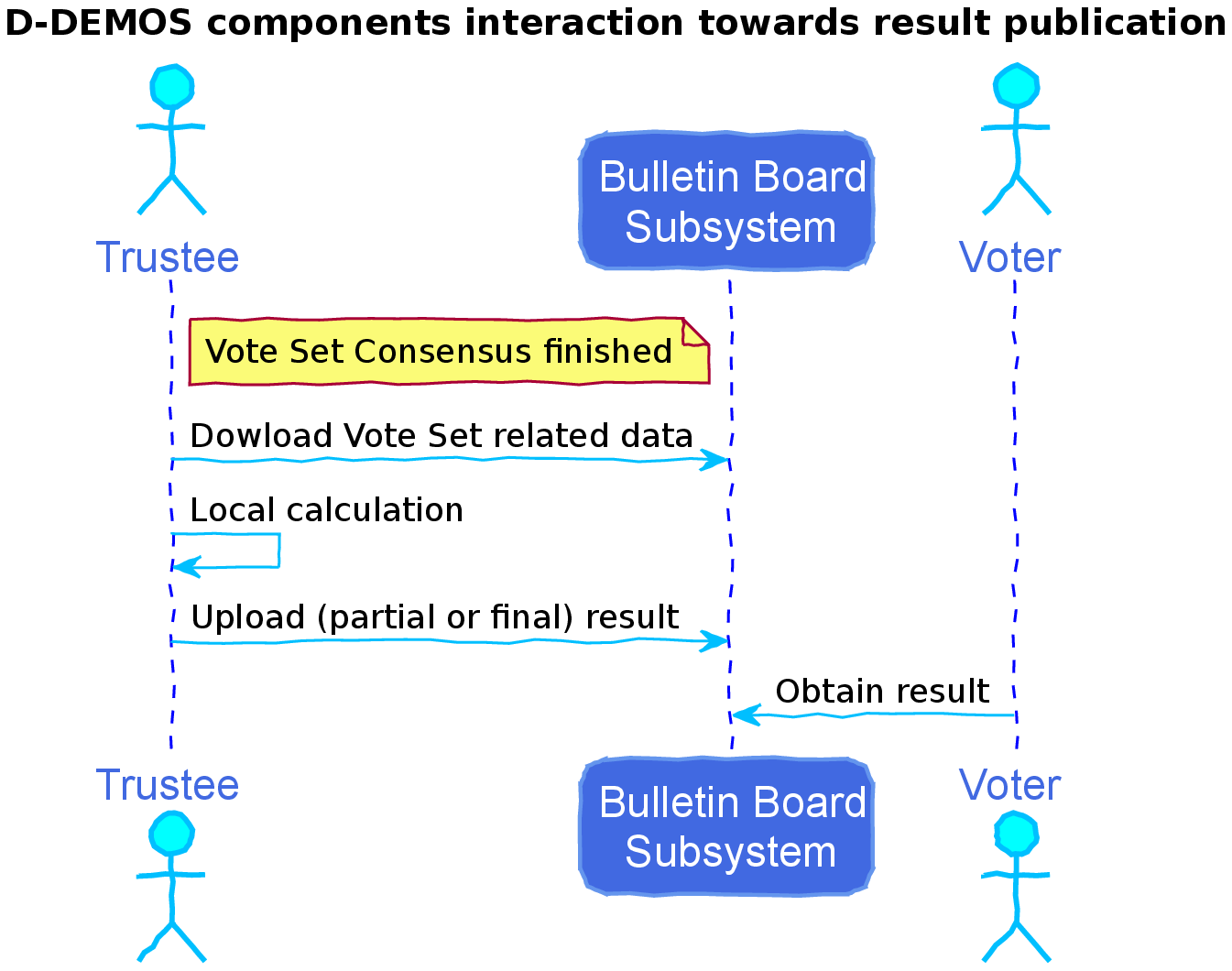}
  \caption{High-level diagram of \trustee{} interactions with the \BB{}, towards result tabulation and publication. \Trustees{} are more than one, and interact with the \BB{} in any order. The \BB{} is a distributed system of its own, but is depicted as a unified entity in this diagram for brevity.}
  \label{figure:ddemos-trustees}
\end{figure}
\else
.
\fi

To ensure voter privacy, the system cannot reveal the content inside an option encoding commitment at any point. 
However, a malicious \EA{} might put an arbitrary value (say $9000$ votes for option $1$) inside such a commitment, causing an incorrect tally result. 
To prevent this, we utilize the Chaum-Pedersen  zero-knowledge proof~\cite{CP}, allowing the \EA{} to show that the content inside each commitment is a valid option encoding, without revealing its actual content.
Namely, the prover uses Sigma OR proof to show that each ElGamal ciphertext encrypts either $0$ or $1$, and the sum of all elements in a vector is $1$.
Our zero knowledge proof is organized as follows.
First, the \EA{} posts the initial part of the proofs on the \BB{}.
Second, during the election, each voter's A/B part choice is viewed as a source of randomness, $0/1$, and all the voters' choices are collected and used as the challenge of our zero knowledge proof.
Finally, the \trustees{} will jointly produce the final part of the proofs and post it on the \BB{} before the opening of the tally. 
Hence, everyone can verify those proofs on the \BB{}. 
We omit the zero-knowledge proof components in this paper and refer the interested reader to~\cite{CP} for details.

\ifextended
\begin{figure}
  \centering
  \includegraphics[height=7cm]{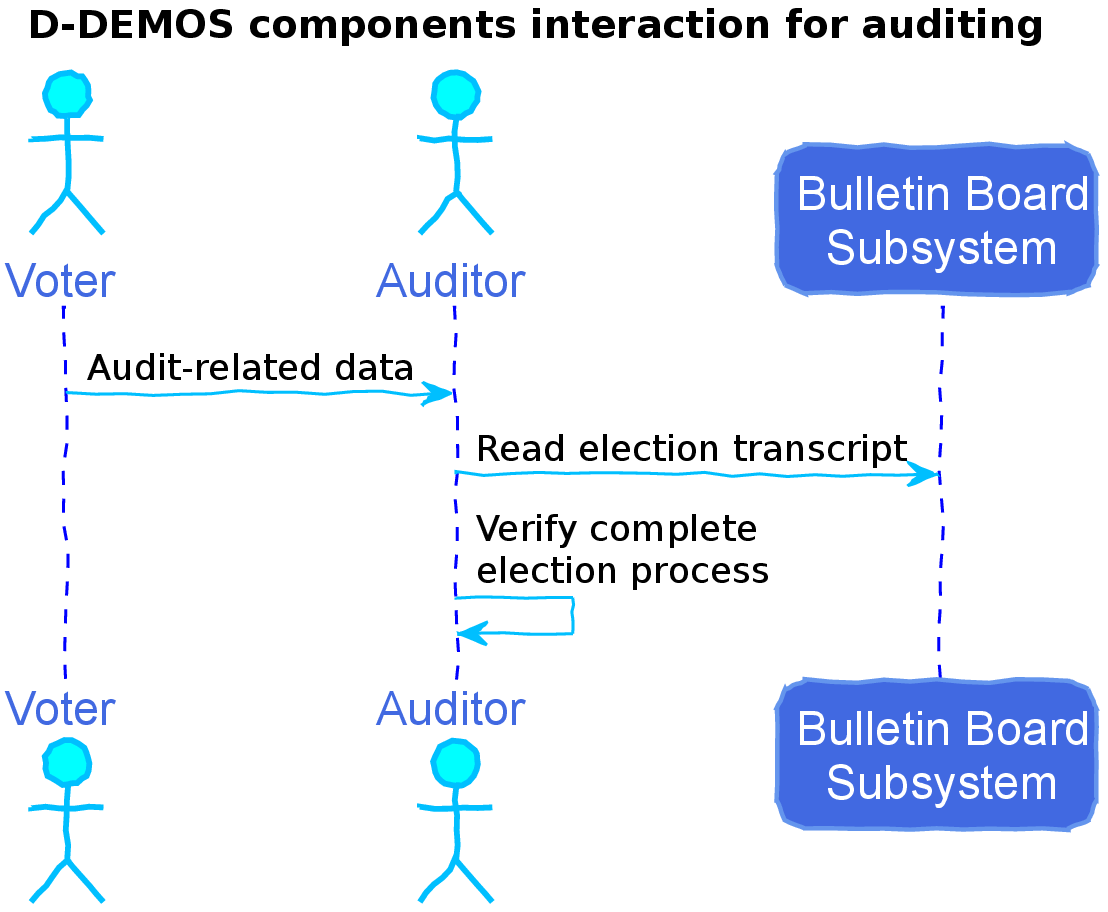}
  \caption{High-level diagram of the system auditing. Voters send Auditors audit-related data that does not violate the voter's privacy. Auditors in turn read from the \BB{} and verify the complete election process. The \BB{} is a distributed system of its own, but is depicted as a unified entity in this diagram for brevity.}
  \label{figure:ddemos-audit}
\end{figure}
\fi

Our design allows any voter to read information from the \BB{}, combine it with her private ballot, and verify her ballot was included in the tally.
Additionally, any third-party auditor can read the \BB{} and verify the complete election process
\ifextended
(see Figure~\ref{figure:ddemos-audit}).
\else
.
\fi
As the number of auditors increases, the probability of election fraud going undetected diminishes exponentially.
For example, even if only $10$ people audit, with each one having $1 \over 2$ probability of detecting ballot fraud, the probability of ballot fraud going undetected is only ${1 \over 2}^{10} = 0.00097$.
Thus, even if the \EA{} is malicious and, e.g.,  tries to point all vote codes to a specific option, this faulty setup will be detected because of the end-to-end verifiability of the complete system.

In this paper, we present two different versions of our voting system, with different performance and security trade-offs.
In the first version, called \emph{D-DEMOS/IC}, Vote Set Consensus is realized by an algorithm achieving Interactive Consistency, and thus requiring synchronization. 
The second version, \emph{D-DEMOS/Async}, uses an asynchronous binary consensus algorithm for Vote Set Consensus, and thus is completely asynchronous. The performance trade-offs between the two are analyzed in Section~\ref{sect:evaluation}.
\input{threat_model}
%
\subsection{Election Authority}
\label{sect:EA}
\EA{} produces the initialization data for each election entity in the setup phase. 
To enhance the system robustness, we let the \EA{} generate all the public/private key pairs for all the system components (except voters) without relying on external PKI support.
We use zero knowledge proofs to ensure the correctness of all the initialization data produced by the \EA{}. 

\subsubsection{Voter Ballots} 
The \EA{} generates one ballot $\mathsf{ballot}_\ell$ for each voter $\ell$, and assigns a unique $64$-bit $\mathsf{serial\textrm{-}no}_\ell$ to it.
As shown below, each ballot consists of two parts: Part A and Part B. 
Each part contains a list of $m$ $\langle\mathsf{vote\textrm{-}code},\mathsf{option},\mathsf{receipt}\rangle$ tuples, one tuple for each election option.
The \EA{} generates the vote-code as a $128$-bit random number, unique within the ballot, and the receipt as $64$-bit random number.

\begin{center}
\ifextended
\else
\begin{scriptsize}
\fi
	\begin{tabular}{|l l l |}
		\hline
		$\mathsf{serial\textrm{-}no}_\ell$ &  &\\
		\hline
		& Part A &\\
		$\mathsf{vote\textrm{-}code}_{\ell,1}$ & $\mathsf{option}_{\ell,1}$ & $\mathsf{receipt}_{\ell,1}$\\
		$\quad\hdots$ & $\quad\hdots$ & $\quad\hdots$\\
		$\mathsf{vote\textrm{-}code}_{\ell,m}$ & $\mathsf{option}_{\ell,m}$ & $\mathsf{receipt}_{\ell,m}$\\
		\hline
		& Part B &\\
		$\mathsf{vote\textrm{-}code}_{\ell,1}$ & $\mathsf{option}_{\ell,1}$ & $\mathsf{receipt}_{\ell,1}$\\
		$\quad\hdots$ & $\quad\hdots$ & $\quad\hdots$\\
		$\mathsf{vote\textrm{-}code}_{\ell,m}$ & $\mathsf{option}_{\ell,m}$ & $\mathsf{receipt}_{\ell,m}$\\
		\hline
	\end{tabular}
\ifextended
\else
\end{scriptsize}
\fi
\end{center}

\subsubsection{\BB{} initialization data} 
The initialization data for all \BB{} nodes is identical, and each \BB{} node publishes its initialization data immediately. 
The \BB{}'s data is used to show the correspondence between the vote codes and their associated cryptographic payload.
This payload comprises the committed option encodings, and their respective zero knowledge proofs of valid encoding (first move of the prover), as described in section~\ref{sec:sysoverview}. 
However, the vote codes must be kept secret during the election, to prevent the adversary from ``stealing'' the voters' ballots and using the stolen vote codes to vote. 
To achieve this, the \EA{} first randomly picks a $128$-bit key, $\mathsf{msk}$, and encrypts each $\mathsf{vote\textrm{-}code}$ using AES-128-CBC with random initialization vector (AES-128-CBC\$) encryption, denoted as $[\mathsf{vote\textrm{-}code}]_{\mathsf{msk}}$.
Each \BB{} node is given $H_{\mathsf{msk}}\leftarrow SHA256(\mathsf{msk} , \mathsf{salt}_\mathsf{msk})$ and $\mathsf{salt}_\mathsf{msk}$, where $\mathsf{salt}_\mathsf{msk}$ is a fresh $64$-bit random salt. Hence, each \BB{} node can be assured the key it reconstructs from \VC{} key-shares (see below) is indeed the key that was used to encrypt these vote-codes. 

The rest of the \BB{} initialization data is as follows: for each $\mathsf{serial\textrm{-}no}_\ell$, and for each ballot part, there is a \emph{shuffled} list of $\left\langle[\mathsf{vote\textrm{-}code}_{\ell,\pi_\ell^X(j)}]_{\mathsf{msk}},\mathsf{payload}_{\ell,\pi_\ell^X(j)} \right\rangle$ tuples, where $\pi_\ell^X\in S_m$ is a random permutation ($X$ is $A$ or $B$).
\ifextended
\begin{center}
	\begin{tabular}{|l c|}
		\hline
		\hspace{55pt}$(H_{\mathsf{msk}}, \mathsf{salt}_\mathsf{msk})$  &\\
		\hline
		\hline
		$\mathsf{serial\textrm{-}no}_\ell$  &\\
		\hline
		& Part A \\
		$[\mathsf{vote\textrm{-}code}_{\ell,\pi_\ell^A(1)}]_{\mathsf{msk}}$ & $\mathsf{payload}_{\ell,\pi_\ell^A(1)}$  \\
		\quad\quad $\vdots$ & $\vdots$\\
		$[\mathsf{vote\textrm{-}code}_{\ell,\pi_\ell^A(m)}]_{\mathsf{msk}}$ & $\mathsf{payload}_{\ell,\pi_\ell^A(m)}$  \\
		\hline
		& Part B \\
		$[\mathsf{vote\textrm{-}code}_{\ell,\pi_\ell^B(1)}]_{\mathsf{msk}}$ & $\mathsf{payload}_{\ell,\pi_\ell^B(1)}$  \\
		\quad\quad $\vdots$ & $\vdots$\\
		$[\mathsf{vote\textrm{-}code}_{\ell,\pi_\ell^B(m)}]_{\mathsf{msk}}$ & $\mathsf{payload}_{\ell,\pi_\ell^B(m)}$  \\
		\hline
	\end{tabular}
\end{center}
\fi
We shuffle the list of tuples of each part to ensure voter's privacy. 
This way, nobody can guess the voter's choice from the position of the cast vote-code in this list.

\subsubsection{\VC{} initialization data} 
The \EA{} uses an $(N_v-f_v,N_v)$-VSS (Verifiable Secret-Sharing) scheme to split $\mathsf{msk}$ and every $\mathsf{receipt}_{\ell,j}$ into $N_v$ shares, denoted as $(\|\mathsf{msk}\|_1,\ldots, \|\mathsf{msk}\|_{N_v})$ and $(\|\mathsf{receipt}_{\ell,j}\|_1,\ldots, \allowbreak\|\mathsf{receipt}_{\ell,j}\|_{N_v})$ respectively. 
For each $\mathsf{vote\textrm{-}code}_{\ell,j}$ in each ballot, the \EA{} also computes $H_{\ell,j}\leftarrow SHA256(\mathsf{vote\textrm{-}code}_{\ell,j} , \mathsf{salt}_{\ell,j})$, where $\mathsf{salt}_{\ell,j}$ is a $64$-bit random number. 
$H_{\ell,j}$ allows each \VC{} node to validate a $\mathsf{vote\textrm{-}code}_{\ell,j}$ individually (without network communication), while still keeping the $\mathsf{vote\textrm{-}code}_{\ell,j}$ secret. 
To preserve voter privacy, these tuples are also shuffled using $\pi_\ell^X$.
\ifextended
The initialization data for $VC_i$ is structured as below:
\begin{center}
	\begin{tabular}{|l c|}
	    \hline
	    \multicolumn{2}{|c|}{$\|\mathsf{msk}\|_i$}\\
	    \hline
		\hline
		$\mathsf{serial\textrm{-}no}_\ell$ &  \\
		\hline
                \multicolumn{2}{|c|}{Part A} \\
		$(H_{\ell,\pi_{\ell}^A(1)},\mathsf{salt}_{\ell,\pi_{\ell}^A(1)})$  & $\|\mathsf{receipt}_{\ell,\pi_{\ell}^A(1)}\|_i$ \\
		\quad $\hdots$ & $\hdots$ \\
	$(H_{\ell,\pi_{\ell}^A(m)},\mathsf{salt}_{\ell,\pi_{\ell}^A(m)})$ & $\|\mathsf{receipt}_{\ell,\pi_{\ell}^A(m)}\|_i$\\
                \hline
                \multicolumn{2}{|c|}{Part B} \\
                $(H_{\ell,\pi_{\ell}^B(1)},\mathsf{salt}_{\ell,\pi_{\ell}^B(1)})$  & $\|\mathsf{receipt}_{\ell,\pi_{\ell}^B(1)}\|_i$ \\
                \quad $\hdots$ & $\hdots$ \\
        $(H_{\ell,\pi_{\ell}^B(m)},\mathsf{salt}_{\ell,\pi_{\ell}^B(m)})$ & $\|\mathsf{receipt}_{\ell,\pi_{\ell}^B(m)}\|_i$\\
		\hline
	\end{tabular}
\end{center}
\fi

\subsubsection{\Trustee{} initialization data} 
The \EA{} uses $(h_t,N_t)$-VSS to split the opening of encoded option commitments $\Com(\vec{e}_i)$ into $N_t$ shares, denoted as $(\|\underline{\vec{e}_i}\|_1,\ldots, \allowbreak\|\underline{\vec{e}_i}\|_{N_t})$.
\ifextended
The initialization data for $\mathsf{Trustee}_i$ is structured as below:
\begin{center}
	\begin{tabular}{|l c|}
		\hline
		$\mathsf{serial\textrm{-}no}_\ell$  &\\
		\hline
		\multicolumn{2}{|c|}{Part A}\\
		 $\Com(\vec{e}_{\pi_\ell^A(i)})$  & $\|\underline{\vec{e}_{\pi_\ell^A(i)}}\|_\ell$\\
		$\cdots$ & $\cdots$\\
		\hline
		\multicolumn{2}{|c|}{Part B}\\
	  $\Com(\vec{e}_{\pi_\ell^B(i)})$  & $\|\underline{\vec{e}_{\pi_\ell^B(i)}}\|_\ell$\\
		$\cdots$ & $\cdots$ \\
		\hline
	\end{tabular}
\end{center}
\fi

Similarly, the state of zero knowledge proofs for ballot correctness is shared among the \trustees{} using $(h_t,N_t)$-VSS. 
For further details, we refer the interested reader to~\cite{CP}.

\subsection{Vote Collectors}\label{subsec:Vcnodes}
The Vote Collection subsystem comprises $N_v$ nodes that collect the votes from the voters and, at election end time, agree on a single set of cast vote codes and upload it to the Bulletin Board. 
In the following subsections, we present two different versions of the \VC{} subsystem, one with a timing assumption (\emph{D-DEMOS/IC}) and one fully asynchronous (\emph{D-DEMOS/Async}).

\subsubsection{Vote Collectors for D-DEMOS/IC}\label{subsec:VcnodesIC}
\input{system-vc-ic}

\subsubsection{Vote Collectors for D-DEMOS/Async}\label{subsec:VcnodesAsync}
\input{system-vc}

\subsection{Voter}\label{subsec:voter}
We expect the voter, who has received a ballot from \EA{}, to know the URLs of at least $f_\mathsf{v}+1$ \VC{} nodes.
To vote, she picks one part of the ballot at random, selects the vote code representing her chosen option, and loops, selecting a \VC{} node at random and posting the vote code, until she receives a valid receipt. 
After the election, the voter can verify two things from the updated \BB{}. 
First, she can verify her cast vote code is included in the tally set. 
Second, she can verify that the unused part of her ballot, as ``opened'' at the \BB{}, matches the copy she received before the election started.
This step verifies that the vote codes are associated with the expected options as printed in the ballot.
Finally, the voter can delegate both of these checks to an \emph{auditor}, without sacrificing her privacy.
This is because the cast vote code does not reveal her choice, and because the unused part of the ballot is completely unrelated to the used one.

\subsection{Bulletin Board}\label{sec:BB}
A \BB{} node functions as a public repository of election-specific information. 
By definition, it can be read via a public and anonymous channel.
Writes, on the other hand, happen over an authenticated channel, implemented with PKI originating from the voting system.
\BB{} nodes are independent from each other, as a \BB{} node never directly contacts another \BB{} node.
Readers are expected to issue a read request to all \BB{} nodes, and trust the reply that comes from the majority.
Writers are also expected to write to all \BB{} nodes; their submissions are always verified, and explained in more detail below.
 
After the setup phase, each \BB{} node publishes its initialization data. 
During election hours, \BB{} nodes remain inert.
After the voting phase, each \BB{} node receives from each \VC{} node, the final vote-code set and the shares of $\mathsf{msk}$.
Once it receives $f_v+1$ identical final vote code sets, it accepts and publishes the final vote code set.
Once it receives $N_v-f_v$ valid key shares (again from \VC{} nodes), it reconstructs the $\mathsf{msk}$, decrypts all the encrypted vote codes in its initialization data, and publishes them.
%
%
 
At this point, the cryptographic payloads corresponding to the cast vote codes are made available to the \trustees{}. 
\Trustees{}, in turn, read from the \BB{} subsystem, perform their individual calculations and then write to the \BB{} nodes; these writes are verified by the \trustees{}' keys, generated by the \EA{}. 
Once enough \trustees{} have posted valid data, the \BB{} node combines them and publishes the final election result.

We intentionally designed our \BB{} nodes to be as simple as possible for the reader, refraining from using a \emph{Replicated State Machine}, which would require readers to run algorithm-specific software. 
The robustness of \BB{} nodes comes from controlling all write accesses to them. 
Writes from \VC{} nodes are verified against their honest majority threshold. 
Further writes are allowed only from \trustees{}, verified by their keys.
 
Finally, a reader of our \BB{} nodes should post her read request to all nodes, and accept what the majority responds with ($f_b+1$ is enough). 
We acknowledge there might be temporary state divergence (among \BB{} nodes), from the time a writer updates the first  \BB{} node, until the same writer  updates the last \BB{} node. 
However, given our thresholds, this should be only momentary, alleviated with simple retries.
Thus, if there is no reply backed by a clear majority, the reader should retry until there is one.
\subsection{\Trustees{}}
After the end of election hours, each \trustee{} fetches all the election data from the \BB{} subsystem and verifies its validity. 
For each ballot, there are two possible valid outcomes:
\begin{inparaenum}[i)]
 \item \label{trustees-one-voted}
 one of the A/B parts are voted,
 \item \label{trustees-none-voted}
 none of the A/B parts are voted.
\end{inparaenum}
If both A/B parts of a ballot are marked as voted, then the ballot is considered as invalid and is discarded. 
Similarly, \trustees{} also discard those ballots where more than one commitments in an A/B part are marked as voted.

In case (\ref{trustees-one-voted}), for each encoded option commitment in the unused part, $\mathsf{Trustee}_\ell$ submits its corresponding share of the opening of the commitment to the \BB{}.
For each encoded option commitment in the voted part, $\mathsf{Trustee}_\ell$ computes and posts the share of the final message of the corresponding zero knowledge proof, showing the validity of those commitments.
Meanwhile, those commitments marked as voted are collected to a tally set $\mathbf{E}_{\mr{tally}}$. 
In case (\ref{trustees-none-voted}), for each encoded option commitment in both parts, $\mathsf{Trustee}_\ell$ submits its corresponding share of the opening of the commitment to the \BB{}.
Finally, denote $\mathbf{D}^{(\ell)}_{\mr{tally}}$ as $\mathsf{Trustee}_\ell$'s set of shares of option encoding commitment openings,   corresponding to the commitments in $\mathbf{E}_{\mr{tally}}$.
$\mathsf{Trustee}_\ell$ computes the opening share for $E_{\mr{sum}}$ as $T_\ell = \sum_{D\in \mathbf{D}^{(\ell)}_{\mr{tally}}}$ and then submits $T_\ell$ to each \BB{} node. 

\subsection{Auditors}
\label{sec:auditors}
Auditors are participants of our system who can verify the election process. 
The role of the auditor can be assumed by voters or any other party.
After election end time, auditors read information from the \BB{} and verify the correct execution of the election, by verifying the following:
\ifextended
\begin{enumerate}
\else
\begin{inparaenum}[1)]
\fi
 \item within each opened ballot, no two vote codes are the same;
 \item there are no two submitted vote codes associated with any single ballot part;
 \item within each ballot, no more than one part has been used; 
 \item all the openings of the commitments are valid;
 \item all the zero-knowledge proofs associated with the used ballot parts are completed and valid.
\ifextended
\end{enumerate}
\else
\end{inparaenum}
\fi
In case they received audit information (an unused ballot part and a cast vote code) from voters who wish to delegate verification, they can also verify:
\ifextended
\begin{enumerate}\setcounter{enumi}{5}
\else
\begin{inparaenum}[a)]
\fi
 \item the submitted vote codes are consistent with the ones received from the voters; 
 \item the openings of the unused ballot parts are consistent with the ones received from the voters.
\ifextended
\end{enumerate}
\else
\end{inparaenum}
\fi

%% file: threat_model.tex
\subsection{System and Threat Model}
\label{subsec:threat}
We assume a fully connected network, where each node can reach any other node with which it needs to communicate. 
The network can drop, delay, duplicate, or deliver messages out of order. 
However, we assume messages are eventually delivered, provided the sender keeps retransmitting them.
For all nodes, we make no assumptions regarding processor speeds.

We assume the EA sets up the election and is destroyed upon completion of the setup, as it does not directly interact with the remaining components of the system, thus reducing the attack surface of the privacy of the voting system as a whole.
We also assume initialization data for every system component is relayed to it via untappable channels.
We assume the adversary does not have the computational power to violate the security of any underlying cryptographic primitives. 
We place no bound on the number of faulty nodes the adversary can coordinate, as long as the number of malicious nodes of each
subsystem is below its corresponding fault threshold. 
Let $N_v$, $N_b$, and $N_t$ be the number of VC nodes, BB nodes, and trustees respectively. 
The voters are denoted by $V_\ell$, $\ell=1,\ldots,n$. 

For both versions of our system, we assume the clocks of VC nodes are synchronized with real world time; this is needed to prohibit voters from casting votes outside election hours. For the safety of \emph{D-DEMOS/Async} version, we make no further timing assumptions.
To ensure liveness, we assume the adversary cannot delay communication between honest nodes above a certain threshold.

For the \emph{D-DEMOS/IC} version, we use the \emph{IC,BC-RBB} algorithm achieving Interactive Consistency (IC) from~\cite{IC@ICPADS2015}, which requires a single synchronization point after the beginning of the algorithm. To accommodate this, we use the election-end time as the starting point of IC, and additionally assume the adversary cannot cause clock drifts between VC nodes also for safety, besides liveness. This is because lost messages in the first round of~\emph{IC,BC-RBB} are considered failures of the sending node.

Formally, we assume that there exists a \emph{global clock} variable $\Cl\in\mathbb{N}$, and that every VC node, BB node and voter $X$ is equipped with an \emph{internal clock} variable $\Cl[X]\in\mathbb{N}$. 
We define the following two events on the clocks:  
\begin{enumerate}[(i).]
 \item The event $\mathsf{Init}(X):$ $\mathsf{Clock}[X]\leftarrow\mathsf{Clock}$, that initializes a node $X$ by synchronizing its internal clock with the global clock. 
 \item The event $\mathsf{Inc}(i):$ $i\leftarrow i+1$, that causes some clock $i$ to advance by one time unit.
\end{enumerate}
\par The adversarial setting for $\A$ upon D-DEMOS is defined in Figure~\ref{fig:threat}. 
\begin{boxfig}{\label{fig:threat} The adversarial setting for the adversary $\A$ acting upon the distributed bulletin board system.}{}
  {\it \underline{The adversarial setting.}}
\begin{enumerate}[(1)]
 \item The EA  initializes every VC node, BB node, trustee of the D-DEMOS system by running $\mathsf{Init}(\cdot)$ in all clocks for synchronization. Then, EA prepares the voters' ballots and all the VC nodes', BB nodes', and trustees' initialization data. Finally, it forwards the ballots for ballot distribution to the voters $V_\ell$, $\ell=1,\ldots,n$.
 \item $\A$ corrupts a fixed subset of VC nodes, a fixed subset of BB nodes, and a fixed subset of trustees. In addition, it defines a fixed subset of corrupt voters $\mathcal{V}_\mathsf{corr}$. 
\item When an honest node $X$ wants to transmit a message $\mathbf{M}$ to an honest node $Y$, then it just sends $(X,\mathbf{M},Y)$ to $\A$.
 \item $\A$ may arbitrarily invoke the events $\mathsf{Inc}(\mathsf{Clock})$ or $\mathsf{Inc}(\mathsf{Clock}[X])$, for any node $X$. Moreover, $\A$ may write on the incoming network tape of any honest component node of D-DEMOS.
 \item For every voter $V_\ell$:
 \begin{enumerate}
  \item If $V_\ell\in\mathcal{V}_\mathsf{corr}$, then $\A$ fully controls $V_\ell$.
  \item If $V_\ell\notin\mathcal{V}_\mathsf{corr}$, then $\A$ may initialize $V_\ell$ by running $\mathsf{Init}(V_\ell)$ only once. If this happens, then the only control of $\A$ over $V_\ell$ is $\mathsf{Inc}(\mathsf{Clock}[V_\ell])$ invocations. Upon initialization, $V_\ell$ engages in the voting protocol.
 \end{enumerate} 
\end{enumerate}
\end{boxfig}
\medskip
\par The description in Figure~\ref{fig:threat} poses no restrictions on the control the adversary has over all internal clocks, or the number of nodes that it may corrupt (arbitrary denial of service attacks or full corruption of D-DEMOS nodes are possible). Therefore, it is necessary to strengthen the model so that we can perform a meaningful security analysis and prove the properties (liveness, safety, end-to-end verifiability, and voter privacy) that D-DEMOS achieves. 
Namely, we require the following: \medskip
\begin{enumerate}[\textbf{A.}]
\item\label{req:tolerance}\textsc{Fault tolerance}. 
We consider arbitrary (Byzantine) failures, because we expect our system to be deployed across separate administrative domains. 
For each of the subsystems, we have the following fault tolerance thresholds:
\begin{itemize}
 \item 
The number of faulty VC nodes, $f_v$, is strictly less than $1/3$ of $N_v$
\ifextended
, i.e., for fixed $f_v$: $$\boxed{N_v\geq3f_v+1.}$$
\else
.
\fi
 \item 
The number of faulty BB nodes, $f_b$, is strictly less than $1/2$ of $N_b$
\ifextended
, i.e., for fixed $f_b$: $$\boxed{N_b\geq2f_b+1.}$$
\else
.
\fi
 \item 
For the trustees' subsystem, we apply $h_t$ out-of $N_t$ threshold secret sharing, where $h_t$ is the number of honest trustees, thus we tolerate $f_t=N_t-h_t$ malicious trustees.  
\end{itemize}
\end{enumerate}
%
%
%
%
\begin{enumerate}[\textbf{B.}]
\item\label{req:loss}{\textsc{Bounded synchronization loss.}}
For the liveness of D-DEMOS (both versions), all system participants are aware of a value $T_\mathsf{end}$ such that for each node $X$, if $\mathsf{Clock}[X]\geq T_\mathsf{end}$, then $X$ considers that the election has ended.  In addition, the safety of D-DEMOS/IC version, assumes two timing points, a starting point (that we set as $T_\mathsf{end}$) and a \emph{barrier}, denoted by $T_\mathsf{barrier}$, that determine the beginning of the \emph{Value Dissemination} phase and the transition to the \emph{Result Consensus} phase of the underlying Interactive Consistency protocol (see Section~\ref{section:ic}), respectively.
\par\hspace{4pt} For the above reasons, we bound the drift on the nodes' internal clocks, assuming an upper bound $\Delta$ of the drift of all honest nodes' internal clocks with respect to the global clock. Formally, we have that: $|\mathsf{Clock}[X]-\mathsf{Clock}|\leq\Delta$  for every node $X$, where $|\cdot|$ denotes the absolute value.
 \end{enumerate}
\begin{enumerate}[\textbf{C.}]
\item\label{req:delay}{\textsc{Bounded communication delay.}}
For the liveness of D-DEMOS (both versions) and the safety of D-DEMOS/IC, we need to ensure eventual message delivery in bounded time. Therefore, we assume that there exists an upper bound $\delta$ on the time that $\A$ can delay the delivery of the messages between honest nodes. Formally, when the honest node $X$ sends $(X,\mathbf{M},Y)$ to $\A$, if the value of the global clock is $T$, then $\A$ must write $\mathbf{M}$ on the incoming network tape of $Y$ by the time that  $\mathsf{Clock}=T+\delta$. We note that $\delta$ should be a reasonably small value for liveness, while for safety of D-DEMOS/IC it suffices to be dominated by the predetermined timeouts of the VC nodes.
 \end{enumerate}
\ifextended
\par\noindent For clarity, we recap the aforementioned requirements in Fig.~\ref{fig:requirements}.
 \begin{figure}[ht]
 \centering
 \begin{small}
 \begin{tabular}{|M{2.7cm}||M{2.1cm}|M{2.1cm}|M{2.1cm}|M{2.1cm}|}
 \hline
\multirow{2}{*}{\textbf{Requirement}}&\multicolumn{2}{c|}{\textbf{D-DEMOS/IC}}&\multicolumn{2}{c|}{\textbf{D-DEMOS/Async}}\\\cline{2-5}
&\small{\textbf{Liveness}}&\small{\textbf{Safety}}&\small{\textbf{Liveness}}&\small{\textbf{Safety}}\\
\hline\hline
Fault tolerance of the VC subsystem&\checkmark&\checkmark&\checkmark&\checkmark\\
\hline
Fault tolerance of the BB subsystem&&\checkmark&&\checkmark\\
\hline
Fault tolerance of the trustees' subsystem&&\checkmark&&\checkmark\\
\hline
Bounded synchronization loss&\checkmark&\checkmark&\checkmark&\\
\hline
Bounded communication delay&\checkmark&\checkmark&\checkmark&\\
\hline
 \end{tabular}
 \end{small}
\caption{Requirements for the liveness and safety of D-DEMOS/IC and D-DEMOS/Async.}
\label{fig:requirements}
 \end{figure}
\fi


%% file: system-vc-ic.tex
\VC{} is a distributed system of $N_v$ nodes, running our \emph{voting} and \emph{vote-set consensus} protocols.
\VC{} nodes have private and authenticated channels to each other, and a public (unsecured) channel for voters.
%
\ifextended
The algorithms implementing our \emph{D-DEMOS/IC} \emph{voting} protocol are presented in Algorithm~\ref{alg:VCIC}.
For simplicity, we present our algorithms operating for a single election.
\else
The algorithms are presented in plain text to preserve space. 
Please consult the extended version of this paper~\cite{extended} for an algorithmic view and a message flow diagram.
\fi

The \emph{voting} protocol starts when a voter submits a \VOTE{}$\langle\mathsf{serial\textrm{-}no},\mathsf{vote\textrm{-}code}\rangle$
message to a \VC{} node.
We call this node the \emph{responder}, as it is responsible for delivering the receipt to the voter.
The \VC{} node confirms the current system time is within the defined election hours, and locates the ballot with the specified $\mathsf{serial\textrm{-}no}$. 
It also verifies this ballot has not been used for this election, either with the same or a different vote code.
Then, it compares the $\mathsf{vote\textrm{-}code}$ against every hashed vote code in each ballot line, until it locates the correct entry.
Subsequently, it obtains from its local database the $\mathsf{receipt\textrm{-}share}$ corresponding to the specific vote-code.
Next, it marks the ballot as $\mathsf{pending}$ for the specific $\mathsf{vote\textrm{-}code}$.
Finally, it multicasts a 
\VOTEP{}$\langle\mathsf{serial\textrm{-}no},\mathsf{vote\textrm{-}code}, \mathsf{receipt\textrm{-}share}\rangle$
message to all \VC{} nodes, disclosing its share of the receipt.
In case the located ballot is marked as $\mathsf{voted}$ for the specific $\mathsf{vote\textrm{-}code}$, the \VC{} node sends the stored $\mathsf{receipt}$ to the voter without any further interaction with other \VC{} nodes.
\\\indent 
Each \VC{} node that receives a \VOTEP{} message, first validates the received $\mathsf{receipt\textrm{-}share}$ according to the verifiable secret sharing scheme used.
Then, it performs the same validations as the responder, and multicasts another \VOTEP{} message (only once), disclosing its share of the receipt. 
When a node collects $h_v = N_v - f_v$ valid shares, it uses the verifiable secret sharing reconstruction algorithm to reconstruct the receipt (the secret) and marks the ballot as $\mathsf{voted}$ for the specific $\mathsf{vote\textrm{-}code}$. 
Additionally, the \emph{responder} node sends this receipt back to the voter.
\ifextended
A message flow diagram of our \emph{voting} protocol is depicted in Figure~\ref{figure:VoteCollectionPhase}.
As is evident from the diagram, the time from the multicast of the first \VOTEP{} message until collecting all receipt shares, is only slightly longer than a single round-trip between two \VC{} nodes.
\begin{figure}[ht!]
  \centering
  \includegraphics[width=0.95\textwidth]{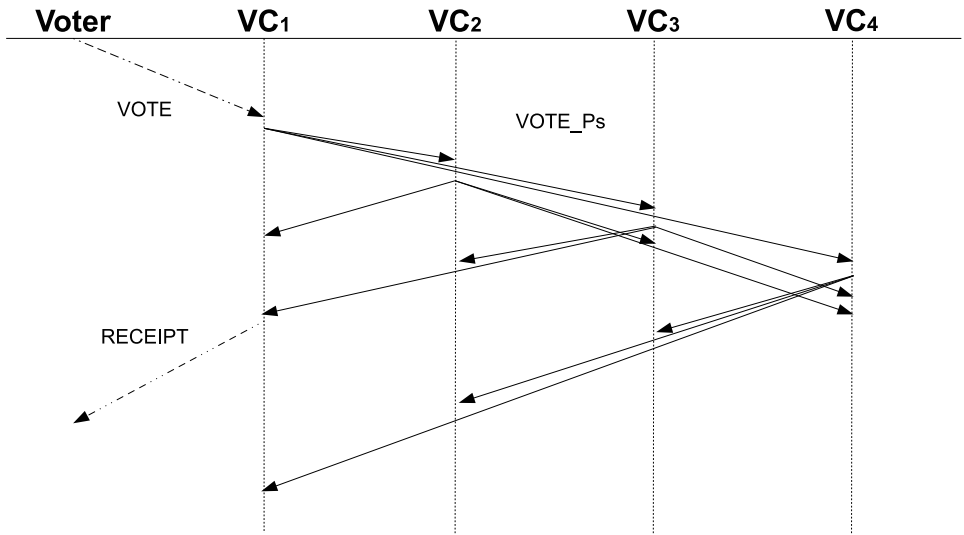}
  \caption{Diagram of message exchanges for a single vote during the D-DEMOS/IC vote collection phase.}
  \label{figure:VoteCollectionPhase}
\end{figure}
\fi
\\\indent 
At election end time, each \VC{} node stops processing \VOTE{} and \VOTEP{} messages, and initiates the \emph{vote-set consensus} protocol.
It creates a set $VS_i$ of $\langle \mathsf{serial\textrm{-}no}, \mathsf{vote\textrm{-}code} \rangle$  tuples, including all \emph{voted} and \emph{pending} ballots.
Then, it participates in the Interactive Consistency (IC) protocol of~\cite{IC@ICPADS2015}, with this set.
At the end of IC, each node contains a vector $\langle VS_1,\ldots,VS_n \rangle$ with the Vote Set of each node, and follows the algorithm of Figure~\ref{fig:afterIC}.
\begin{boxfig}{\label{fig:afterIC}High level description of algorithm after IC.}{}
Cross-tabulate $\langle VS_1,\ldots,VS_n \rangle$ per ballot, creating a list of vote codes for each ballot. Perform the following actions for each ballot:
\begin{enumerate}
 \item If the list contains two or more distinct vote codes, mark the ballot as $\mathsf{NotVoted}$ and exit.\label{afteric:discard}
 \item If a vote code $vc_a$ appears at least $N_v - 2f_v$ times in the list, mark the ballot as $\mathsf{Voted}$ for $vc_a$ and exit.\label{afteric:keep}
 \item Otherwise, mark the ballot as $\mathsf{NotVoted}$ and exit.\label{afteric:not-voted}
\end{enumerate}
\end{boxfig}
Step~\ref{afteric:discard} makes sure any ballot with multiple submitted vote codes is discarded. 
Since vote codes are private, and cannot be guessed by malicious vote collectors, the only way for multiple vote codes to appear is if malicious voters are involved, against whom our system is not obliged to respect our \emph{contract}.\\
\indent With a single vote code remaining, step~\ref{afteric:keep} considers the threshold above which to consider a ballot as voted for a specific vote code. 
We select the $N_v - 2f_v$ threshold for which we are certain that even the following extreme scenario is handled. 
If the \emph{responder} is malicious, submits a receipt to an honest voter, but denies it during \emph{vote-set consensus}, the remaining $N_v - 2f_v$ honest \VC{} nodes that revealed their receipt shares for the generation of the receipt, are enough for the system to accept the vote code (receipt generation requires $N_v - f_v$ nodes, of which $f_v$ may be malicious, thus $N_v - 2f_v$ are necessarily honest). 
\\
\indent Finally, step~\ref{afteric:not-voted} makes sure vote codes that occur less than $N_v - 2f_v$ times are discarded. 
Under this threshold, there is no way a receipt was ever generated.\\
\\\indent 
At the end of this algorithm, each node submits the resulting set of \emph{voted} $\langle \mathsf{serial\textrm{-}no}, \mathsf{vote\textrm{-}code} \rangle$ tuples to each \BB{} node, which concludes its operation for the specific election. 

\ifextended
\begin{algorithm}
\caption{Vote Collector algorithms for D-DEMOS/IC}
\label{alg:VCIC}
\begin{algorithmic}[1]
\NoThen
\footnotesize
\Procedure{on VOTE}{$\mathsf{serial\textrm{-}no},\mathsf{vote\textrm{-}code}$} from $source$:
  \If{$SysTime()$ between $start$ and $end$}
    \State $b := $locateBallot($\mathsf{serial\textrm{-}no}$)
    \If{$b.\mathsf{status} == \mathsf{NotVoted}$}
      \State $l$ := ballot.VerifyVoteCode($\mathsf{vote\textrm{-}code}$)
      \If{$l \neq null$}
        \State $b.\mathsf{status} := \mathsf{Pending}$
        \State $b.\mathsf{used\textrm{-}vc} := \mathsf{vote\textrm{-}code}$
        \State $b.\mathsf{lrs} := \{\}$ \Comment{list of receipt shares}
        \State sendAll(VOTE\_P$\langle\mathsf{serial\textrm{-}no},\mathsf{vote\textrm{-}code}, l.\mathsf{share}\rangle$)
        \State wait for $(N_v-f_v)$ VOTE\_P messages, fill $b.\mathsf{lrs}$
        \State $b.\mathsf{receipt} := \rec(b.\mathsf{lrs})$
        \State $b.\mathsf{status} := \mathsf{Voted}$
        \State $send(source, b.\mathsf{receipt})$
      \EndIf
    \ElsIf{$b.\mathsf{status} == \mathsf{Voted}$ \textbf{AND} $b.\mathsf{used\textrm{-}vc} == \mathsf{vote\textrm{-}code}$}
      \State send ($source$, $\mathsf{ballot.receipt}$)
    \EndIf
  \EndIf
\EndProcedure
\item[]
%
\Procedure{on VOTE\_P}{$\mathsf{serial\textrm{-}no},\mathsf{vote\textrm{-}code}, \mathsf{share}$} from $source$:
   \If{$SysTime()$ between $start$ and $end$}
      \State $b := $locateBallot($\mathsf{serial\textrm{-}no}$)
    \If{$b.\mathsf{status} == \mathsf{NotVoted}$}
       \State $l$ := ballot.VerifyVoteCode($\mathsf{vote\textrm{-}code}$)
       \If{$l \neq null$}
         \State $b.\mathsf{status} := \mathsf{Pending}$
         \State $b.\mathsf{used\textrm{-}vc} := \mathsf{vote\textrm{-}code}$
         \State $b.\mathsf{lrs}.\mathsf{Append}(\mathsf{share})$
         \State sendAll(VOTE\_P($\mathsf{serial\textrm{-}no},\mathsf{vote\textrm{-}code}, l.\mathsf{share}$) )
       \EndIf
    \ElsIf{$b.\mathsf{status} == \mathsf{Voted}$ \textbf{AND} $b.\mathsf{used\textrm{-}vc} == \mathsf{vote\textrm{-}code}$}
        \State $b.\mathsf{lrs}.\mathsf{Append}(\mathsf{share})$
        \If{size($b.\mathsf{lrs}$) $>= N_v-f_v$}
          \State $b.\mathsf{receipt} := \rec(b.\mathsf{lrs})$
          \State $b.\mathsf{status} :=  \mathsf{Voted}$
        \EndIf
    \EndIf
  \EndIf
\EndProcedure
\item[]

\Function{Ballot::VerifyVoteCode}{$\mathsf{vote\textrm{-}code}$}
  \For{$l=1$ to $\mathsf{ballot\_lines}$}
    \If{$\mathsf{lines}[l].\mathsf{hash} == h(\mathsf{vote\textrm{-}code}||\mathsf{lines}[l].\mathsf{salt})$}
      \Return $l$
    \EndIf
  \EndFor
  \Return $null$
\EndFunction
\end{algorithmic}
\end{algorithm}
\fi

%% file: system-vc.tex
We make the following enhancements to the Vote Collection subsystem, to achieve the completely asynchronous version \emph{D-DEMOS/Async}.
During voting we introduce another step, which guarantees only a single vote code can be accepted (towards producing a receipt) for a given ballot. We also employ an asynchronous binary consensus primitive to achieve Vote Set Consensus.
\\\indent
More specifically, during voting, the \emph{responder} \VC{} node validates the submitted vote code, but before disclosing its receipt share, it multicasts an \ENDORSE{}$\langle\mathsf{serial\textrm{-}no},\mathsf{vote\textrm{-}code}\rangle$ message to all \VC{} nodes.
Each \VC{} node, after making sure it has not endorsed another vote code for this ballot, responds with an \ENDORSEMENT{}$\langle\mathsf{serial\textrm{-}no},\mathsf{vote\textrm{-}code},\mathsf{sig_{VC_i}}\rangle$ message, where $\mathsf{sig_{VC_i}}$ is a digital signature of the specific serial-no and vote-code, with $VC_i$'s private key.
The responder collects $N_v-f_v$ valid signatures and forms a uniqueness certificate $\mathsf{UCERT}$ for this ballot. 
It then discloses its receipt share via the \VOTEP{} message, but also attaches the formed $\mathsf{UCERT}$ in the message.
\\\indent
Each \VC{} node that receives a VOTE\_P message, first verifies the validity of $\mathsf{UCERT}$ and discards the message on error.
On success, it proceeds as per the \emph{D-DEMOS/IC} protocol (validating the receipt share it receives and then disclosing its own receipt share).
\\\indent
\ifextended
The algorithms implementing our \emph{D-DEMOS/Async} \emph{voting} protocol are presented in Algorithm~\ref{alg:VC}. 

The voting process is outlined in the diagram of Figure~\ref{figure:VoteCollectionAsyncPhase}, where we now see two round-trips are needed before the receipt is reconstructed and posted to the voter.
\begin{figure}[ht!]
  \centering
  \includegraphics[width=0.95\textwidth]{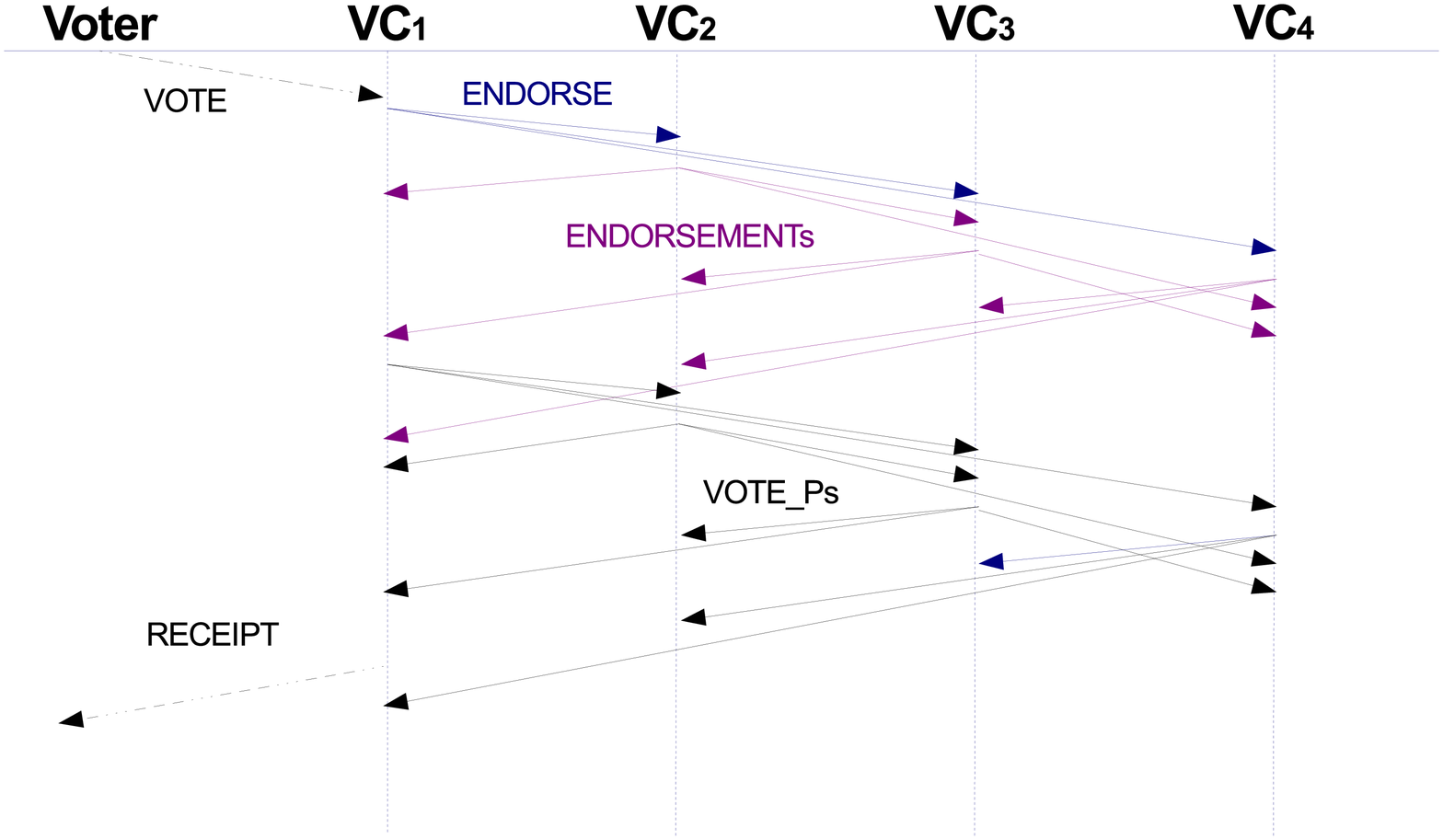}
  \caption{Diagram of message exchanges for a single vote during the D-DEMOS/Async vote collection phase.}
  \label{figure:VoteCollectionAsyncPhase}
\end{figure}
\fi
\\\indent 
The formation of a valid $\mathsf{UCERT}$ gives our algorithms the following guarantees:
\begin{enumerate}[a)]
\item No matter how many responders and vote codes are active at the same time for the same ballot, if a $\mathsf{UCERT}$ is formed for vote code $vc_a$, no other uniqueness certificate for any vote code different than $vc_a$ can be formed.
\item By verifying the $\mathsf{UCERT}$ before disclosing a \VC{} node's receipt share, we guarantee the voter's receipt cannot be reconstructed unless a valid $\mathsf{UCERT}$ is present. 
\end{enumerate}
\indent
At election end time, each \VC{} node stops processing \ENDORSE{}, \ENDORSEMENT{}, \VOTE{} and \VOTEP{} messages, and follows the \emph{vote-set consensus} algorithm in Figure~\ref{fig:vcsasync}, for each registered ballot.

\begin{boxfig}{\label{fig:vcsasync}High level description of algorithm for asynchronous vote set consensus.}{}
\begin{enumerate}
\item \label{bc-step-announce} 
Send \ANNOUNCE{}$\langle\mathsf{serial\textrm{-}no},\mathsf{vote\textrm{-}code}, \mathsf{UCERT}\rangle$ to all nodes. 
The vote-code will be \emph{null} if the node knows of no vote code for this ballot.

\item \label{bc-step-wait} 
Wait for $N_v-f_v$ such messages. 
If any of these messages contains a valid vote code $vc_a$, accompanied by a valid $\mathsf{UCERT}$, change the local state immediately, by setting $vc_a$ as the vote code used for this ballot.

\item \label{bc-step-bc} 
Participate in a Binary Consensus protocol, with the subject ``Is there a valid vote code for this ballot?''. 
Enter with an opinion of $1$, if a valid vote code is locally known, or a $0$ otherwise.

\item \label{bc-step-result-0} 
If the result of Binary Consensus is $0$, consider the ballot not voted. 

\item \label{bc-step-result-1} 
Else, if the result of Binary Consensus is $1$, consider the ballot voted. 
There are two sub-cases here: 
  \begin{enumerate}[a)]
  \item \label{bc-step-result-1-known} 
  If vote code $vc_a$, accompanied by a valid $\mathsf{UCERT}$ is locally known, consider the ballot voted for $vc_a$. 

  \item \label{bc-step-result-1-unknown} 
  If, however, $vc_a$ is not known, send a \RECREQ{}$\langle\mathsf{serial\textrm{-}no}\rangle$ message to all \VC{} nodes, wait for the first valid \RECRES{}$\langle\mathsf{serial\textrm{-}no}, vc_a, \mathsf{UCERT}\rangle$ response, and update the local state accordingly.
  \end{enumerate}
\end{enumerate}
\end{boxfig}

Steps~\ref{bc-step-announce}-\ref{bc-step-wait} ensure used vote codes are dispersed across nodes. 
Recall our receipt generation requires $N_v-f_v$ shares to be revealed by distinct \VC{} nodes, of which at least $N_v-2f_v$ are honest. 
Note that any two $N_v-f_v$ subsets of $N_v$ have at least one honest node in common. 
Because of this, if a receipt was generated, at least one  honest node's \ANNOUNCE{} will be processed by every honest node, and all honest \VC{} nodes will obtain the corresponding vote code in these two steps.
Consequently, all honest nodes enter step~\ref{bc-step-bc} with an opinion of $1$ and binary consensus is guaranteed to deliver $1$ as the resulting value, thus safeguarding our contract against the voters.
In any case, step~\ref{bc-step-bc} guarantees all \VC{} nodes arrive at the same conclusion, on whether this ballot is voted or not.

In the algorithm outlined above, the result from binary consensus is translated from $0$/$1$ to a status of ``not-voted'' or a unique valid vote code, in steps~\ref{bc-step-result-0}-\ref{bc-step-result-1}. 
The~\ref{bc-step-result-1-unknown} case of this translation, in particular, requires additional explanation. 
Assume, for example, that a voter submitted a valid vote code $vc_a$, but a receipt was not generated before election end time. 
In this case, an honest vote collector node $VC_i$ may not be aware of $vc_a$ at step~\ref{bc-step-bc}, as steps~\ref{bc-step-announce}-\ref{bc-step-wait} do not make any guarantees in this case. 
Thus, $VC_i$ may rightfully enter consensus with a value of $0$.  
However, when honest nodes' opinions are mixed, the consensus algorithm may produce any result.
In case the result is $1$, $VC_i$ will not possess the correct vote code $vc_a$, and thus will not be able to properly translate the result. 
This is what our recovery sub-protocol is designed for. 
$VC_i$ will issue a \RECREQ{} multicast, and we claim that another honest node, $VC_h$ exists that \emph{possesses} $vc_a$ and \emph{replies} with it. 
The reason for the existence of an honest $VC_h$ is straightforward and stems from the properties of the binary consensus problem definition.
If all honest nodes enter binary consensus with the same opinion $a$, the result of any consensus algorithm is guaranteed to be $a$. 
Since we have an honest node $VC_i$, that entered consensus with a value of $0$, but a result of $1$ was produced, there has to exist another honest node $VC_h$ that entered consensus with an opinion of $1$. 
Since $VC_h$ is honest, it must \emph{possess} $vc_a$, along with the corresponding $\mathsf{UCERT}$ (as no other vote code $vc_b$ can be active at the same time for this ballot). 
Again, because $VC_h$ is honest, it will follow the protocol and \emph{reply} with a well formed RECOVER-REPLY. 
Additionally, the existence of $\mathsf{UCERT}$ guarantees that any malicious replies can be safely identified and discarded.

As per \emph{D-DEMOS/IC}, at the end of this algorithm, each node submits the resulting set of \emph{voted} $\langle \mathsf{serial\textrm{-}no}, \mathsf{vote\textrm{-}code} \rangle$ tuples to each \BB{} node, which concludes its operation for the specific election. 


\ifextended
\begin{algorithm}
\caption{Vote Collector algorithms for D-DEMOS/Async}
\label{alg:VC}
\begin{algorithmic}[1]
\NoThen
\footnotesize
\Procedure{on VOTE}{$\mathsf{serial\textrm{-}no},\mathsf{vote\textrm{-}code}$} from $source$:
  \If{$SysTime()$ between $start$ and $end$}
    \State $b := $locateBallot($\mathsf{serial\textrm{-}no}$)
    \If{$b.\mathsf{status} == \mathsf{NotVoted}$}
      \State $l$ := ballot.VerifyVoteCode($\mathsf{vote\textrm{-}code}$)
      \If{$l \neq null$}
        \State $b.\mathsf{UCERT} := \{\}$ \Comment{Uniqueness certificate}
        \State sendAll(ENDORSE$\langle\mathsf{serial\textrm{-}no},\mathsf{vote\textrm{-}code}\rangle$)
        \State wait for $(N_v-f_v)$ valid replies, fill $b.\mathsf{UCERT}$
        \State $b.\mathsf{status} := \mathsf{Pending}$
        \State $b.\mathsf{used\textrm{-}vc} := \mathsf{vote\textrm{-}code}$
        \State $b.\mathsf{lrs} := \{\}$ \Comment{list of receipt shares}
        \State sendAll(VOTE\_P$\langle\mathsf{serial\textrm{-}no},\mathsf{vote\textrm{-}code}, l.\mathsf{share}\rangle$)
        \State wait for $(N_v-f_v)$ VOTE\_P messages, fill $b.\mathsf{lrs}$
        \State $b.\mathsf{receipt} := \rec(b.\mathsf{lrs})$
        \State $b.\mathsf{status} := \mathsf{Voted}$
        \State $send(source, b.\mathsf{receipt})$
      \EndIf
    \ElsIf{$b.\mathsf{status} == \mathsf{Voted}$ \textbf{AND} $b.\mathsf{used\textrm{-}vc} == \mathsf{vote\textrm{-}code}$}
      \State send ($source$, $\mathsf{ballot.receipt}$)
    \EndIf
  \EndIf
\EndProcedure
\item[]
%
\Procedure{on VOTE\_P}{$\mathsf{serial\textrm{-}no},\mathsf{vote\textrm{-}code}, \mathsf{share}, \mathsf{UCERT}$} from $source$:
   \If{$\textsf{UCERT}$ is not valid}
      \State return
   \EndIf
   \If{$SysTime()$ between $start$ and $end$}
      \State $b := $locateBallot($\mathsf{serial\textrm{-}no}$)
    \If{$b.\mathsf{status} == \mathsf{NotVoted}$}
       \State $l$ := ballot.VerifyVoteCode($\mathsf{vote\textrm{-}code}$)
       \If{$l \neq null$}
         \State $b.\mathsf{status} := \mathsf{Pending}$
         \State $b.\mathsf{used\textrm{-}vc} := \mathsf{vote\textrm{-}code}$
         \State $b.\mathsf{lrs}.\mathsf{Append}(\mathsf{share})$
         \State sendAll(VOTE\_P($\mathsf{serial\textrm{-}no},\mathsf{vote\textrm{-}code}, l.\mathsf{share}$) )
       \EndIf
    \ElsIf{$b.\mathsf{status} == \mathsf{Voted}$ \textbf{AND} $b.\mathsf{used\textrm{-}vc} == \mathsf{vote\textrm{-}code}$}
        \State $b.\mathsf{lrs}.\mathsf{Append}(\mathsf{share})$
        \If{size($b.\mathsf{lrs}$) $>= N_v-f_v$}
          \State $b.\mathsf{receipt} := \rec(b.\mathsf{lrs})$
          \State $b.\mathsf{status} :=  \mathsf{Voted}$
        \EndIf
    \EndIf
  \EndIf
\EndProcedure
\item[]

\Function{Ballot::VerifyVoteCode}{$\mathsf{vote\textrm{-}code}$}
  \For{$l=1$ to $\mathsf{ballot\_lines}$}
    \If{$\mathsf{lines}[l].\mathsf{hash} == h(\mathsf{vote\textrm{-}code}||\mathsf{lines}[l].\mathsf{salt})$}
      \Return $l$
    \EndIf
  \EndFor
  \Return $null$
\EndFunction
\end{algorithmic}
\end{algorithm}
\fi

%% file: attacks.tex
\section{Potential attacks}\label{section:attacks}
In this section, we outline some of the possible attacks against the D-DEMOS systems, and the way our systems thwart them.
This is a high level discussion, aiming to help the reader understand \emph{why} our systems work reliably.
In Appendix~\ref{sec:security_full}, we provide the formal proofs of correctness and privacy, which are the foundation of this discussion.

In this high-level description, we intentionally do not focus on Denial-of-Service attacks, as these kind of attacks attempt to stop the system from producing a result, or stop voters from casting their votes. 
Although these attacks are important, they cannot be hidden, as voters will notice immediately the system not responding (either because of our receipt mechanism and our liveness property, or because of lack of information in the \BB{}).
Instead, we focus on attacks on the correctness of the election result, as these have consequences simple voters cannot identify easily.
In this discussion, we assume the fault thresholds of section~\ref{subsec:threat} are not violated, and the attacker cannot violate the security of the underlying cryptographic primitives.


In this section, we focus on correctness, noting that our systems' privacy is achieved by the security of our cryptographic schemes (see Sections~\ref{subsec:sec_tools} and \ref{subsec:sec_priv} for details), and the partial initialization data that each node of the distributed subsystems receives at the setup phase.

\subsection{Malicious Election Authority Component}
At a high level, the \EA{} produces vote codes and corresponding receipts.
Vote codes are pointers to the associated cryptographic payload, which includes \emph{option encodings}.
Options encodings are used to produce the tally using homomorphic addition.
If the \EA{} miss-encodes any option, it will be identified by the Zero-Knowledge proof validation performed by the Auditors.

The \EA{} may instead try to ``point'' a vote code to a valid but different option encoding (than the one described in the voter's ballot), in an attempt to manipulate the result.
In this case, the \EA{} cannot predict which one of the two parts the voter will use.
Recall that the unused part of the ballot will be opened in the \BB{} by the \trustees{}, and thus the voters can read and verify the correctness of their unused ballot parts.

As explained in detail in section \ref{subsec:sec_e2e}, if none of the above attacks take place, there is perfect consistency between each voter's ballot and its corresponding information on the \BB{}. 
Because of this, as well as the correctness and the perfect hiding property of our commitment scheme, the homomorphic tally will be opened to the actual election result.

\subsection{Malicious Voter}
\label{sec:AttacksVoter}
A malicious voter can try to submit multiple vote codes to the \VC{} subsystem, attempting to cause disagreement between its nodes.
In this case, a receipt \emph{may} be generated, depending on the order of delivery of network messages.
Note that, our safety \emph{contract} allows our system to either accept only one vote code for this ballot, or discard the ballot altogether, as the voter is malicious and our contract holds only against honest voters.

In the D-DEMOS/IC case, this is resolved at the \emph{Vote Set Consensus} phase.
During the \emph{voting} phase, each \VC{} node accepts only the first vote code it receives (via either a \VOTE{} or a \VOTEP{} message), and attempts to follow our \emph{voting} protocol.
This results in the generation of at most one receipt, for one of the posted vote codes.
However, during \emph{Vote Set Consensus}, honest \VC{} nodes will typically identify the multiple posted vote codes and discard the ballot altogether, even if a receipt was indeed generated.
If the ballot is not discarded (e.g., because malicious vote collector nodes hid the extra vote codes and honest nodes knew only of one), our $N_v-2f_v$ threshold guarantees that no vote codes with generated receipts are discarded.

In the D-DEMOS/Async case, this is resolved completely at the \emph{voting} phase.
Each \VC{} node still accepts only the first vote code it receives, but additionally attempts to build a \UCERT{} for it.
As the generation of a \UCERT{} is guaranteed to be successful only for a single vote code, the outcome of the \emph{voting} protocol will be either no \UCERT{} being built, resulting in considering the ballot as not-voted, or a single \UCERT{} generated.

Thus, the two systems behave differently in the case of multiple posted vote codes, as D-DEMOS/IC typically discards such ballots, while D-DEMOS/Async may process some of them, when a \UCERT{} is successfully built.

\subsection{Malicious Vote Collector}
A malicious \VC{} node cannot easily guess the vote codes in the voters' ballots, as they are randomly generated.
Additionally, because vote codes are encrypted in the local state of each \VC{} node, the latter cannot decode and use them.
Note that, a vote code in a voter's ballot is considered private until the voter decides to use it and transmits it over the network.
From this point on, the vote code can be intercepted by the attacker, as the only power it gives him is to cast it.

A malicious \VC{} node can obtain vote codes from colluding malicious voters.
In this case, the only possible attack on correctness is exactly the same as if it originated from the malicious voter herself, and we already described our counter-measures in Section~\ref{sec:AttacksVoter}.

A malicious \VC{} node may become a \emph{responder}. 
In this case, this \VC{} node may \emph{selectively} transmit the cast vote code to a subset of the remaining \VC{} nodes, potentially including all the other malicious and colluding nodes, and deliver the receipt to an honest voter. 
Consequently, the attacker controlling the malicious entities, may try to ``confuse'' the honest \VC{} nodes and have them disagree on whether the ballot is voted or not, by having all malicious \VC{} nodes lie at \emph{vote set consensus} time, reporting the ballot as not voted.

Recall that, for the receipt to be generated, $N_v-f_v$ \VC{} nodes need to cooperate, of which up to $f_v$ may be malicious.
This leaves $N_v-2f_v$ honest nodes always present.

In the case of D-DEMOS/IC, these $N_v-2f_v$ honest nodes will show up in the per ballot cross-tabulation, and will drive the decision to mark the ballot as voted (note that, in the algorithm of Figure~\ref{fig:afterIC}, $N_v-2f_v$ is the lower threshold for a ballot to be marked as voted).
In the case of D-DEMOS/Async, we include the \ANNOUNCE{}-exchanging phase before the consensus algorithm, to guarantee at least one of the $N_v-2f_v$ honest nodes' \ANNOUNCE{} message will be processed by every honest node.
In this case, all honest nodes will agree on entering consensus that the ballot is voted, which guarantees the outcome of consensus to be in accordance.

\subsection{Malicious \BB{} nodes and \trustees{}}
Malicious entities between both the \BB{} nodes and the \trustees{} cannot influence the security of both systems.
The reason is, a node of each of these two subsystems does not communicate with the remaining nodes of the same subsystem, and thus cannot influence either the correctness, or progress of the system as a whole.

%% file: evaluation_ext.tex
\section{Implementation and evaluation}\label{section:impl_and_eval}
\input{implementation}

\subsection{Evaluation}
\label{sect:evaluation}
We experimentally evaluate the performance of our voting system, focusing mostly on our vote collection algorithm, which is the most performance critical part. 
We conduct our experiments using a cluster of 12 machines, connected over a Gigabit Ethernet switch. 
The first 4 are equipped with Hexa-core Intel Xeon E5-2420 @ 1.90GHz, 16GB RAM, and one 1TB SATA disk, running CentOS 7 Linux, and we use them to run our VC nodes. 
The remaining 8 comprise dual Intel(R) Xeon(TM) CPUs @ 2.80GHz, with 4GB of main memory, and two 50GB disks, running CentOS 6 Linux, and we use them as clients. 

We implement a multi-threaded voting client to simulate concurrency. 
This client starts the requested number of threads, each of which loads its corresponding ballots from disk and waits for a signal to start.
From then on, the thread enters a loop where it picks one VC node and vote code at random, requests the voting page from the selected VC (HTTP GET), submits its vote (HTTP POST), and waits for the reply (receipt). 
This simulates multiple concurrent voters casting their votes in parallel, and gives an understanding of the behavior of the system under the corresponding load. 
We employ the PostgreSQL RDBMS~\cite{PostgreSQL} to store all VC initialization data from the \EA{}. 

\begin{figure*}
\centering
{
  \subfloat[]
  {
    \includegraphics[width=0.50\textwidth]{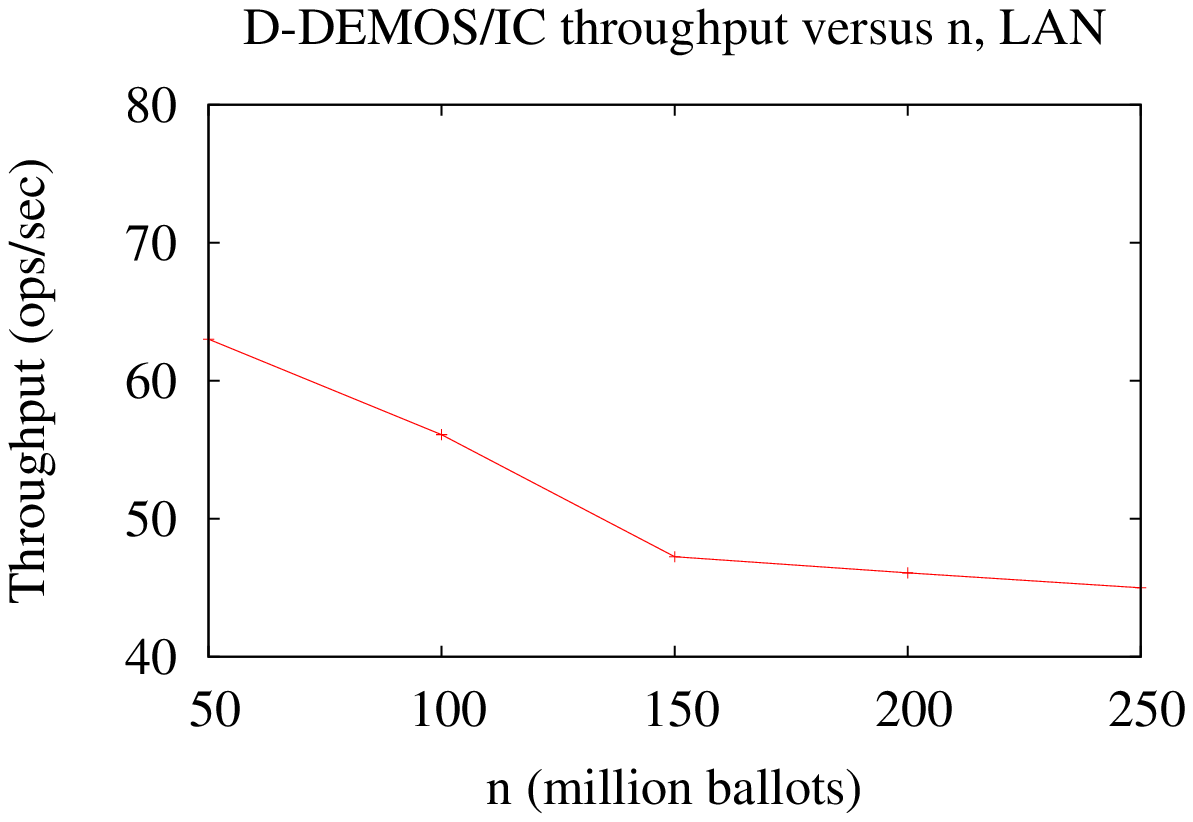}
    \label{fig:LANThroughputnIC}
  }
  \subfloat[]
  {
    \includegraphics[width=0.50\textwidth]{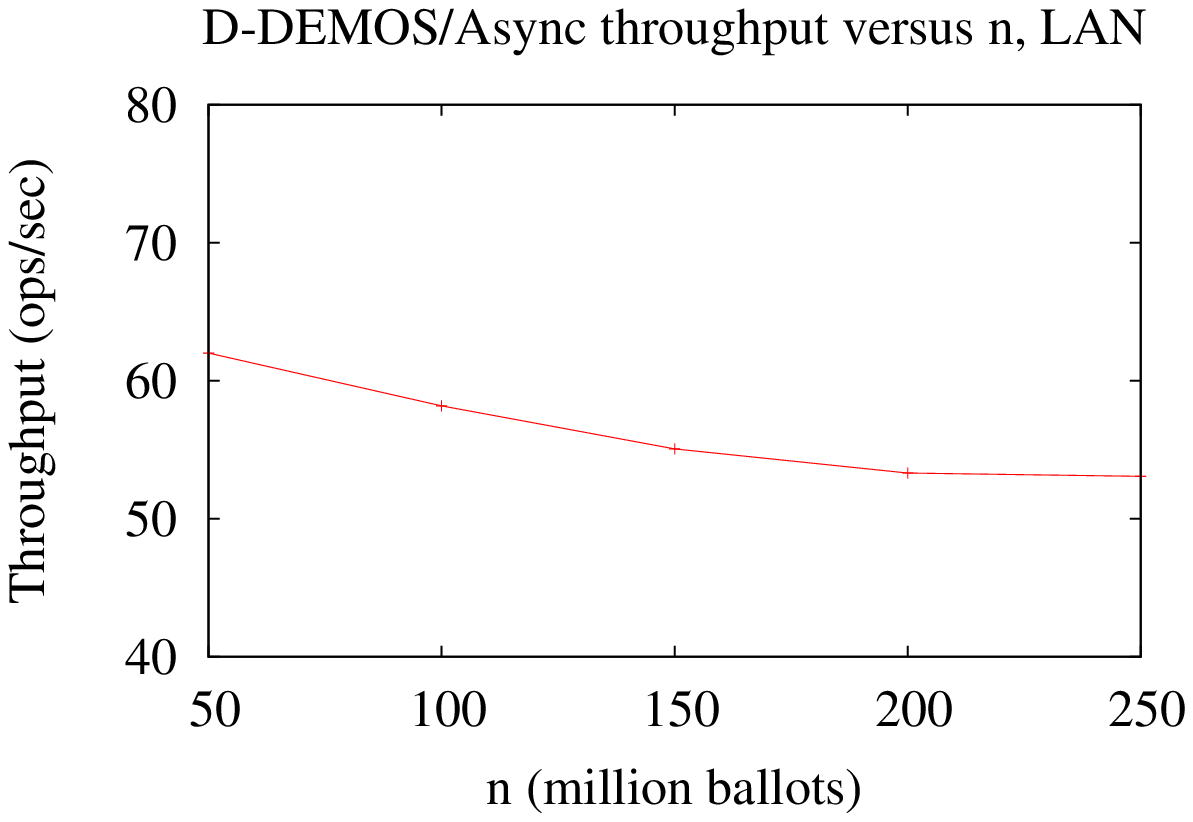}
    \label{fig:LANThroughputnAsync}
  }
  \caption{Vote collection throughput graphs for D-DEMOS/IC (\ref{fig:LANThroughputnIC}) and D-DEMOS/Async(\ref{fig:LANThroughputnAsync}), versus the number of total election ballots $n$.} 
  \label{fig:LANThroughputnAll}
}
\end{figure*}

We start off by demonstrating our system's capability of handling large-scale elections. 
To this end, we generate election data for referendums, i.e., $m=2$, and vary the total number of ballots $n$ from 50 million to 250 million (note the 2012 US voting population size was 235 million). 
This causes the database size to increase accordingly and impact queries.
We fix the number of concurrent clients to 400 and cast a total of 200,000 ballots, which are enough for our system to reach its steady-state operation (larger experiments result in the same throughput). 
Figure~\ref{fig:LANThroughputnAll} shows the throughput of both D-DEMOS/IC and D-DEMOS/Async declines slowly, even with a five-fold increase in the number of eligible voters. The cause of the decline is the increase of the database size. 

\begin{figure*}
\centering
{
  \subfloat[]
  {
    \includegraphics[width=0.50\textwidth]{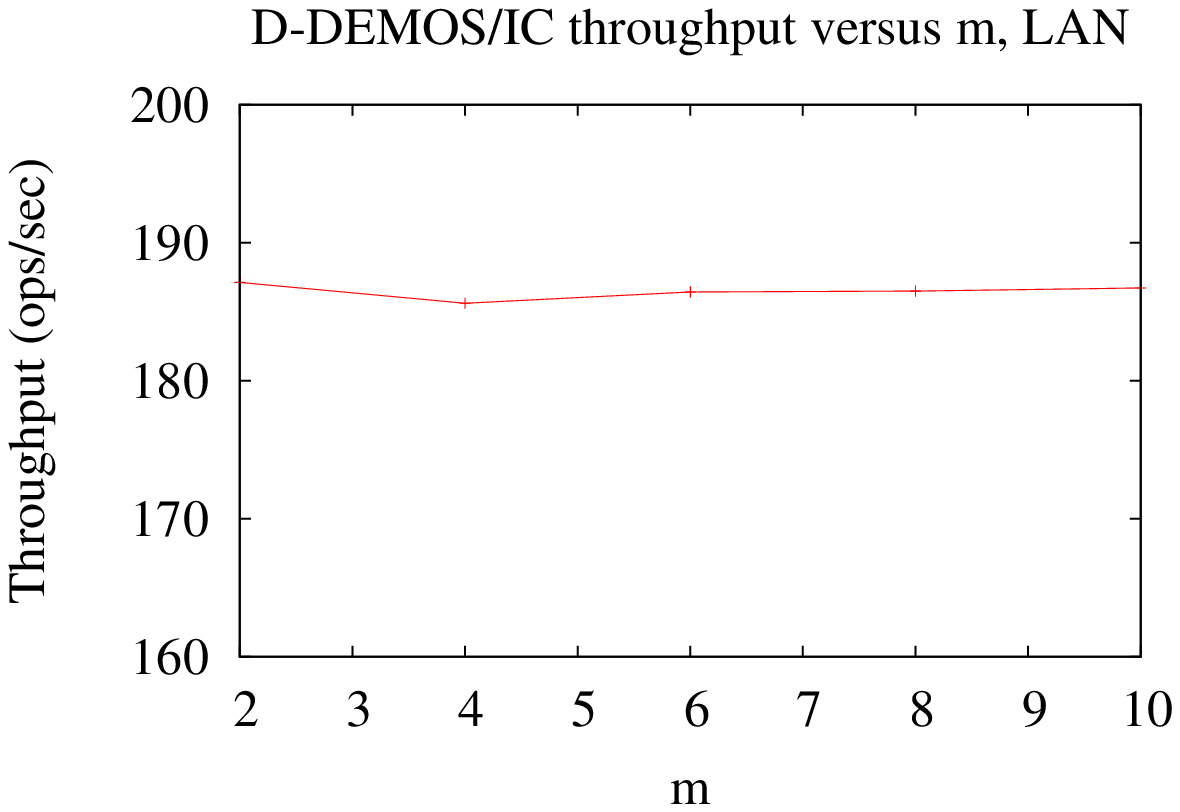}
    \label{fig:LANThroughputmIC}
  }
  \subfloat[]
  {
    \includegraphics[width=0.50\textwidth]{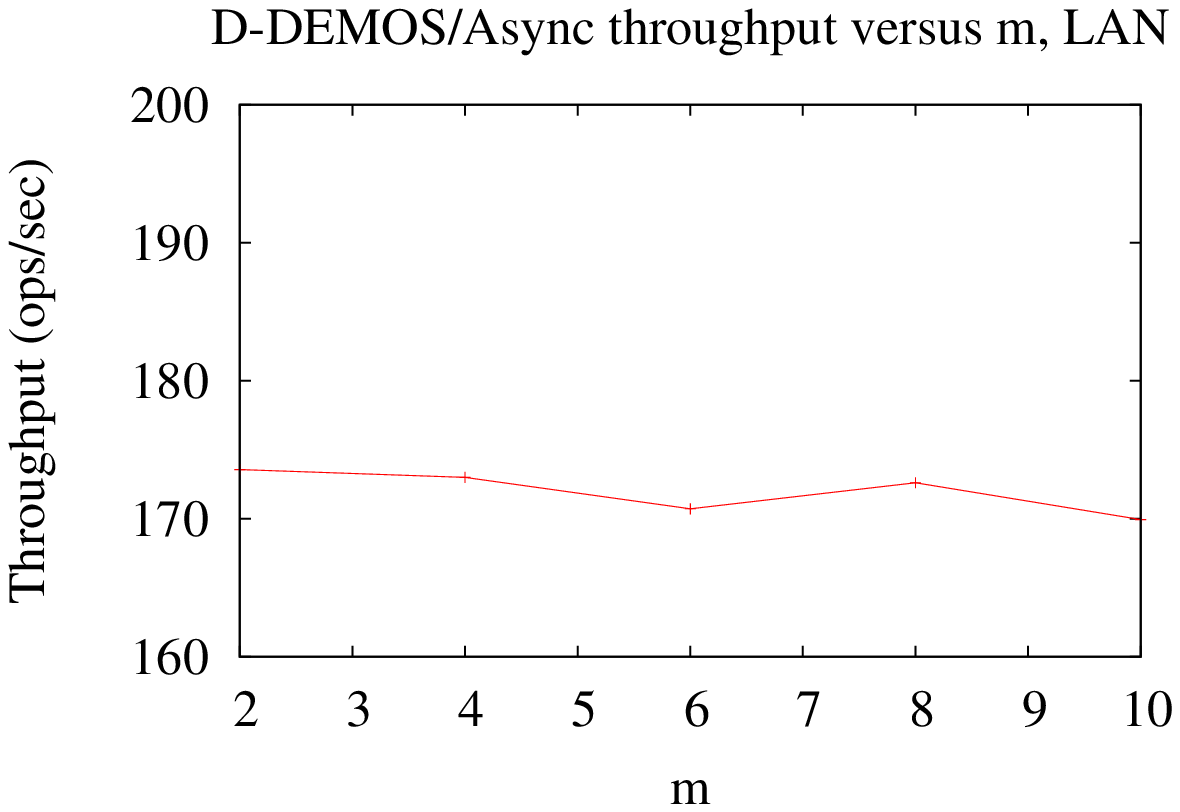}
    \label{fig:LANThroughputmAsync}
  }
  \caption{Vote collection throughput graphs for D-DEMOS/IC (\ref{fig:LANThroughputmIC}) and D-DEMOS/Async(\ref{fig:LANThroughputmAsync}), versus the number of election options $m$.} 
  \label{fig:LANThroughputmAll}
}
\end{figure*}

In our second experiment, we explore the effect of $m$, i.e., the number of election options, on system performance.
We vary the number of options from $m=2$ to $m=10$. 
Each election has a total of $n=200,000$ ballots which we spread evenly across 400 concurrent clients. 
As illustrated in Figure \ref{fig:LANThroughputmAll}, our vote collection protocol
manages to deliver approximately the same throughput regardless of the value of $m$, for both D-DEMOS/IC and D-DEMOS/Async. 
Notice that the major extra overhead $m$ induces during vote collection, is the increase in the number of hash verifications during vote code validation, as there are more vote codes per ballot. The increase in number of options has a minor impact on the database size as well (as each ballots has $2m$ options).

Next, we evaluate the scalability of our vote collection protocol by varying the number of vote collectors and concurrent clients.
We eliminate the database, by caching the election data in memory and servicing voters from the cache, to measure the net communication and processing costs of our voting protocol.
We vary the number of VC nodes from 4 to 16, and distribute them across the 4 physical machines. 
Note that, co-located nodes are unable to produce vote receipts via local messages only, since the $N_{v}-f_{v}$ threshold cannot be satisfied, i.e., cross-machine communication is still the dominant factor in receipt generation.
For election data, we use the dataset with  $n=200,000$ ballots and $m=4$ options, which is enough for our system to reach its steady state. 

\begin{figure*}
\centering
{
  \subfloat[]
  {
    \includegraphics[width=0.50\textwidth]{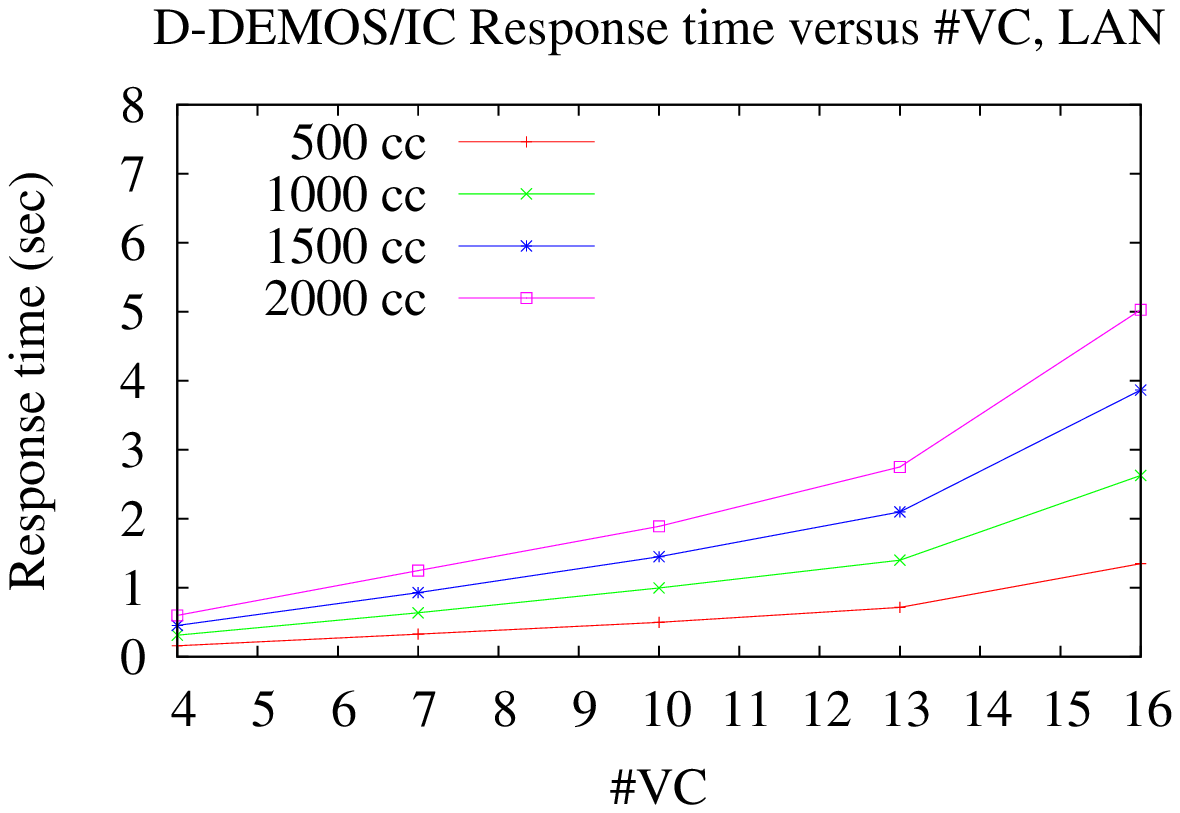}
    \label{fig:LANResponseTimeVCIC}
  }
  \subfloat[]
  {
    \includegraphics[width=0.50\textwidth]{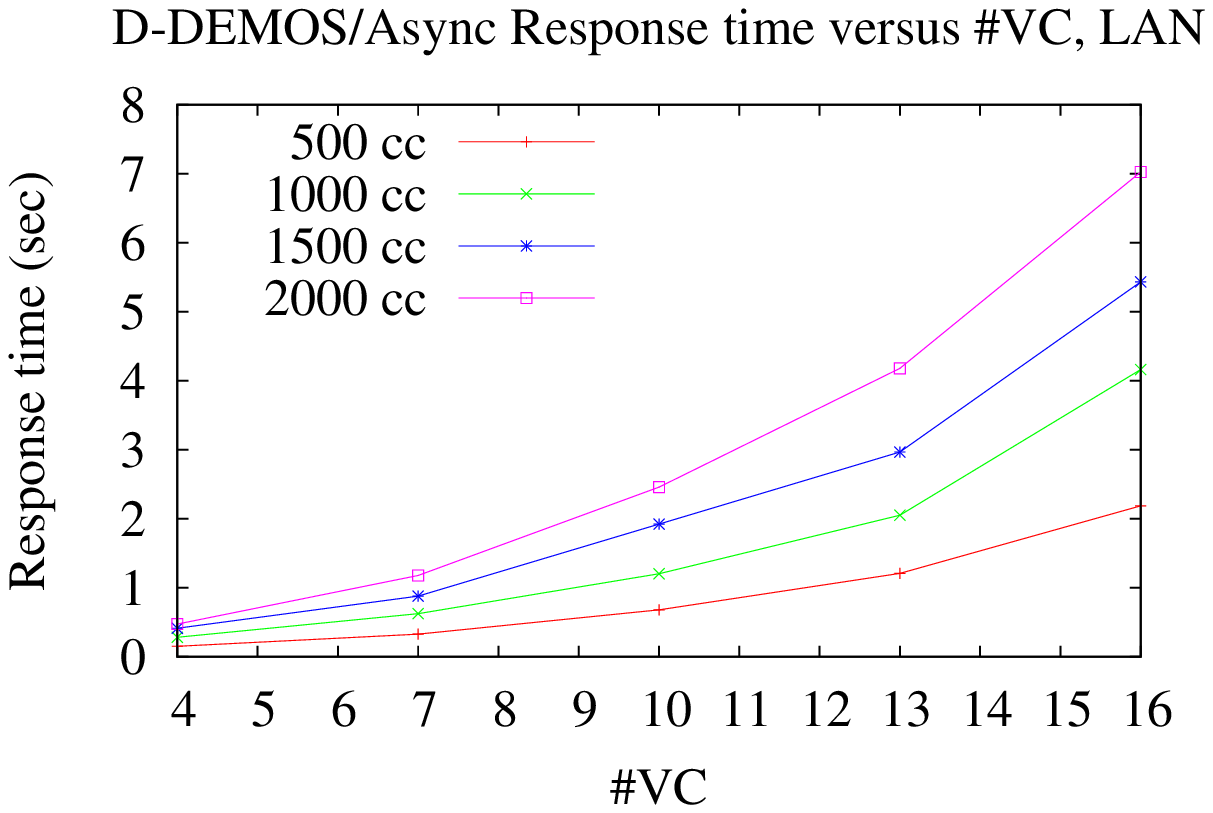}
    \label{fig:LANResponseTimeVCAsync}
  }
  \caption{Vote Collection response time of D-DEMOS/IC (\ref{fig:LANResponseTimeVCIC}) and D-DEMOS/Async (\ref{fig:LANResponseTimeVCAsync}), versus the number of \VC{} nodes, under a LAN setting. Election parameters are $n$ = 200,000 and $m$ = 4.} 
  \label{fig:LANResponseTimeVCAll}
}
\end{figure*}
In Figure~\ref{fig:LANResponseTimeVCAll}, we plot the average response time of both our vote collection protocols, versus the number of vote collectors, under different concurrency levels, ranging from 500 to 2000 concurrent clients.
Results for both systems illustrate an almost linear increase in the client-perceived latency, for all concurrency scenarios, up to 13 \VC{} nodes. 
From this point on, when four logical \VC{} nodes are placed on a single physical machine, we notice a non-linear increase in latency.
We attribute this to the overloading of the memory bus, a resource shared among all processors of the system, which services all (in-memory) database operations.
D-DEMOS/IC has a slower response time with its single round intra-\VC{} node communication, while D-DEMOS/Async is slightly slower due to the extra Uniqueness Certificate producing round. 

\begin{figure*}
\centering
{
  \subfloat[]
  {
    \includegraphics[width=0.50\textwidth]{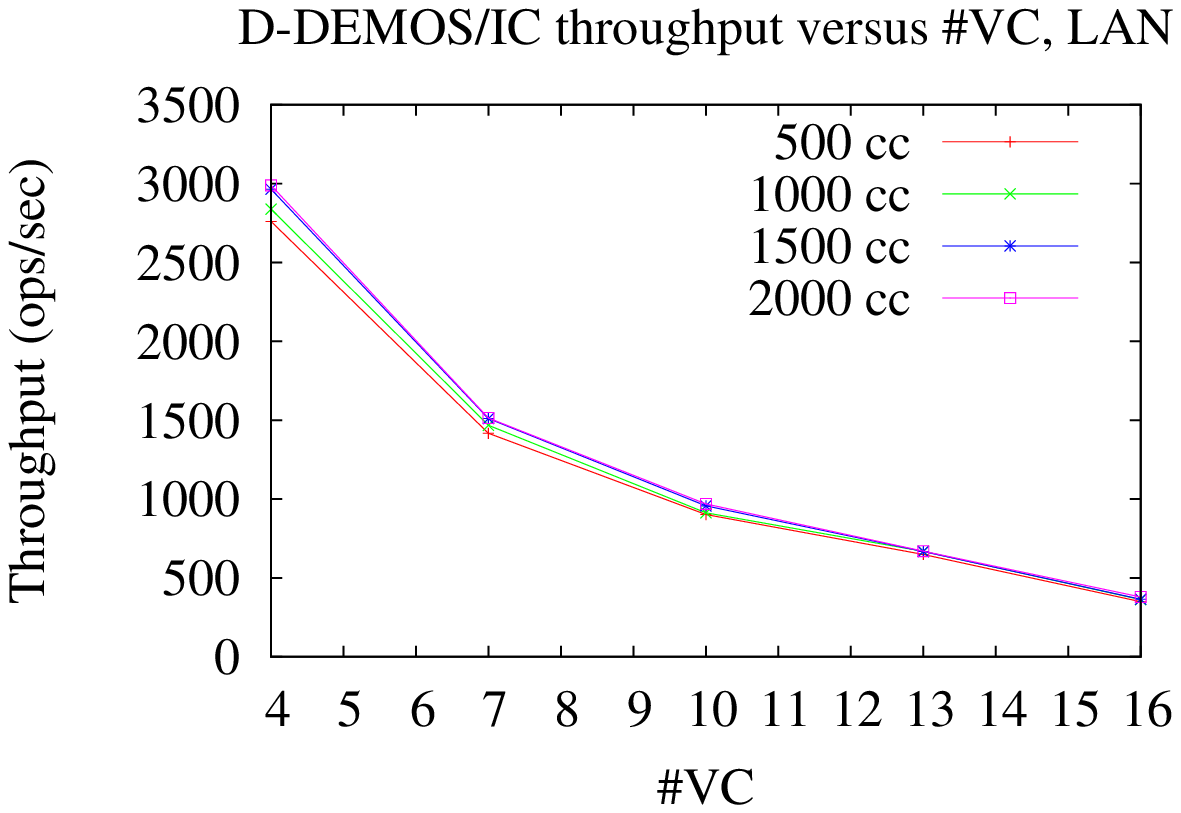}
    \label{fig:LANThroughputVCIC}
  }
  \subfloat[]
  {
    \includegraphics[width=0.50\textwidth]{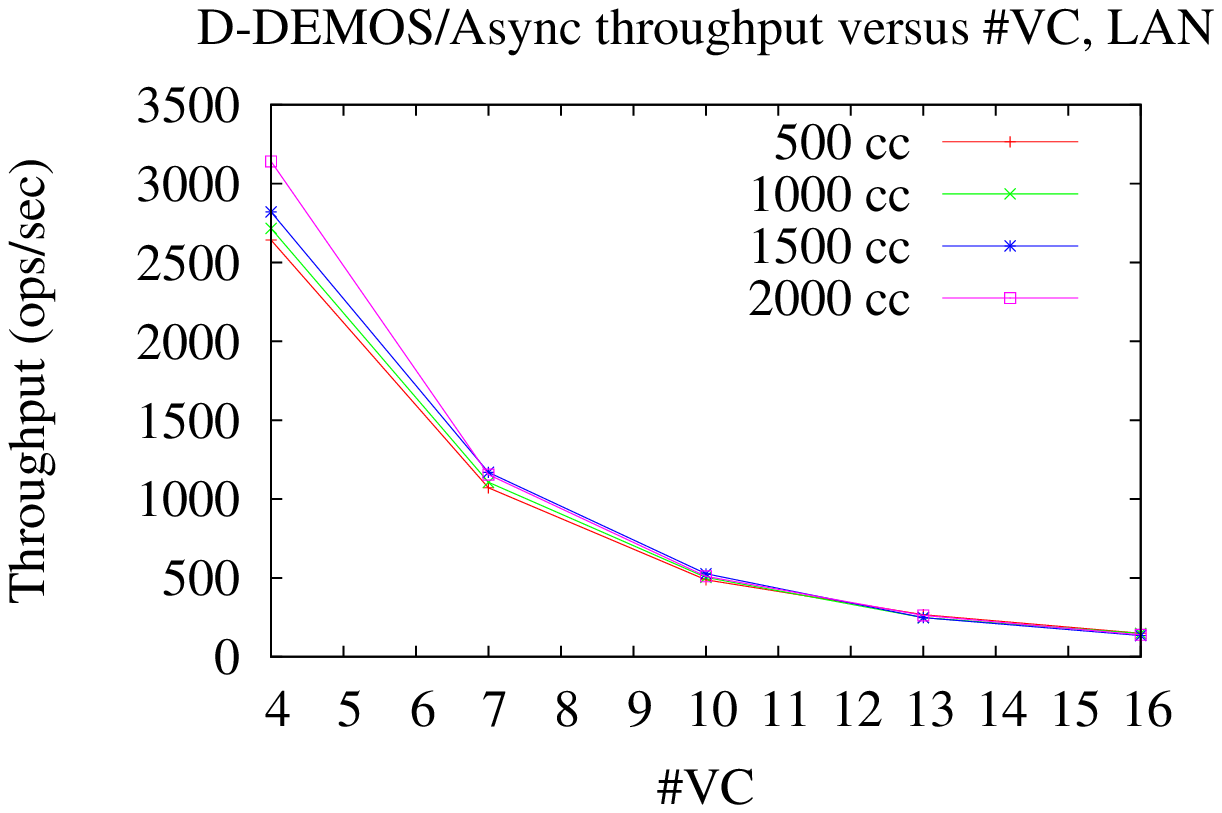}
    \label{fig:LANThroughputVCAsync}
  }
  \caption{Vote Collection throughput of D-DEMOS/IC (\ref{fig:LANThroughputVCIC}) and D-DEMOS/Async (\ref{fig:LANThroughputVCAsync}), versus the number of \VC{} nodes, under a LAN setting. Election parameters are $n$ = 200,000 and $m$ = 4.} 
  \label{fig:LANThroughputVCAll}
}
\end{figure*}
Figure~\ref{fig:LANThroughputVCAll} shows the throughput of both our vote collection protocols, versus the number of vote collectors, under different concurrency levels.
We observe that, in terms of overall system throughput, the penalty of tolerating extra failures (increasing the number of vote collectors) manifests early on. 
We notice an almost 50\% decline in system throughput from 4 to 7 \VC{} nodes for D-DEMOS/IC, and a bigger one for D-DEMOS/Async. 
However, further increases in the number of vote collectors lead to a much smoother, linear decrease. 
Overall, D-DEMOS/IC achieves better throughput than D-DEMOS/Async, due to exchanging fewer messages and lacking signature operations. 

\begin{figure*}
\centering
{
  \subfloat[]
  {
    \includegraphics[width=0.50\textwidth]{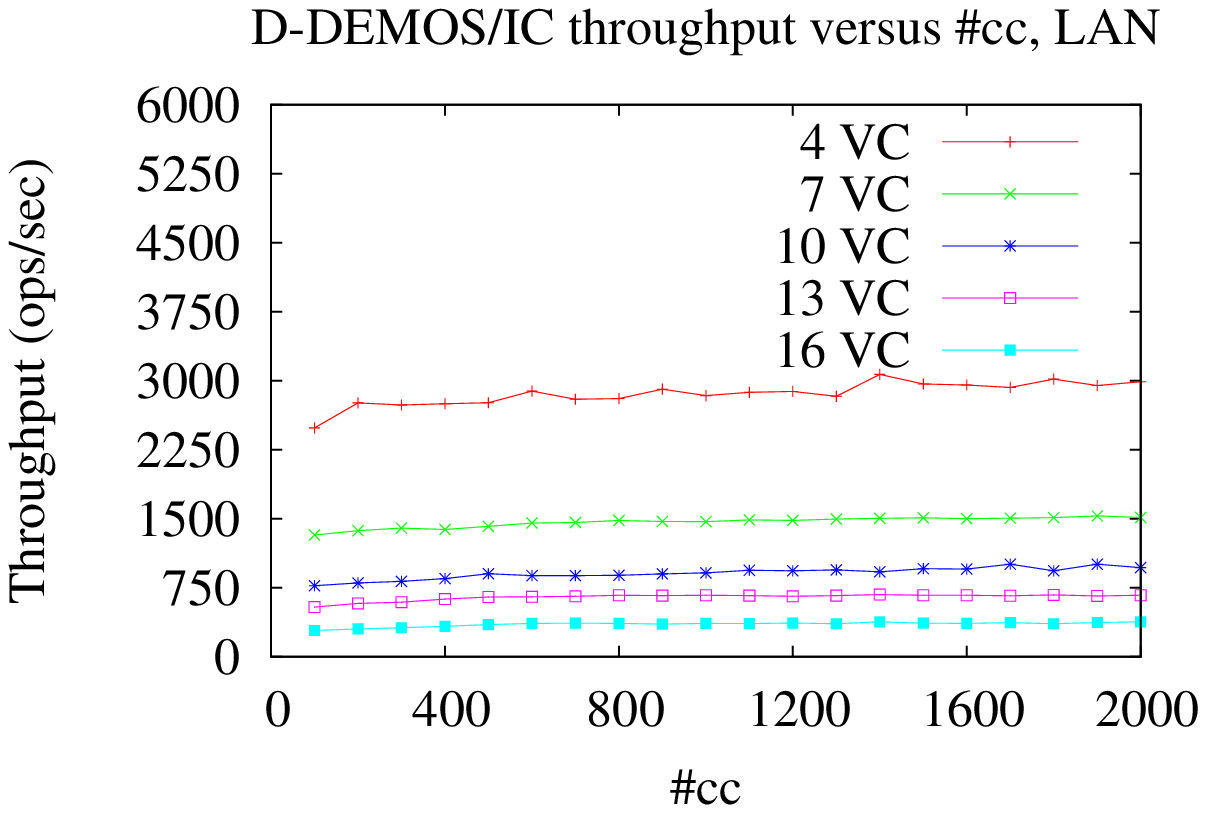}
    \label{fig:LANThroughputCCIC}
  }
  \subfloat[]
  {
    \includegraphics[width=0.50\textwidth]{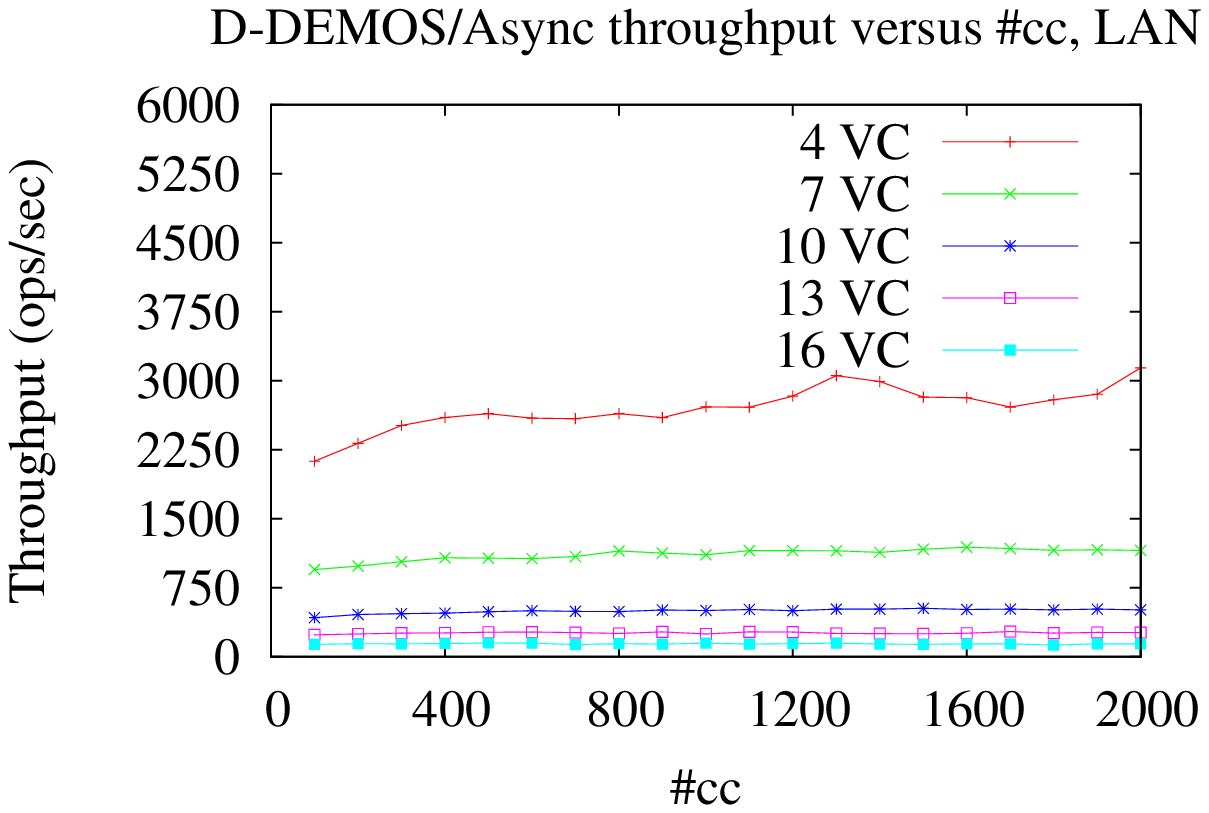}
    \label{fig:LANThroughputCCAsync}
  }
  \caption{Vote Collection throughput of D-DEMOS/IC (\ref{fig:LANThroughputCCIC}) and D-DEMOS/Async (\ref{fig:LANThroughputCCAsync}), versus the number of concurrent clients, under a LAN setting. Plots illustrate performance for different cardinalities of \VC{} nodes, thus different fault tolerance settings. Election parameters are $n$ = 200,000 and $m$ = 4.} 
  \label{fig:LANThroughputCCAll}
}
\end{figure*}
In Figure~\ref{fig:LANThroughputCCAll}, we plot a different view of both our systems' throughput, this time versus the concurrency level (ranging from 100 to 2000). 
Plots represent number of \VC{} node settings (4 to 16), thus different fault tolerance levels.
Results show both our systems have the nice property of delivering nearly constant throughput, regardless of the incoming request load, for a given number of \VC{} nodes.

\begin{figure*}
\centering
{
  \subfloat[]
  {
    \includegraphics[width=0.50\textwidth]{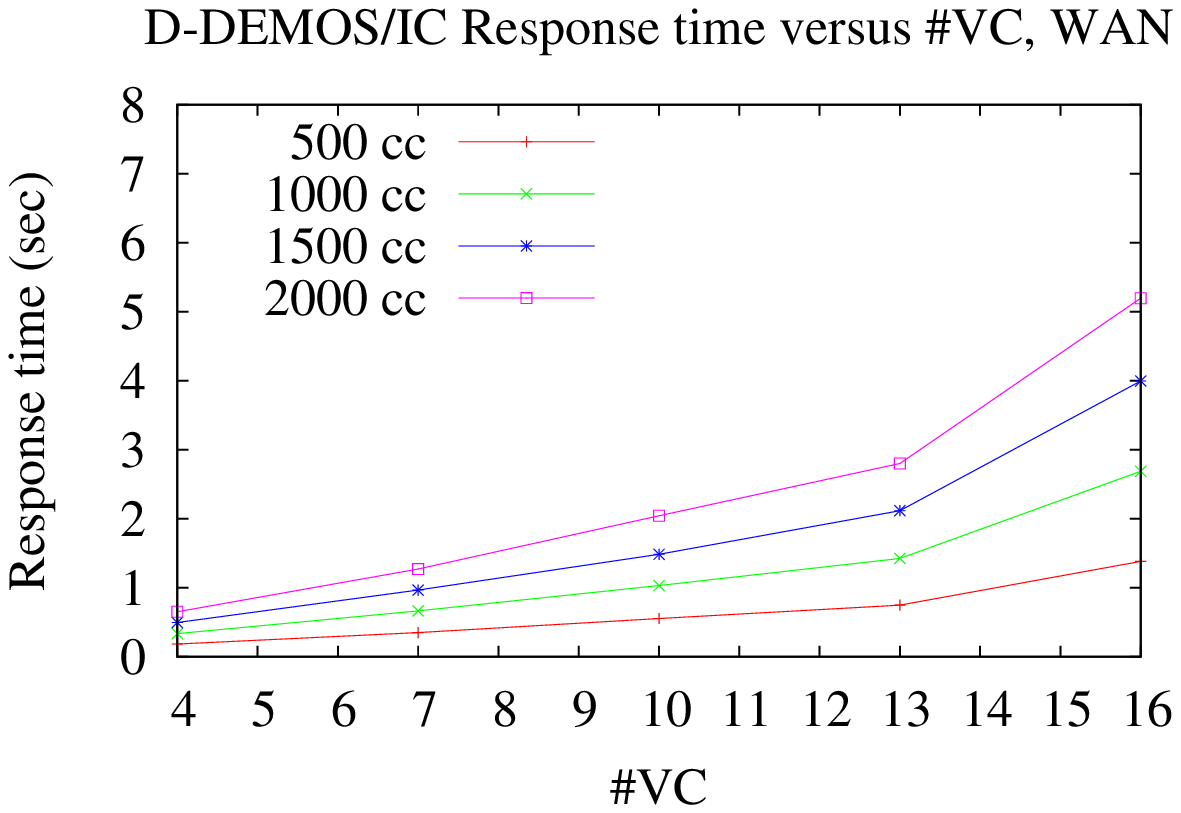}
    \label{fig:WANResponseTimeVCIC}
  }
  \subfloat[]
  {
    \includegraphics[width=0.50\textwidth]{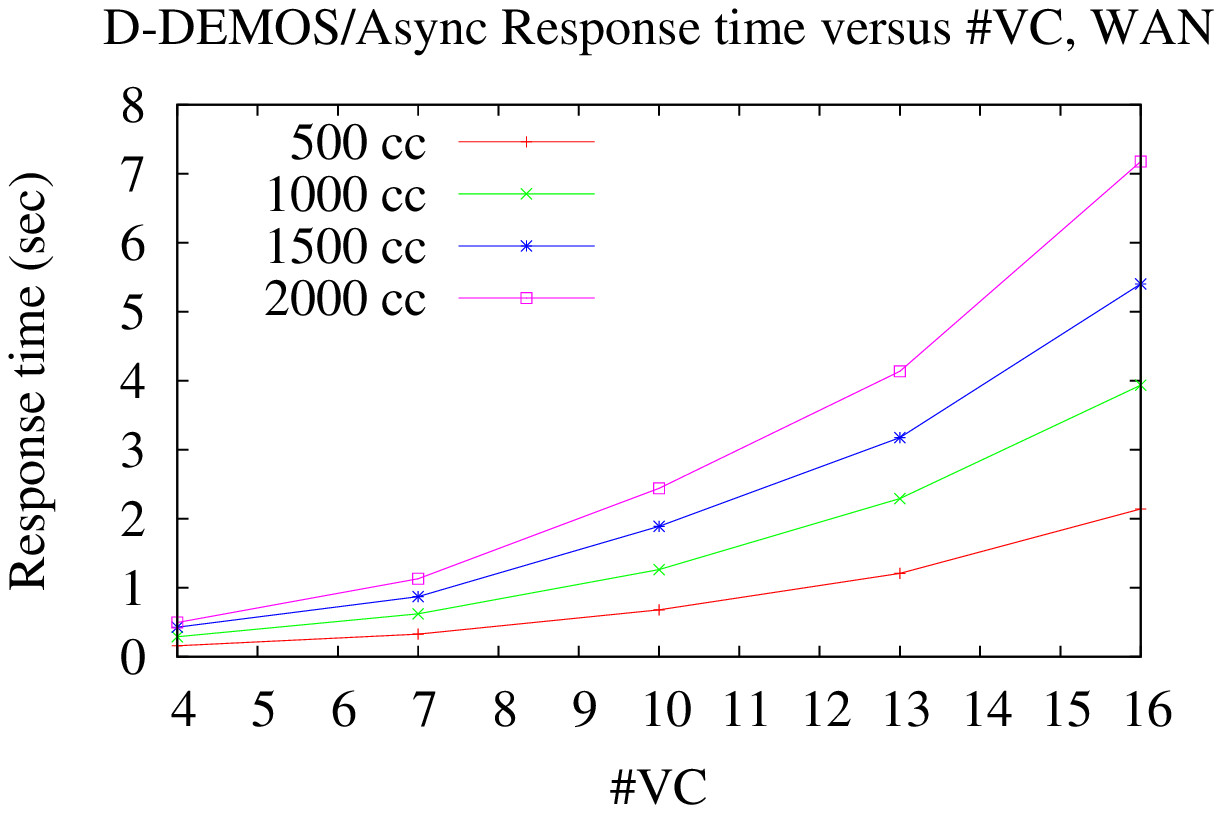}
    \label{fig:WANResponseTimeVCAsync}
  }
  \caption{Vote Collection response time of D-DEMOS/IC (\ref{fig:WANResponseTimeVCIC}) and D-DEMOS/Async (\ref{fig:WANResponseTimeVCAsync}), versus the number of \VC{} nodes, under a WAN setting. Election parameters are $n$ = 200,000 and $m$ = 4.} 
  \label{fig:WANResponseTimeVCAll}
}
\end{figure*}
\begin{figure*}
\centering
{
  \subfloat[]
  {
    \includegraphics[width=0.50\textwidth]{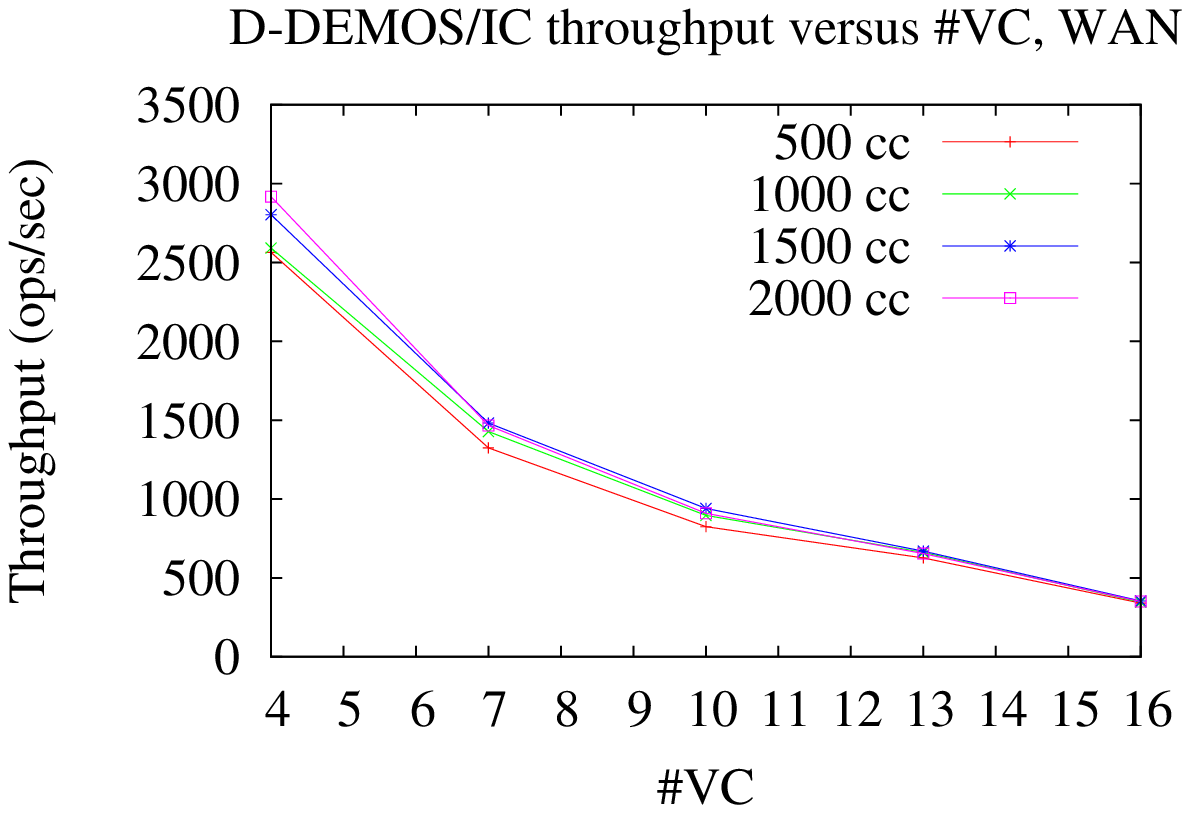}
    \label{fig:WANThroughputVCIC}
  }
  \subfloat[]
  {
    \includegraphics[width=0.50\textwidth]{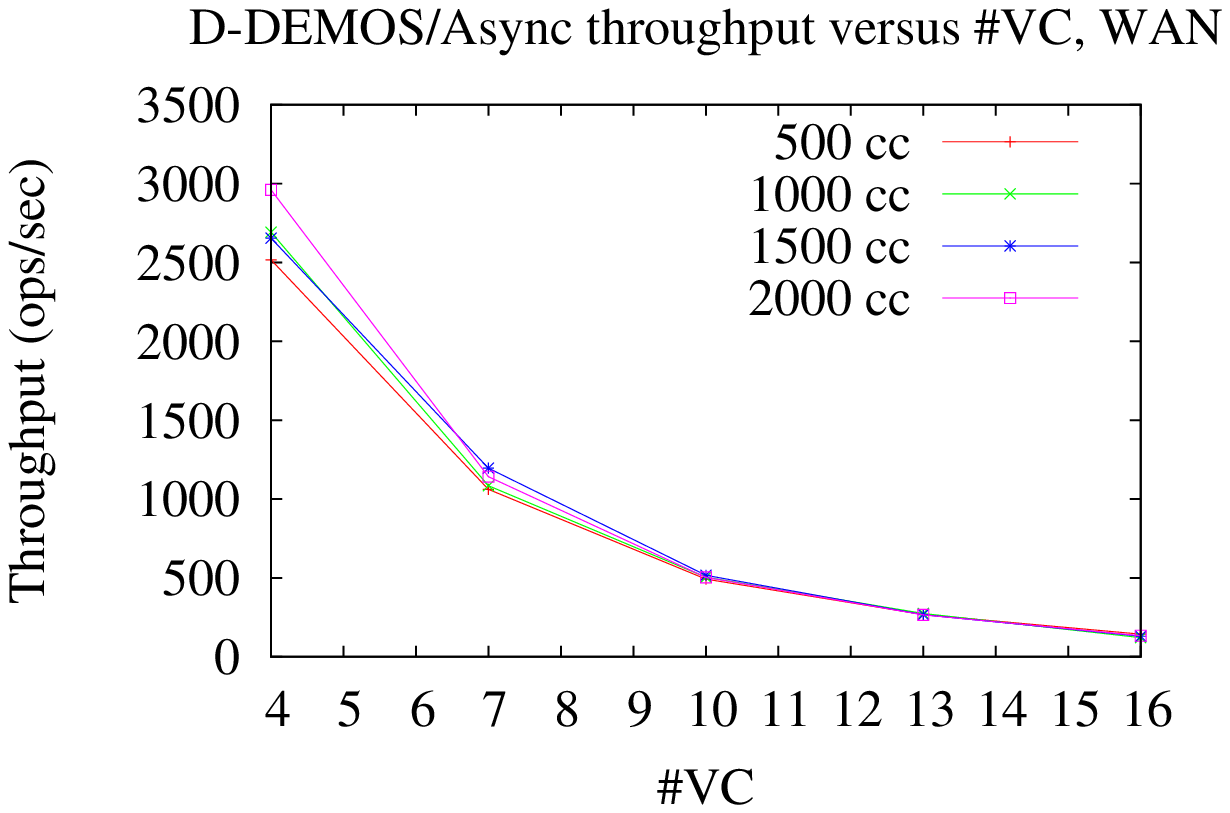}
    \label{fig:WANThroughputVCAsync}
  }
  \caption{Vote Collection throughput of D-DEMOS/IC (\ref{fig:WANThroughputVCIC}) and D-DEMOS/Async (\ref{fig:WANThroughputVCAsync}), versus the number of \VC{} nodes, under a WAN setting. Election parameters are $n$ = 200,000 and $m$ = 4.} 
  \label{fig:WANThroughputVCAll}
}
\end{figure*}
\begin{figure*}
\centering
{
  \subfloat[]
  {
    \includegraphics[width=0.50\textwidth]{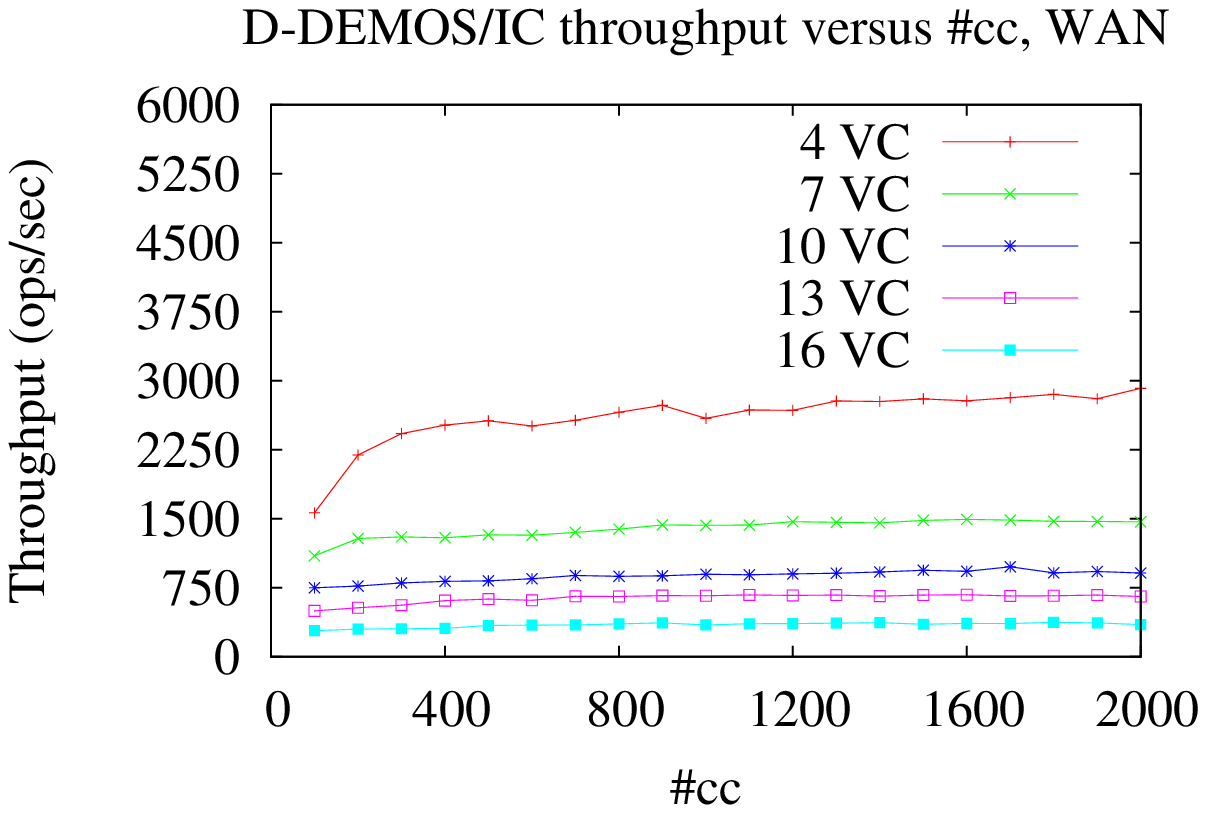}
    \label{fig:WANThroughputCCIC}
  }
  \subfloat[]
  {
    \includegraphics[width=0.50\textwidth]{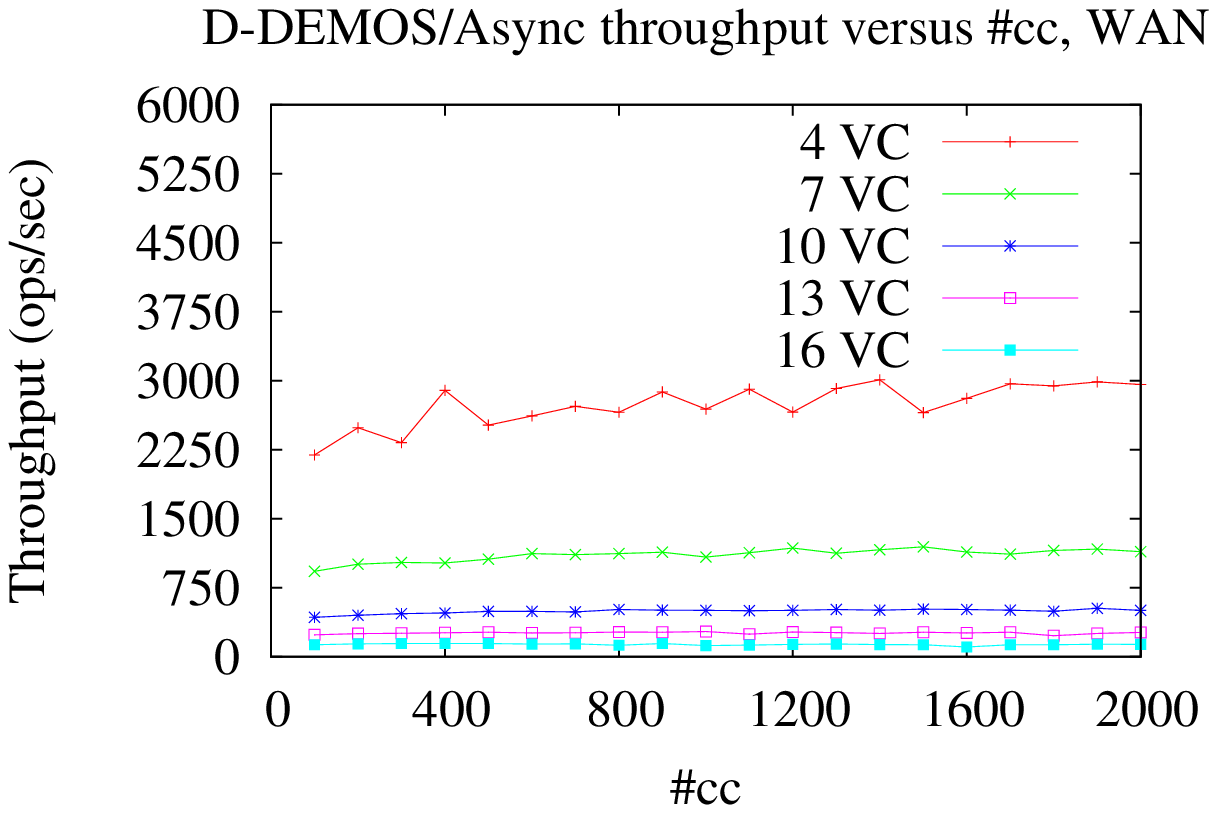}
    \label{fig:WANThroughputCCAsync}
  }
  \caption{Vote Collection throughput of D-DEMOS/IC (\ref{fig:WANThroughputCCIC}) and D-DEMOS/Async (\ref{fig:WANThroughputCCAsync}), versus the number of concurrent clients, under a WAN setting. Plots illustrate performance for different cardinalities of \VC{} nodes, thus different fault tolerance settings. Election parameters are $n$ = 200,000 and $m$ = 4.} 
  \label{fig:WANThroughputCCAll}
}
\end{figure*}

We repeat the same experiment by emulating a WAN environment using \emph{netem}~\cite{Netem}, a network emulator for Linux.
We inject a uniform latency of 25ms (typical for US coast-to-coast communication~\cite{coast2coast}) for each network packet exchanged between vote collector nodes, and present our results in Figures \ref{fig:WANResponseTimeVCAll}, \ref{fig:WANThroughputVCAll}, and \ref{fig:WANThroughputCCAll}.
A simple comparison between LAN and WAN plots illustrates our system manages to deliver the same level of throughput and average response time, regardless of the increased intra-\VC{} communication latency.

The benefits of the in memory approach, expressed both in terms of sub-second client (voter) response time and increased system throughput, make it an attractive alternative to the more standard database setup. 
For instance, in cases where high-end server machines are available, it would be possible to service mid to large scale elections completely from memory. 
We estimate the size of the in-memory representation of a $n=200K$ ballot election, with $m=4$ options, at approximately 322MB (see~\cite{ManeasThesis} for derivation details).
In this size, we include 64-bit Java pointers overhead, as we are using simple hash-maps of plain old Java classes. 
This size can be decreased considerably in a more elaborate implementation, where data is serialized by Google Protocol Buffers, for example.

\begin{figure*}
\centering
{
  \subfloat[]
  {
    \includegraphics[width=0.50\textwidth]{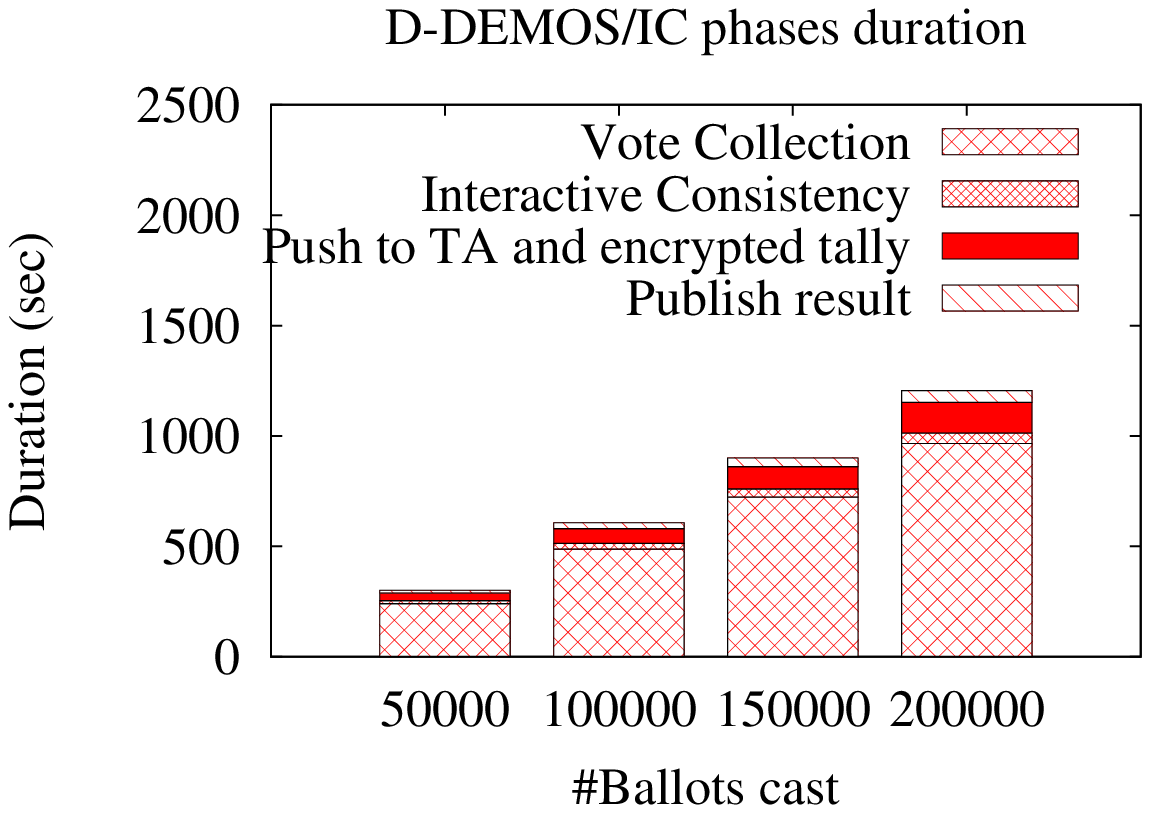}
    \label{fig:OverallLatencyBarIC}
  }
  \subfloat[]
  {
    \includegraphics[width=0.50\textwidth]{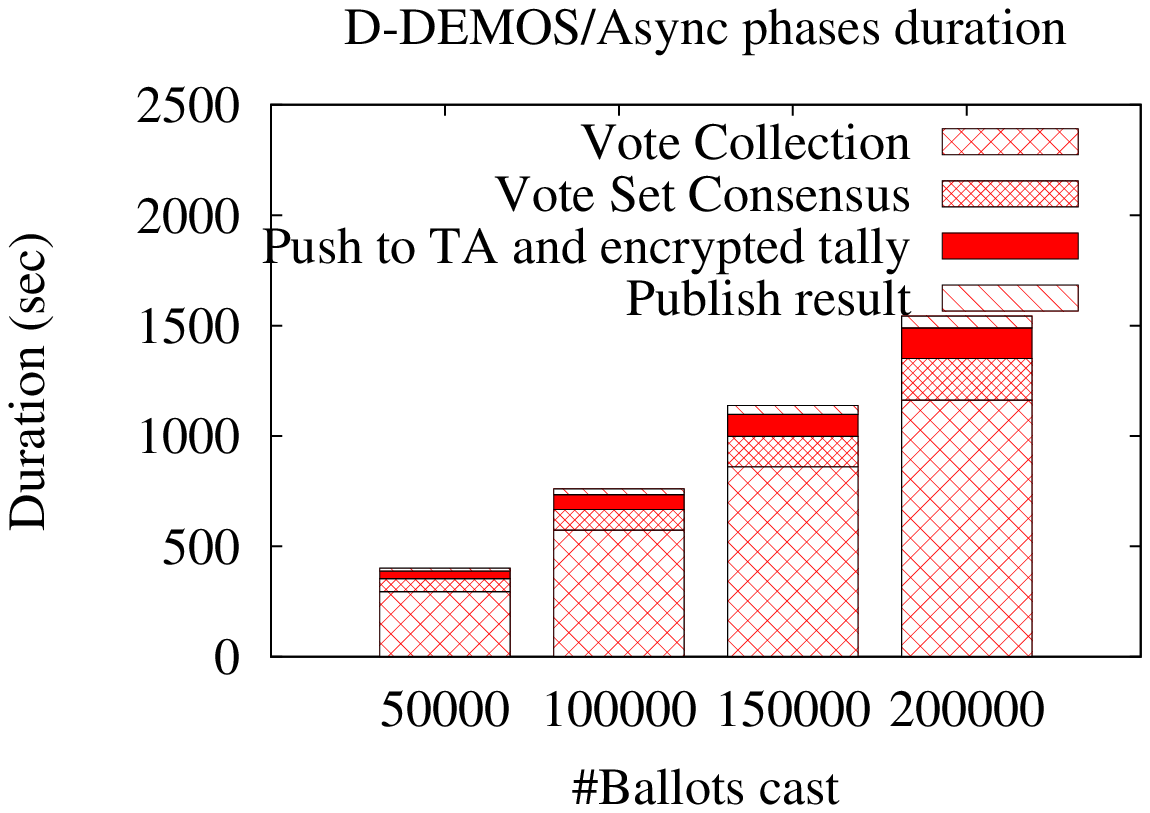}
    \label{fig:OverallLatencyBarAsync}
  }
  \caption{This figure illustrates the duration of all system phases. Results depicted are for 4 VCs, $n$ = 200,000 and $m$ = 4. All phases are disk based.} 
  \label{fig:OverallLatencyBarAll}
}
\end{figure*}

Finally, in Figure \ref{fig:OverallLatencyBarAll}, we illustrate a breakdown of the duration of each phase of the complete voting system (D-DEMOS/IC and D-DEMOS/Async), versus the total number of ballots cast. 
We assume immediate phase succession, i.e., the vote collection phase ends when all votes have been cast, at which point the vote set consensus phase starts, and so on. 
The ``Push to BB and encrypted tally'' phase is the time it takes for the vote collectors to push the final vote code set to the BB nodes, including all actions necessary by the BB to calculate and publish the encrypted result.  
The ``Publish result'' phase is the time it takes for \trustees{} to calculate and push their share of the opening of the  final tally to the BB, and for the BB to publish the final tally.   
Note that, in most voting procedures, the vote collection phase would in reality last several hours and even days as stipulated by national law (see Estonia voting system). 
Thus, looking only at the post-election phases of the system, we see that the time it takes to publish the final tally on the BB is quite fast. 
Comparing the two versions of D-DEMOS, we observe D-DEMOS/IC is faster during both Vote Collection and Vote Set Consensus phases.
This is expected, because of the extra communication round of D-DEMOS/Async during voting, as well as the more complex consensus-per-ballot approach to achieving Vote Set Consensus. 
However, D-DEMOS/Async is more robust than D-DEMOS/IC, as it does not require any kind of synchronization between nodes.

Overall, although we introduced Byzantine Fault Tolerance across all phases of a voting system (besides setup), we demonstrate it achieves high performance, enough to run real-life elections of large electorate bodies.

%% file: implementation.tex
\subsection{Implementation}
\par\noindent\textbf{\emph{Voting system:}}
We implement the Election Authority component of our system as a standalone C++ application, and all other components in Java.
Whenever we store data structures on disk, or transmit them over the network, we use Google Protocol Buffers~\cite{protobuf} to encode and decode them efficiently. 
We use the MIRACL library~\cite{MIRACL} for elliptic-curve cryptographic operations.
In all applications requiring a database, we use the PostgreSQL relational database system~\cite{PostgreSQL}.

We build an \emph{asynchronous communications stack} (ACS) on top of Java, using Netty~\cite{Netty} and the asynchronous PostgreSQL driver from~\cite{PGASYNC}, using TLS based authenticated channels for inter-node communication, and a public HTTP channel for public access.
This infrastructure uses connection-oriented sockets, but allows the applications running on the upper layers to operate in a message-oriented fashion.
We use this infrastructure to implement \VC{} and \BB{} nodes. 
We implement ``verifiable secret sharing with honest dealer'', by utilizing Shamir's Secret Share library implementation~\cite{JavaShamir}, and having the \EA{} sign each share.

For D-DEMOS/IC, we use the implementation of \emph{IC,BC-RBB} (Interactive Consistency algorithm, using asynchronous binary consensus and reliable broadcast without signatures) from~\cite{IC@ICPADS2015}. 
We use the election end time as a synchronization point to start the algorithm, and configure the timeout of the first phase of the algorithm according to the number of \VC{} nodes and the number of ballots in the election.

For D-DEMOS/Async, we implement Bracha's Binary Consensus directly on top of the ACS, and we use that to implement our Vote Set Consensus algorithm (depicted in Figure~\ref{fig:vcsasync}). 
We introduce a version of Binary Consensus that operates in batches of arbitrary size; this way, we achieve greater network efficiency.

Additionally, we batch most of the asynchronous vote set consensus ``announce'' phase's messages.
If this phase was implemented without optimization, it would result in a message complexity of $n*N_v$ (individual ANNOUNCE messages), imposing a significant network load.
This is because each node has to multicast an \ANNOUNCE{} message for each ballot, and wait for $n(N_v-f_v)$ replies to progress.
To optimize it, we have each node consult its local database and diagnose cases where another node already knows the correct vote code and $\mathsf{UCERT}$ for a specific ballot.
This is feasible because when a node $VC_b$ discloses its share using the $\mathsf{VOTE\_P}$ message, it also includes the $\mathsf{UCERT}$, and this fact is recorded in the recipient's node ($VC_a$) database along with the sender node's share.
For these cases, we produce $\mathsf{ANNOUNCE\_RANGE}$ messages addressed to individual nodes, having the source node $VC_a$ announce a range of ballot serial numbers as voted, a fact that is already known to the recipient node $VC_b$ (because $VC_a$ located recorded $\mathsf{VOTE\_P}$ messages from $VC_b$).
We use the same mechanism to announce ranges of not-voted ballots.
\\\par\noindent
\textbf{Trustee Android application:}
In addition to the web interface for \trustees{}, we also implement a specialized \Trustee{} Android application.
We re-use the MIRACL library on Android and provide a simple user interface for \trustees{}, where they use a single button press to perform each of their required tasks: download their initialization data from the \EA{}, download election data from the \BB{}, calculate their cryptographic contribution to the result opening, and finally upload their share of the opening to the \BB{}.
\\\par\noindent
\textbf{Web browser replicated service reader:} 
Our choice to model the Bulletin Board as a replicated service of non-cooperating nodes puts the burden of response verification on the reader of the service; a human reader is expected to manually issue a read request to all nodes, then compare the responses and pick the one posted by the majority of nodes. 
To alleviate this burden, we implement a web browser extension which automates this task, as an extension for Mozilla Firefox. The user sets up the list of URLs for the replicated service.  The add-on 1) intercepts any HTTP request towards any of these URLs, 2) issues the same request to the rest of the nodes, and 3) captures the responses, compares them in binary form, and routes the response coming from the majority, as a response to the original request posted by the user. Majority is defined by the number of defined URL prefixes; for 3 such URLs, the first 2 equal replies suffice. 
 
With the above approach, the user never sees a wrong reply, as it is filtered out by the extension. 
Also note this process will be repeated for all dependencies of the initial web page (images, scripts, CSS), as long at they come from the same source (with the same URL prefix), verifying the complete user visual experience in the browser.  

Note that, this mechanism is required only when reading data from the Bulletin Board, such as the election result, or audit information. This mechanism is neither needed nor used during voting, where the voter interacts with the Vote Collection subsystem using our voting protocol.

%% file: conclusion.tex
\section{Conclusion and future work}\label{section:conclusion}
We have presented the world's first suite of state-of-the-art, end-to-end verifiable, distributed internet voting systems with no single point of failure besides setup. 
Both systems allow voters to verify their vote was tallied-as-intended without the assistance of special software or trusted devices, and external auditors to verify the correctness of the election process. 
Additionally, the systems allows voters to delegate auditing to a third party auditor, without sacrificing their privacy.
We have provided a model and security analysis of both voting systems.
Finally, we have implemented prototypes of the integrated systems, measured their performance, and demonstrated their ability to handle large-scale elections.  

We have used our system to conduct exit polls at three large voting sites for two national-level elections.  
We look forward to gaining more experience and feedback about our systems by exploring their use in election and decision-making procedures at all levels throughout 
the Greek university system, and studying their adoption for use by the General Confederation of Greek Workers, the largest civil union of workers in Greece.  
Finally, our systems currently support only \emph{1-out-of-m} elections, in which voters choose one out of m options from their ballots. 
As future work, we will expand our systems to support \emph{k-out-of-m} elections.

%% file: acknowledgement.tex
\ifextended
\\
\fi
\\\noindent
\textbf{Acknowledgements:}
This work was partially supported by ERC Starting Grant \# 279237 and by the FINER project funded by the Greek Secretariat of Research and Technology under action "ARISTEIA 1".
\ifextended
\\
\fi
We thank Vasileios Poulimenos for his effort in developing the Android application for the trustees interface.
\ifextended
\\
\fi
A preliminary version of this work appears in IEEE ICDCS 2016~\cite{DDEMOS@ICDCS2016}.

%% file: security.tex
\section{Security of D-Demos}\label{sec:security_full}
In this section, we present at length the security properties that D-DEMOS achieves.
Specifically, we show that D-DEMOS/IC
\ifextended
and D-DEMOS/Async achieve 
\else
achieves 
\fi
liveness and safety, according to which every voter that submits her vote prior to a well-defined time threshold, will obtain a valid receipt (liveness) and her vote will be included in the election tally and published in the BB (safety contract). 
 In addition, 
\ifextended
 both versions achieve 
\else
 the system achieves
\fi 
 end-to-end verifiability and voter privacy at the same level as~\cite{DEMOS}\footnote{In~\cite{DEMOS}, the authors use the term \emph{voter privacy/receipt-freeness}, but they actually refer to the same property.}, thus allowing a top-tier integrity guarantee without compromising secrecy.
\ifextended
\else
Due to lack of space, we omit the proofs specific to D-DEMOS/Async, and provide only a sketch of the end-to-end verifiability proof, which applies to both systems. 
We refer the interested reader to the extended version of this paper~\cite{extended} for the full proofs for both systems.
\fi
\smallskip\par
We use $m$, $n$ to denote the number of options and voters respectively. We denote by $\lambda$ the cryptographic security parameter and we write $\mathsf{negl}(\lambda)$ to denote that a function is negligible in $\lambda$, i.e., it is asymptotically smaller than the inverse of any polynomial in  $\lambda$. 
\smallskip\par
The remaining sections reference heavily the Cryptographic Tools section (\ref{subsec:sec_tools}), which includes the notions and claims about the security of the cryptographic tools we use in the two versions of D-DEMOS. 

%
%
 %
\subsection{Liveness}\label{subsec:sec_liveness}
%
\input{liveness}

\subsection{Safety}\label{subsec:sec_safety}
%
\input{safety}
\subsection{End-to-end Verifiability}\label{subsec:sec_e2e}
\input{E2E}

%
\subsection{Voter Privacy}\label{subsec:sec_priv}
%
\input{privacy}

%% file: liveness.tex

To prove the liveness that D-DEMOS guarantees, we assume (i) an upper bound $\delta$ on the delay of the delivery of messages and (ii) an upper bound $\Delta$ on the drift of all clocks (see Assumptions~\textbf{B} and \textbf{C} in Section~\ref{subsec:threat}). Furthermore, to express liveness rigorously, we formalize the behavior of honest voters regarding maximum waiting before vote resubmission as follows: 
\begin{definition}[$\mathrm{[}d\mathrm{]}$-\textsc{\textbf{Patience}}]\label{dfnt:patient}
Let $V$ be an honest voter that submits her vote at some VC node when $\Cl[V]=T$. We say that $V$ is $[d]$-\emph{patient}, when the following condition holds:
\noindent If $V$ does not obtain a valid receipt by the time that $\Cl[V]=T+d$, then she will blacklist this VC node and submit the same vote to another randomly selected VC node.

\end{definition}
%
%
\subsubsection{Liveness of D-DEMOS/IC}\label{subsubsec:liveveness_IC}
\par
\noindent Using Definition~\ref{dfnt:patient}, we prove the liveness of D-DEMOS/IC in the following theorem. A crucial step in the proof, is to compute an upper bound on the time required for an honest responder $VC$ node to issue a receipt to $V$. This bound will be derived by the time upper bounds that correspond to each step of the voting protocol, as described in Sections~\ref{subsec:VcnodesIC} and~\ref{subsec:voter}, taking also into account the $\Delta,\delta$ upper bounds.
\ifextended
In Fig.~\ref{tab:liveness-IC}, we provide upper bounds on the advance of the global clock and the internal clocks of $V$ and the VC nodes, so that we illustrate the description of the computation described below.
 \begin{figure}[ht]
 \footnotesize
\begin{center}
 \begin{tabular}{|M{2.7cm}||M{2.2cm}|M{2.2cm}|M{2.2cm}|M{2.2cm}|}
 \hline
\multirow{2}{*}{\small{\textbf{Step}}}&\multicolumn{4}{c|}{\small{\textbf{Time upper bounds at each clock}}}\\\cline{2-5}
&\small{$\Cl$}&\small{$\Cl[V]$}&\small{$\Cl[VC]$}&\small{honest VC nodes' clocks}\\
\hline\hline
$V$ is initialized & \cellcolor{cellgray}$T$ & $T$ & $T+\Delta$& $T+\Delta$\\
\hline
$V$ submits her vote to $VC$ & $ T+T_\mathsf{comp}+\Delta$ & \cellcolor{cellgray}$T+T_\mathsf{comp}$ & $T+T_\mathsf{comp}+2\Delta$& $T+T_\mathsf{comp}+2\Delta$\\
\hline
$VC$ receives $V$'s ballot & $ \cellcolor{cellgray}T+T_\mathsf{comp}+\Delta+\delta$ & $T+T_\mathsf{comp}+2\Delta+\delta$ & $T+T_\mathsf{comp}+2\Delta+\delta$&$T+T_\mathsf{comp}+2\Delta+\delta$\\
\hline
$VC$ verifies the validity of $V$'s ballot and broadcasts its share & $ T+2T_\mathsf{comp}+3\Delta+\delta$ & $T+2T_\mathsf{comp}+4\Delta+\delta$ & \cellcolor{cellgray}$T+2T_\mathsf{comp}+2\Delta+\delta$&$T+2T_\mathsf{comp}+4\Delta+\delta$\\
\hline
All the other honest VC nodes receive $VC$'s share& \cellcolor{cellgray}$T+2T_\mathsf{comp}+3\Delta+2\delta$ & $T+2T_\mathsf{comp}+4\Delta+2\delta$ & $T+2T_\mathsf{comp}+4\Delta+2\delta$&$T+2T_\mathsf{comp}+4\Delta+2\delta$\\
\hline
All the other honest VC nodes verify the validity of $V$'s share and broadcast their shares &$T+3T_\mathsf{comp}+5\Delta+2\delta$ & $T+3T_\mathsf{comp}+6\Delta+2\delta$ & $T+3T_\mathsf{comp}+6\Delta+2\delta$& \cellcolor{cellgray}$T+3T_\mathsf{comp}+4\Delta+2\delta$\\
\hline
$VC$ receives all the $N_v-1$ other honest VC nodes' shares & \cellcolor{cellgray}$T+3T_\mathsf{comp}+5\Delta+3\delta$ & $T+3T_\mathsf{comp}+6\Delta+3\delta$ & $T+3T_\mathsf{comp}+6\Delta+3\delta$&$T+3T_\mathsf{comp}+6\Delta+3\delta$\\
\hline
$VC$ verifies the validity of all the $N_v-1$ other honest VC nodes' shares & $T+(N_v+2)T_\mathsf{comp}+7\Delta+3\delta$ & $T+(N_v+2)T_\mathsf{comp}+8\Delta+3\delta$ & 
\cellcolor{cellgray}$T+(N_v+2)T_\mathsf{comp}+6\Delta+3\delta$&$T+(N_v+2)T_\mathsf{comp}+8\Delta+3\delta$\\
\hline
$VC$ reconstructs and sends $V$'s receipt & $T+(N_v+3)T_\mathsf{comp}+7\Delta+3\delta$ & $T+(N_v+3)T_\mathsf{comp}+8\Delta+3\delta$ & 
\cellcolor{cellgray}$T+(N_v+3)T_\mathsf{comp}+6\Delta+3\delta$&$T+(N_v+3)T_\mathsf{comp}+8\Delta+3\delta$\\
\hline
$V$ obtains her receipt & \cellcolor{cellgray}$T+(N_v+3)T_\mathsf{comp}+7\Delta+4\delta$ & $T+(N_v+3)T_\mathsf{comp}+8\Delta+4\delta$ & 
$T+(N_v+3)T_\mathsf{comp}+8\Delta+4\delta$&$T+(N_v+3)T_\mathsf{comp}+8\Delta+4\delta$\\
\hline
$V$ verifies the validity of her receipt &$T+(N_v+4)T_\mathsf{comp}+7\Delta+4\delta$ &\cellcolor{cellgray} $T+(N_v+4)T_\mathsf{comp}+8\Delta+4\delta$ & 
$T+(N_v+4)T_\mathsf{comp}+8\Delta+4\delta$&$T+(N_v+4)T_\mathsf{comp}+8\Delta+4\delta$\\
\hline
 \end{tabular}\end{center}
\caption{Time upper bounds at $\Cl,\Cl[V]$, $\Cl[VC]$ and other honest VC nodes' clocks at each step of the interaction of the voter $V$ with responder $VC$ during D-DEMOS/IC voting phase. The grayed cells indicate the reference point of the clock drifts at each step.}
\label{tab:liveness-IC}
 \end{figure}
\fi
\begin{theorem}[\textbf{Liveness of D-Demos/IC}]\label{thm:liveness_IC}
Consider a D-DEMOS/IC run with $n$ voters, $m$ options and $N_v$ VC nodes. Let $\A$ be an adversary against D-DEMOS/IC under the model described in Section~\ref{subsec:threat} that corrupts up to $f_v<N_v/3$ VC nodes.  Assume there is an upper bound $\Delta$ on clock synchronization loss and an upper bound $\delta$ on the delay of message delivery among honest VC nodes. Let $T_\mathsf{comp}$ be the worst-case running time of any procedure run by the VC nodes and the voters described in Sections~\ref{subsec:VcnodesIC} and~\ref{subsec:voter} respectively, during the voting protocol.
 \par Let $T_\mathsf{end}$ denote the election end time. Define
 \[T_\mathsf{wait}:=(N_v+4)T_\mathsf{comp}+8\Delta+4\delta\;.\] 
 Then, the following conditions hold:
 \begin{enumerate}
   \item\label{liveness-1-IC} Every $[T_\mathsf{wait}]$-patient voter $V$ that is engaged in the voting protocol by the time that $\Cl[V]=T_\mathsf{end}-(f_v+1)\cdot T_\mathsf{wait}$, will obtain a valid receipt.
   \item\label{liveness-2-IC} Every $[T_\mathsf{wait}]$-patient voter $V$ that is engaged in the voting protocol by the time that  
  $\Cl[V]=T_\mathsf{end}-y\cdot T_\mathsf{wait}$, where $y\in[f_v]$,
  will obtain a valid receipt with more than $1-3^{-y}$ probability.
 \end{enumerate}

\end{theorem}

\begin{proof}
 Let $V$ be a $[T_\mathsf{wait}]$-patient voter initialized by the adversary $\A$ when $\Cl=\Cl[V]=T$.  Upon initialization, $V$'s internal clock is synchronized with the global clock at time $\Cl=\Cl[V]=T$. After at most $T_\mathsf{comp}$ steps, $V$  submits her vote $(\mathsf{serial\textrm{-}no},\mathsf{vote\textrm{-}code})$ at internal clock time: $\Cl[V]=T+T_\mathsf{comp}$, hence at global clock time: $\Cl\leq T+\Delta$. Thus, $VC$ will receive the vote of $V$ at internal time $\Cl[VC]\leq (T+T_\mathsf{comp})+2\Delta+\delta$.  Then, $VC$ performs at most $T_\mathsf{comp}$ steps to verify the validity of the vote before it broadcasts its receipt share.
\par All the other honest VC nodes will receive $VC$'s receipt share by global clock time: 
\[\Cl\leq (T+T_\mathsf{comp}+2\Delta+\delta)+(T_\mathsf{comp}+\Delta+\delta)=T+2T_\mathsf{comp}+3\Delta+2\delta,\]
which implies that the time at their internal clocks is at most
  $T+2T_\mathsf{comp}+4\Delta+2\delta$.
  Then, they will verify $VC$'s share and broadcast their shares for $V$'s vote after at most $T_\mathsf{comp}$ steps. The global clock at that point is no more than
 \[\Cl\leq(T+2T_\mathsf{comp}+4\Delta+2\delta)+T_\mathsf{comp}+\Delta=T+3T_\mathsf{comp}+5\Delta+2\delta.\]
Therefore, $VC$ will obtain the other honest VC nodes' shares at most when
 \[\Cl[VC]\leq(T+3T_\mathsf{comp}+5\Delta+2\delta)+\Delta+\delta=T+3T_\mathsf{comp}+6\Delta+3\delta\]
 and will process them in order to reconstruct the receipt for $V$. In order to collect $N_v-f_v-1$ receipt shares that are sufficient for reconstruction, $VC$ may have to perform up to $N_v-1$ receipt-share verifications, as the $f_v$ malicious VC nodes may also send invalid messages.
This verification requires at most $(N_v-1)\cdot T_\mathsf{comp}$ steps. Taking into account the $T_\mathsf{comp}$ steps for the reconstruction process, we conclude that $VC$ will finish computation by global time
 \begin{equation*}
\begin{split}
 =(T+3T_\mathsf{comp}+6\Delta+3\delta)+(N_v-1)T_\mathsf{comp}+T_\mathsf{comp}+\Delta
 =T+(N_v+3)T_\mathsf{comp}+7\Delta+3\delta.
\end{split}
\end{equation*}
Finally, $V$ will obtain the receipt after at most $\delta$ delay from the moment that $VC$ finishes computation, and she needs $T_\mathsf{comp}$ steps to verify the validity of this receipt. Taking into consideration the drift on $V$'s internal clock, we have that if $V$ is honest and  has not yet obtained a receipt by the time that 
   \begin{equation*}
\begin{split}\Cl[V]&=\big(T+(N_v+3)T_\mathsf{comp}+7\Delta+3\delta\big)+T_\mathsf{comp}+\Delta+\delta=T+T_\mathsf{wait},\\
\end{split}
\end{equation*}
 \noindent then, being $[T_\mathsf{wait}]$-patient, she can blacklist $VC$ and resubmit her vote to another VC node. We will show that the latter fact implies conditions (\ref{liveness-1-IC}) and (\ref{liveness-2-IC}) in the statement of the theorem:\\
\par\textbf{Condition (\ref{liveness-1-IC}): }since there are at most $f_v$ malicious VC nodes, $V$ will certainly run into an honest VC node at her $(f_v+1)$-th attempt (if reached). Therefore, if $V$ is engaged in the voting protocol by the time that $\Cl[V]=T_\mathsf{end}-(f_v+1)\cdot T_\mathsf{wait}$, then she will obtain a receipt.\\
\par\textbf{Condition (\ref{liveness-2-IC}): } if $V$ has waited for more than $y\cdot T_\mathsf{wait}$ time and has not yet received a receipt, then it has run at least $y$ failed attempts in a row. 
  At the $j$-th attempt, $V$ has $\dfrac{f_v-(j-1)}{N_v-(j-1)}$ probability to randomly select one of the remaining $f_v-(j-1)$ malicious VC nodes out of the $N_v-(j-1)$ non-blacklisted VC nodes. Thus,
 the probability that $V$ runs at least $y$ failed attempts in a row is
  \begin{equation*}
\begin{split}
 \prod_{j=1}^{y}\dfrac{f_v-(j-1)}{N_v-(j-1)}&= \prod_{j=1}^{y}\dfrac{f_v-(j-1)}{3f_v+1-(j-1)}<3^{-y}.
\end{split}
\end{equation*}  
Therefore, if $V$ is engaged in the voting protocol by the time that  $\Cl[V]=T_\mathsf{end}-y\cdot T_\mathsf{wait}$, then the probability that she will obtain a receipt is more than $1-3^{-y}$.
\end{proof}
%
%
%
\ifextended
\subsubsection{Liveness of D-DEMOS/Async}\label{subsubsec:liveveness_async}
The proof of liveness in the asynchronous version of D-DEMOS differs from the one of D-DEMOS/IC in the computation of the $T_\mathsf{wait}$ upper bound, which now depends on the steps of the VC nodes presented in Section~\ref{subsec:VcnodesAsync}. 
The upper bounds on the advance of the the global clock and the internal clocks of $V$ and the VC nodes is analogously differentiated, as depicted in Fig.~\ref{tab:liveness-async}.
%
%
 \begin{figure}[ht]
 \footnotesize
\begin{center}
 \begin{tabular}{|M{3.1cm}||M{2.2cm}|M{2.2cm}|M{2.2cm}|M{2.2cm}|}
 \hline
\multirow{3}{*}{\small{\textbf{Step}}}&\multicolumn{4}{c|}{\small{\textbf{Time upper bounds at each clock}}}\\\cline{2-5}
&\small{$\Cl$}&\small{$\Cl[V]$}&\small{$\Cl[VC]$}&\small{honest VC nodes' clocks}\\
\hline\hline
$V$ is initialized & \cellcolor{cellgray}$T$ & $T$ & $T+\Delta$& $T+\Delta$\\
\hline
$V$ submits her vote to $VC$ & $ T+T_\mathsf{comp}+\Delta$ & \cellcolor{cellgray}$T+T_\mathsf{comp}$ & $T+T_\mathsf{comp}+2\Delta$&$T+T_\mathsf{comp}+2\Delta$\\
\hline
$VC$ receives $V$'s ballot & $ \cellcolor{cellgray}T+T_\mathsf{comp}+\Delta+\delta$ & $T+T_\mathsf{comp}+2\Delta+\delta$ & $T+T_\mathsf{comp}+2\Delta+\delta$&$T+T_\mathsf{comp}+2\Delta+\delta$\\
\hline
$VC$ verifies the validity of $V$'s ballot and broadcasts an ENDORSE message & $ T+2T_\mathsf{comp}+3\Delta+\delta$ & $T+2T_\mathsf{comp}+4\Delta+\delta$ & \cellcolor{cellgray}$T+2T_\mathsf{comp}+2\Delta+\delta$&$T+2T_\mathsf{comp}+4\Delta+\delta$\\
\hline
All the other honest VC nodes receive $VC$'s ENDORSE message& \cellcolor{cellgray}$T+2T_\mathsf{comp}+3\Delta+2\delta$ & $T+2T_\mathsf{comp}+4\Delta+2\delta$ & $T+2T_\mathsf{comp}+4\Delta+2\delta$&$T+2T_\mathsf{comp}+4\Delta+2\delta$\\
\hline
All the other honest VC nodes verify the validity of the ENDORSE message and respond with an ENDORSEMENT message & $T+3T_\mathsf{comp}+5\Delta+2\delta$ & $T+3T_\mathsf{comp}+6\Delta+2\delta$ & $T+3T_\mathsf{comp}+6\Delta+2\delta$&\cellcolor{cellgray}$T+3T_\mathsf{comp}+4\Delta+\delta$\\
\hline
$VC$ receives the ENDORSEMENT messages of all the other honest VC nodes &\cellcolor{cellgray}$T+3T_\mathsf{comp}+5\Delta+3\delta$ & $T+3T_\mathsf{comp}+6\Delta+3\delta$ & $T+3T_\mathsf{comp}+6\Delta+3\delta$&$T+3T_\mathsf{comp}+6\Delta+3\delta$\\
\hline
$VC$ verifies the validity of all the $N_v-1$ received messages until it obtains $N_v-f_v$ valid ENDORSEMENT messages & $ T+(N_v+2)T_\mathsf{comp}+7\Delta+3\delta$ & $T+(N_v+2)T_\mathsf{comp}+8\Delta+3\delta$ & \cellcolor{cellgray}$T+(N_v+2)T_\mathsf{comp}+6\Delta+3\delta$& $T+(N_v+2)T_\mathsf{comp}+8\Delta+3\delta$\\
\hline
$VC$ forms UCERT certificate and broadcsts its share and UCERT &$ T+(N_v+3)T_\mathsf{comp}+7\Delta+3\delta$ & $T+(N_v+3)T_\mathsf{comp}+8\Delta+3\delta$ & \cellcolor{cellgray}$T+(N_v+3)T_\mathsf{comp}+6\Delta+3\delta$& $T+(N_v+3)T_\mathsf{comp}+8\Delta+3\delta$\\
\hline
All the other honest VC nodes receive $VC$'s broadcast share and UCERT& \cellcolor{cellgray}$ T+(N_v+3)T_\mathsf{comp}+7\Delta+4\delta$ & $T+(N_v+3)T_\mathsf{comp}+8\Delta+4\delta$ &$T+(N_v+3)T_\mathsf{comp}+8\Delta+4\delta$& $T+(N_v+3)T_\mathsf{comp}+8\Delta+4\delta$\\
\hline
All the other honest VC nodes verify the validity of UCERT and $V$'s share and broadcast their shares &$ T+(N_v+4)T_\mathsf{comp}+9\Delta+4\delta$ & $T+(N_v+4)T_\mathsf{comp}+10\Delta+4\delta$ &$T+(N_v+4)T_\mathsf{comp}+10\Delta+4\delta$&\cellcolor{cellgray} $T+(N_v+4)T_\mathsf{comp}+8\Delta+4\delta$\\
\hline
$VC$ receives all the other honest VC nodes' shares &  \cellcolor{cellgray}$ T+(N_v+4)T_\mathsf{comp}+9\Delta+5\delta$ & $T+(N_v+4)T_\mathsf{comp}+10\Delta+5\delta$ &$T+(N_v+4)T_\mathsf{comp}+10\Delta+5\delta$&$T+(N_v+4)T_\mathsf{comp}+10\Delta+5\delta$\\
\hline
$VC$ verifies the validity of all the $N_v-1$ received messages until it obtains $N_v-f_v$ valid shares &$ T+(2N_v+3)T_\mathsf{comp}+11\Delta+5\delta$ & $T+(2N_v+3)T_\mathsf{comp}+12\Delta+5\delta$ &  \cellcolor{cellgray}$T+(2N_v+3)T_\mathsf{comp}+10\Delta+5\delta$&$T+(2N_v+3)T_\mathsf{comp}+12\Delta+5\delta$\\
\hline
$VC$ reconstructs and $V$'s receipt and sends it to $V$& $ T+(2N_v+4)T_\mathsf{comp}+11\Delta+5\delta$ & $T+(2N_v+4)T_\mathsf{comp}+12\Delta+5\delta$ &  \cellcolor{cellgray}$T+(2N_v+4)T_\mathsf{comp}+10\Delta+5\delta$&$T+(2N_v+4)T_\mathsf{comp}+12\Delta+5\delta$\\
\hline
$V$ obtains her receipt &\cellcolor{cellgray} $ T+(2N_v+4)T_\mathsf{comp}+11\Delta+6\delta$ & $T+(2N_v+4)T_\mathsf{comp}+12\Delta+6\delta$ &  $T+(2N_v+4)T_\mathsf{comp}+12\Delta+6\delta$&$T+(2N_v+4)T_\mathsf{comp}+12\Delta+6\delta$\\
\hline
$V$ verifies the validity of her receipt &$ T+(2N_v+5)T_\mathsf{comp}+11\Delta+6\delta$ &\cellcolor{cellgray}  $T+(2N_v+5)T_\mathsf{comp}+12\Delta+6\delta$ &  $T+(2N_v+5)T_\mathsf{comp}+12\Delta+6\delta$&$T+(2N_v+5)T_\mathsf{comp}+12\Delta+6\delta$\\
\hline
 \end{tabular}\end{center}
  \caption{Time upper bounds at $\Cl,\Cl[V]$, $\Cl[VC]$ and other honest VC nodes' clocks at each step of the interaction of the voter $V$ with responder $VC$ during D-DEMOS/Async voting phase. The grayed cells indicate the reference point of the clock drifts at each step.}
\label{tab:liveness-async}
 \end{figure}
\begin{theorem}[\textbf{Liveness of D-Demos/Async}]\label{thm:liveness_async}
Consider a D-DEMOS/Async run with $n$ voters, $m$ options and $N_v$ VC nodes. Let $\A$ be an adversary against D-DEMOS/Async with $m$ options and $n$ voters under the model described in Section~\ref{subsec:threat} that corrupts up to $f_v<N_v/3$ VC nodes.  Assume there is an upper bound $\Delta$ on clock synchronization loss and an upper bound $\delta$ on the delay of message delivery among honest VC nodes. Let $T_\mathsf{comp}$ be the worst-case running time of any procedure run by the VC nodes and the voters described in Sections~\ref{subsec:VcnodesAsync} and~\ref{subsec:voter} respectively, during the voting protocol.
 \par Let $T_\mathsf{end}$ denote the election end time. Define
 \[T_\mathsf{wait}:=(2N_v+5)T_\mathsf{comp}+12\Delta+6\delta\;.\] 
 Then, the following conditions hold:
 \begin{enumerate}
   \item\label{liveness-1-async} Every $[T_\mathsf{wait}]$-patient voter that is engaged in the voting protocol by the time that $\Cl[V]=T_\mathsf{end}-(f_v+1)\cdot T_\mathsf{wait}$, will obtain a valid receipt.
   \item\label{liveness-2-async} Every $[T_\mathsf{wait}]$-patient voter that is engaged in the voting protocol by the time that  
  $\Cl[V]=T_\mathsf{end}-y\cdot T_\mathsf{wait}$, where $y\in[f_v]$,
  will obtain a valid receipt with more than $1-3^{-y}$ probability.
 \end{enumerate}

\end{theorem}

\begin{proof} 
\ifextended
The $T_\mathsf{wait}$ upper bound is computed according to the diagram in Figure~\ref{figure:VoteCollectionAsyncPhase}. 
\else
The $T_\mathsf{wait}$ upper bound is computed according to the algorithm description in Section~\ref{subsec:VcnodesAsync}. 
\fi
Following the reasoning in the proof of Theorem~\ref{thm:liveness_IC}, we get that 
\[T_\mathsf{wait}:=(2N_v+5)T_\mathsf{comp}+12\Delta+6\delta\;.\] 
Subsequently, we show that conditions~(\ref{liveness-1-async}) and~(\ref{liveness-2-async}) hold for any $[T_\mathsf{wait}]$-patient voter,  exactly as in the proof of Theorem~\ref{thm:liveness_IC}.

\end{proof}
\else
\fi

%% file: safety.tex
%
%
D-DEMOS's safety guarantee is expressed as a contract adhered by the VC subsystem, stated in Section~\ref{sec:sysoverview}. This contract is fulfilled by both D-DEMOS versions, though D-DEMOS/IC requires some additional assumptions to hold, as compared with D-DEMOS/Async that assumes only fault tolerance of the underlying subsystems (see Section~\ref{subsec:threat}). 
\ifextended
Moreover, the proofs of safety of the two versions diverge. Specifically, the safety of D-DEMOS/IC relies on the security of the fixed SHA-256 hash function and the AES-128-CBC\$ symmetric encryption scheme. Therefore, the safety statement is with respect to specific security parameters. On the contrary, the safety of D-DEMOS/Async depends on the RSA signature scheme, therefore our analysis follows an asymptotic approach.
\fi
%
%
\subsubsection{Safety of D-DEMOS/IC}\label{subsubsec:safety_ic}
As in liveness, we assume the upper bounds $\delta,\Delta$ on the delay of message delivery and the drifts of all nodes' clocks to implement $T_\mathsf{end}$ and  $T_\mathsf{barrier}$ as the starting point and the barrier of the IC protocol. We consider 128-bit security of the commitment scheme assuming that every adversary running in less than $2^{64}$ steps has no more than $2^{-128}$ probability of obtaining any information about a single committed value (i.e., we set $c=6/7$, where $c$ is mentioned in Section~\ref{subsubsec:commit}).
%
\begin{theorem}[\textbf{Safety of D-Demos/IC}]\label{thm:safety_ic}
 Consider a D-DEMOS/IC run with $n$ voters, $m$ options, $N_v$ VC nodes, $N_b$ BB nodes and $N_t$ trustees under the restriction than $m\cdot n\leq2^{41}$. Let $\A$ be an adversary against D-DEMOS under the model described in Section~\ref{subsec:threat} that corrupts up to $f_v<N_v/3$ VC nodes, up to $f_b<N_b/2$ BB nodes and up to $N_t-h_t$ out-of $N_t$ trustees. Assume there is an upper bound $\Delta$ on clock synchronization loss and an upper bound $\delta$ on the delay of message delivery.  Let $T_\mathsf{end}$ be the end of the voting phase and $T_\mathsf{barrier}$ be the end of the value dissemination phase of the interactive consistency protocol, as described in Section~\ref{subsec:threat}. Then, all honest voters who received a valid receipt from a VC node, are assured that their vote will be published on the honest BB nodes and included in the election tally, with probability at least 
   \[1-\dfrac{nf_v}{2^{64}-f_v}-\big(3(mn)^3+2^{25}(mn)^2+2^{64}mn\big)\cdot2^{-125}\;.\]
\end{theorem}

\begin{proof} A crucial step for proving the safety of D-DEMOS/IC is to ensure it is hard for the adversary to compute non-submitted valid vote codes from the ballots of honest voters. This is done in the following claim.\\
\begin{myclaim}\label{clm:safety_ic}
\emph{The probability that $\A$ outputs a vote code from the ballot of some honest voter $V$ which was not cast by $V$ is less than $\big(3(mn)^3+2^{25}(mn)^2+2^{64}mn\big)\cdot2^{-125}\;.$}\medskip
%
\end{myclaim}
\begin{claimproof}
 Let $\mbf{C}$ be the set of all vote codes generated by the EA. An arbitrary execution of $\mc{A}$ determines the following subsets of $\mbf{C}$: (i) the set of vote codes $\mbf{C}_1$  that all honest voters submitted at the election phase , (ii)  the set of the vote codes  $\mbf{C}_2$ located in unused ballots of honest voters that did not engage in the voting protocol and (iii) the set of vote codes $\mbf{C}_3$ in the ballots of corrupted voters. \smallskip 
 \par Since every vote code is a random 128-bit string, the event that $\A$ guesses some of the $2mn$ vote codes can happen with no more than $2mn(2^{-128})=2^{-127}mn$ probability. Furthermore, $\A$ is restricted by the fault tolerance thresholds of the VC, BB and trustees subsystems. Hence, by (i) the random vote code generation, (ii) the fault tolerance thresholds, (iii) the hiding property of the commitment scheme and (iv) the perfect simulatability of the zero-knowledge proofs, we assume that except for some probability bounded by $2^{-127}mn+0+2^{-127}mn+0=2^{-126}mn$, the information associated with the vote codes that $\A$ obtains is, 
 \begin{enumerate}[(i).]
  \item The VC initialization data (for every VC node that $\A$ corrupts). 
 %
 %
  \item All the BB initialization data. The part of these data that is associated with the vote codes is the list of all AES-128-CBC\$ vote code encryptions under $\mathsf{msk}$.
  \item The set $\mbf{C}_1\cup\mbf{C}_2\cup\mbf{C}_3$.
 \end{enumerate}
\noindent\underline{\emph{Reduction to IND-CPA security of AES-128-CBC\$.}} Given the code of $\A$, we construct an algorithm $\mc{B}$ against the $(t,q,(2t+258\cdot q+3q^2)\cdot2^{-128})$-IND-CPA security of the underlying AES-128-CBC\$ (see Section~\ref{subsubsec:aes}). Namely, $\mc{B}$ invokes $\mc{A}$ and attempts to simulate a setup and run of D-DEMOS/IC as follows:
%
 \begin{enumerate}[1.]
   \item $\mc{B}$ chooses a random triple $(j^*,\ell^*,X^*)\in[m]\times[n]\times\{A,B\}$.
   \item For every $(j,\ell,X)\in[m]\times[n]\times\{A,B\}\setminus\{(j^*,\ell^*,X^*)\}$, $\mc{B}$ executes the following steps:
   \begin{enumerate}
    \item $\mc{B}$ chooses a random 64-bit $\mathsf{vote\textrm{-}code}_{\ell,j}^X$ and associates it with $\mathsf{option}^X_{\ell,j}$.
    \item $\mc{B}$ makes an encryption query $\big(m_{0,\ell,j}^X,m_{1,\ell,j}^X\big)=\big(\mathsf{vote\textrm{-}code}_{\ell,j}^X,\mathsf{vote\textrm{-}code}_{\ell,j}\big)^X$ and receives an AES-128-CBC\$ encryption of $\mathsf{vote\textrm{-}code}_{\ell,j}^X$.
    \item $\mc{B}$ chooses a random $\mathsf{salt}_{\ell,j}^X$ and computes $H_{\ell,j}^X\leftarrow SHA256(\mathsf{vote\textrm{-}code}_{\ell,j}^X, \mathsf{salt}_{\ell,j}^X)$.
    \item $\mc{B}$ generates the cryptographic payload $\mathsf{payload}_{\ell,\pi_\ell^X(j)} $ associated with $\mathsf{option}^X_{\ell,j}$.
   \end{enumerate}
   \item $\mc{B}$ chooses random values  $\mathsf{vote\textrm{-}code}^*_0,\mathsf{vote\textrm{-}code}^*_1\in\{0,1\}^{64},\mathsf{salt}^*\in\{0,1\}^{64}$.
   \item $\mc{B}$ makes the encryption query challenge $\mathsf{vote\textrm{-}code}^*_0,\mathsf{vote\textrm{-}code}^*_1$ and receives the AES-128-CBC\$ encryption $y^*$ of $\mathsf{vote\textrm{-}code}^*_b$, where $b$ is the outcome of a coin-flip.
   \item $\mc{B}$ tabulates BB initialization data as EA does, by using $\mathsf{vote\textrm{-}code}^*_0$ as the vote code associated with
   $\mathsf{option}_{\ell^*,j^*}$,  the hash $SHA256(\mathsf{vote\textrm{-}code}^*_0,\mathsf{salt}^*)$ as $H_{\ell^*,j^*}^{X^*}$  and $y^*$ as the AES-128-CBC\$ ciphertext that corresponds to $\mathsf{vote\textrm{-}code}^*_0$ . 
   \item $\mc{B}$ interacts with $\mc{A}$ according to the model described in Section~\ref{subsec:threat}.
   \item If $\mc{A}$ outputs $\mathsf{vote\textrm{-}code}^*_0$, then $\mc{B}$ outputs 0. Otherwise, $\mc{B}$ outputs 1.
 \end{enumerate}
 %
Let $G$ be the event that $\mc{A}$ outputs some $\mathsf{vote\textrm{-}code}\in\mbf{C}\setminus(\mbf{C}_1\cup\mbf{C}_2\cup\mbf{C}_3)$. By the construction of $\mc{B}$, if the IND-CPA challenge bit $b$ is $0$, then $\mc{B}$ simulates a D-DEMOS/IC election perfectly. Furthermore, if $b=0$ and $\mathsf{vote\textrm{-}code}$ corresponds to the randomly chosen position $(j^*,\ell^*,X^*)\in[m]\times[n]\times\{A,B\}$, then it outputs 0 ($\mathsf{vote\textrm{-}code}=\mathsf{vote\textrm{-}code}^*_0$). Since $\mc{B}$ randomly guesses the triple $(\ell^*,j^*,X^*)$, we have that
 \begin{equation}\label{eq:lemma_2}
 \begin{split}
&\Pr[\mc{B}\mbox{ outputs }1\mid b=0]=1-\Pr[\mc{B}\mbox{ outputs }0\mid b=0]=1-\dfrac{\Pr[G\mid b=0]}{2mn}\;.
 \end{split}
 \end{equation}
 On the other hand, if $b=1$, then $\mathsf{vote\textrm{-}code}^*_0$ is the preimage of $SHA256(\mathsf{vote\textrm{-}code}^*_0 , \mathsf{salt}^*)$, while $y^*$ is the encryption of an independently generated vote code. Based on this observation, we construct an algorithm $\mc{C}$ that acts as an attacker against the $(t,t^2\cdot2^{-256})$-collision resistance of SHA-256 (see Section~\ref{subsubsec:crhf}). Namely, on input some hash value $H$, $\mc{C}$ executes the following steps:
%
\begin{enumerate}[1.]
 \item $\mc{C}$ chooses a random triple $(j^*,\ell^*,X^*)\in[m]\times[n]\times\{A,B\}$.
 \item For every $(j,\ell,X)\in[m]\times[n]\times\{A,B\}$, $\mc{C}$ chooses random values $\mathsf{vote\textrm{-}code}_{\ell,j}^{X}\in\{0,1\}^{160},\mathsf{salt}_{\ell,j}^{X}\in\{0,1\}^{64}$.
 \item $\mc{C}$ tabulates all election information normally except that for $(\ell^*,j^*,X^*)$ it provides $H$ instead of the hash value $SHA256(\mathsf{vote\textrm{-}code}^{X^*}_{\ell,j}, \mathsf{salt}_{\ell,j}^{X^*})$.
 \item  $\mc{C}$ interacts with $\mc{A}$ according to the model described in Section~\ref{subsec:threat}.
 \item $\mc{C}$ receives the output of $\mc{A}$, labeled by $z$.
 \item $\mc{C}$ searches for a $w\in\{0,1\}^{64}$ s.t. $h(z,w)=H$. If $\mc{C}$ finds such a $w$, then it outputs $z||w$. Otherwise, it aborts.  
\end{enumerate}
%
For simplicity and w.l.o.g., we can assume that for each $(j,\ell,X)\in[m]\times[n]\times\{A,B\}$, the time complexity for information preparation is on the order of $256^3$ (cube of the string length, set to 256 bits).
The running time of $\mc{A}$ is $2^{64}$. Assuming linear complexity for hashing and checking a random value, the brute force search for the correct $w$ in step 6. takes $2^{64}\cdot256=2^{72}$ steps.  Therefore, given that $mn\leq2^{41}$, we conclude the $\mc{C}$ runs in steps bounded by
$2mn\cdot256^3+2^{64}+2^{64}\cdot256\leq mn2^{25}+2^{64}+2^{72}<2^{73}\;.$\smallskip
\par By the $(t,t^2\cdot2^{-256})$-collision resistance of $h(\cdot)$ (see Section~\ref{subsubsec:crhf}), the probability that $\mc{C}$ finds a preimage of $H$ is less than $2^{146}\cdot2^{-256}<2^{-110}$. By the construction of $\mc{C}$, if $\mc{A}$ outputs the vote code that corresponds to position $(\ell^*,j^*,X^*)\in[n]\times[m]\times\{A,B\}$, then $\mc{C}$ certainly wins. Therefore, we have that
 \begin{equation}\label{eq:lemma_3}
  \begin{split}
 \Pr[&\mc{B}\mbox{ outputs }1\mid b=1]=1-\Pr[\mc{B}\mbox{ outputs }1\mid b=1]=1-\dfrac{\Pr[G\mid b=1]}{2mn}-2^{-126}mn\geq\\
  &\geq1-\Pr[\mc{C}\mbox{ returns the preimage of SHA-}256]%
  >1-2^{-110}-2^{-126}mn\;.
  \end{split}
 \end{equation}
Hence, by Eq.~\eqref{eq:lemma_2},\eqref{eq:lemma_3}, we conclude that
   \begin{equation}\label{eq:lemma_4}
   \begin{split}
   \mathbf{Adv}_\mathsf{128-AES-CBC\$}^{\mathsf{IND-CPA}}(\mathcal{B})
>\dfrac{\Pr[G\mid b=0]}{2mn}-2^{-110}-2^{-126}mn\;.
%
%
\end{split}
  \end{equation}

Along the lines of the time complexity analysis of $\mc{C}$, the time complexity of $\mc{B}$ is bounded by $2mn\cdot256^3+2^{64}=2^{25}mn+2^{64}<2^{66}$, where we used that $mn\leq2^{41}$,
In addition, $\mc{B}$ makes at most $2\cdot m\cdot n$ queries. Hence, by the $(t,q,(2t+258\cdot q+3q^2)\cdot2^{-128})$- IND-CPA security of AES-CBC (see Section~\ref{subsubsec:aes}) and~\eqref{eq:lemma_4}, we conclude that
 \begin{equation*}
  \begin{split}
\dfrac{\Pr[G\mid b=0]}{2mn}&-2^{-110}-2^{-126}mn<(2^{26}mn+2^{65}+516 mn+12\cdot (mn)^2)\cdot2^{-128}\Rightarrow\\
 \Rightarrow \Pr[G\mid b=0]&
<\big(3(mn)^3+2^{25}(mn)^2+2^{64}mn\big)\cdot2^{-125}\;,
  \end{split}
 \end{equation*}
 which completes the proof of the claim, as the election simulation for $b=0$ is perfect.\\
 $\quad$\hfill
\end{claimproof}
\medskip\par Given Claim~\ref{clm:safety_ic}, the proof is completed in two stages. 
\vspace{5pt}
\par\emph{1.Vote set consensus. }
By the upper bound restriction on all clock drifts, all honest VC nodes will enter the Value Dissemination phase at $T_\mathsf{end}$ and the Result Consensus phase 
of the Interactive Consistency protocol at $T_\mathsf{barrier}$ within some distance $\Delta$ from the global clock. The agreement property of interactive consistency ensures that all honest VC nodes will contain the same vector 
$\langle VS_1,\ldots,VS_n \rangle$ of all nodes' sets of voted and pending ballots. Subsequently, all honest VC nodes, execute the same deterministic algorithm of Figure~\ref{fig:afterIC}, and will agree on the same set of votes denoted by $\mathsf{Votes}$. This will be the set of votes that are marked to be tallied by the honest VC nodes.\\
\par\emph{2. Protocol contract. }
Let $V_\ell$ be an honest voter that has obtained a receipt for his vote $\langle \mathsf{serial\textrm{-}no}, \mathsf{vote\textrm{-}code} \rangle$, but his vote is not included in $\mathsf{Votes}$. By the vote consensus property proved previously, we have that some honest VC node $VC$, decided to discard $V_\ell$'s vote. According to the algorithm described in Figure~\ref{fig:afterIC} that determines $\mathsf{Votes}$, the latter can happen only because either \textbf{Case (i)}: $\A$ succeeds in guessing the valid receipt of $V_\ell$, or \textbf{Case (ii)}: a $\mathsf{vote\textrm{-}code\textrm{-}2}$ different than $\mathsf{vote\textrm{-}code}$ appears in the list for the ballot indexed by $\mathsf{serial\textrm{-}no}$ or \textbf{Case (iii)}: $\mathsf{vote\textrm{-}code}$ appears less than 
$N_v-2f_v$ times in the list for the ballot indexed by $\mathsf{serial\textrm{-}no}$. We study all Cases (i),(ii),(iii):\medskip
\par\textbf{Case (i).} If $\A$ succeeds in guessing a valid receipt, then it can force the VC subsystem to consider $V$'s ballot not voted by not participating in the receipt reconstruction. By the information theoretic security of the VSS scheme, given that $\A$ is restricted by the fault tolerance thresholds, its guess of the receipt must be at random. Since there are at most $f_v$ malicious VC nodes, the adversary has at most $f_v$ attempts to guess the receipt. Moreover, the receipt is a randomly generated 64-bit string, so after $i$ attempts, $\A$ has to guess among $(2^{64}-i)$ possible choices. Taking a union bound for $n$ voters, the probability that $\A$ succeeds for any of the obtained receipts is no more than
\[\displaystyle\sum_{\ell=1}^{n}\bigg(\displaystyle\sum_{i=0}^{f_v-1}\dfrac{1}{2^{64}-i}\bigg)\leq \dfrac{nf_v}{2^{64}-f_v}\;.\]
\par\textbf{Case (ii).} $V_\ell$ is honest, hence it has submitted the same vote in every possible attempt to vote prior to the one she obtained her receipt. Therefore, Case (ii) may occur only if the adversary $\A$ manages to produce $\mathsf{vote\textrm{-}code\textrm{-}2}$ by the vote code related election information it has access to. Namely, 
(a) the set of vote codes that all honest voters submitted at the election phase, (b)  the set of the vote codes that were located in unused ballots and (c) the set of vote codes in the ballots of corrupted voters.
By assumption, $\mathsf{vote\textrm{-}code\textrm{-}2}$ is in neither of these three sets. Hence, by Claim~\ref{clm:safety_ic}, the probability that $\A$ computes $\mathsf{vote\textrm{-}code\textrm{-}2}$ is less than $\big(3(mn)^3+2^{25}(mn)^2+2^{64}mn\big)\cdot2^{-125}.$\medskip
\par\textbf{Case (iii).} 
In order for $V_\ell$ to obtain a receipt, 
at least $N_v-f_v$ VC nodes must collaborate by providing their shares. The faulty VC nodes are at most $f_v$, so at least $N_v-2f_v$ honest VC nodes will include $\langle \mathsf{serial\textrm{-}no}, \mathsf{vote\textrm{-}code} \rangle$ in their set of voted and pending ballots. Thus, Case (iii) cannot occur.\medskip
 \par Consequently, all the honest VC nodes will forward the agreed set of votes (hence, also $V_\ell$'s vote) to the BB nodes.
 By the fault tolerance threshold for the BB subsystem, the $f_b$ honest BB nodes will publish $V_\ell$'s vote. Finally, the $h_t$ out-of $N_t$ honest trustees will read $V$'s vote from the majority of BB nodes and include it in the election tally. Therefore,  the probability that $\A$ achieves in excluding the vote of at least one honest voter that obtained a valid receipt from the BB or the election tally is less than $\dfrac{nf_v}{2^{64}-f_v}-\big(3(mn)^3+2^{25}(mn)^2+2^{64}mn\big)\cdot2^{-125},$ which completes the proof.
  %
%
%
\end{proof}
%
%
\ifextended
\subsubsection{Safety of D-DEMOS/Async}\label{subsubsec:safety_async}
The safety of D-DEMOS/Async is founded on the certificate generation mechanism among the VC nodes, which in turn exploits the security of the underlying signature scheme. 
\begin{theorem}[\textbf{Safety of D-Demos/Async}]\label{thm:sec_safety_async}
Let $\A$ be an adversary against D-DEMOS under the model described in Section~\ref{subsec:threat} that corrupts up to $f_v<N_v/3$ VC nodes, up to $f_b<N_b/2$ BB nodes and up to $N_t-h_t$ out-of $N_t$ trustees.  
Then, all honest voters who received a valid receipt from a VC node, are assured that their vote will be published on the honest BB nodes and included in the election tally, with probability at least 
\[1-\dfrac{nf_v}{2^{64}-f_v}-\mathsf{negl}(\lambda)\;.\]
\end{theorem}
\begin{proof}
 Let $V_\ell$ be an honest voter. Then, $\A$'s strategy on attacking safety (i.e., provide a valid receipt to $V_\ell$ but force the VC subsystem to discard $V$'s ballot), is captured by either one of the two following cases: \textbf{Case (i):} $\A$ produces the receipt without being involved in a complete interaction with the VC subsystem (i.e., with at least $f_v+1$ honest VC nodes). \textbf{Case (ii):} $\A$ provides a properly reconstructed receipt via a complete interaction with the VC subsystem (in both cases we assume $\A$ controls the \emph{responder} VC node).
\par Let $E_1$ (resp. $E_2$) be the event that Case 1 (resp. Case 2)  happens. We study both cases:\\
 \par\textbf{Case (i).} In this case, $\A$ must produce a receipt that matches $V$'s ballot with less than $N_v-f_v$ shares. $\A$ may achieve this by either one of the following ways:
 \begin{enumerate}[1.]
  \item $\A$ attempts to guess the valid receipt; If $\A$ succeeds, then it can force the VC subsystem to consider $V$'s ballot not voted as no valid UCERT certificate will be generated for $V$'s ballot (malicious \emph{responder} does not send an ENDORSE message). As shown in the proof of Theorem~\ref{thm:safety_ic}, the probability of a successful guess for $\A$ is less than $\frac{nf_v}{2^{64}-f_v}\;.$
  %
 \item $\A$ attempts to produce fake UCERT certificates by forging digital signatures of other nodes. By the security of the digital signature scheme, this attack has $\negl(\lambda)$ success probability.
 \end{enumerate}
By the above, we have that $\Pr[\A\mbox{ wins }|E_1]\leq\dfrac{nf_v}{2^{64}-f_v}+\negl(\lambda)\;.$\\
%
 \par\textbf{Case (ii).} In this case, by the security arguments stated in Section~\ref{subsec:Vcnodes} (steps~\ref{bc-step-announce}-~\ref{bc-step-result-1}), every honest VC node will include the vote of $V_\ell$ in the set of voted tuples. This is because a) it locally knows the valid (certified) vote code for $V_\ell$ which is  accompanied by UCERT or b) it has obtained the valid vote code via a RECOVER-REQUEST message. Recall that unless there are fake certificates (which happens with negligible probability) there can be only one valid vote code for $V_\ell$.
 \par Consequently, all the honest VC nodes will forward the agreed set of votes (hence, also $V_\ell$'s vote) to the BB nodes.
 By the fault tolerance threshold for the BB subsystem, the $f_b$ honest BB nodes will publish $V$'s vote. Finally, the $h_t$ out-of $N_t$ honest trustees will read $V_\ell$'s vote from the majority of BB nodes and include it in the election tally. Thus, we have that  $\Pr[\A\mbox{ wins }|E_2]=\negl(\lambda)\;.$\\
%
\par Therefore, all the votes of honest voters that obtained a valid receipt, will be published on the honest BB nodes and included in the election tally, with probability at least \[1-\Pr[\A\mbox{ wins }]\geq1-\Pr[\A\mbox{ wins }|E_1]-\Pr[\A\mbox{ wins }|E_2]\geq1-\dfrac{nf_v}{2^{64}-f_v}-\mathsf{negl}(\lambda)\;.\]
\end{proof}
%
%
\else
%
\fi

%% file: E2E.tex
	\begin{boxfig}{\label{fig:int.game}The E2E Verifiability Game between the challenger $\Ch$ and the adversary $\mathcal{A}$ using the vote extractor $\mathcal{E}$.}{}	
			\underline{E2E Verifiability Game $G_{\mathsf{e2e\textrm{-}ver}}^{\mathcal{A},\mathcal{E},d,\theta}(1^\la,m,n,N_v,N_b,N_t)$}
	\begin{enumerate}[(i).]
	\item[]
		\item $\mathcal{A}$ on input $1^\lambda, n, m,N_v,N_b,N_t$, chooses a list of options $\{\mathsf{option}_1,\ldots,\mathsf{option}_m\}$, a set of voters $\mathcal{V}=\{V_1,\ldots,V_n\}$, a set of VC nodes $\mathcal{VC}=\{\mathsf{VC}_1,\ldots,\mathsf{VC}_{N_v}\}$, a set of 
  BB nodes $\mathcal{BB}=\{\mathsf{BB}_1,\ldots,\mathsf{BB}_{N_b}\}$, and a set of trustees $\mathcal{T}=\{T_1,\ldots,T_{N_t}\}$. It provides the challenger $\mathsf{Ch}$ with all the above sets. Throughout the game, $\mathcal{A}$ controls the EA, all the VC nodes and  all the trustees. In addition, $\mathcal{A}$ may corrupt a fixed set of less than $\lfloor N_b/2\rfloor$ BB nodes,  denoted by $\mathcal{BB}_\mathsf{succ}$  (i.e., the majority of the BB nodes remain honest). On the other hand, $\mathsf{Ch}$ plays the role of all the honest BB nodes.
		\item
		$\mathcal{A}$ and $\Ch$ engage in an interaction
		where $\mathcal{A}$ schedules the vote casting executions of all voters. 
		For each voter $V_\ell$, $\mathcal{A}$ can either completely control the voter 
		or allow $\Ch$ to operate on $V_\ell$'s behalf, in which case  $\mathcal{A}$ provides $\Ch$ with an option selection $\mathsf{option}_{i_\ell}$. 
		Then, $\Ch$ casts a vote for $\mathsf{option}_{i_\ell}$,  and, provided the voting execution terminates successfully,  $\Ch$ obtains
		the audit information $\mathsf{audit}_\ell$ on behalf of $V_\ell$.
		\item Finally, $\mathcal{A}$ posts a version of the election transcript $\mathsf{info}_j$ in every honest BB node $\mathsf{BB}_j\notin\mathcal{BB}_\mathsf{corr}$.
	\end{enumerate}   
	Let $\mathcal{V}_\mathsf{succ}$ be the set of honest voters (i.e., those controlled by  $\Ch$) that terminated successfully. The game returns a bit which is $1$ if and
	only if the  following conditions hold true:
	\begin{enumerate}[(1)]
	        \item\label{e2e-cond2} $\forall\mathsf{BB}_{j}, \mathsf{BB}_{j'}\notin\mathcal{BB}_\mathsf{corr}:$ $\mathsf{info}_j=\mathsf{info}_{j'}:=\mathsf{info}$ 
		\item\label{e2e-cond3} $ | \mathcal{V}_\mathsf{succ}|  \geq \theta$ (i.e., at least $\theta$ honest voters terminated). 
		\item\label{e2e-cond4} $\forall\ell\in[n]:$ if $ V_\ell\in\mathcal{V}_\mathsf{succ}$ then $V_\ell$ verifies successfully, when given $(\mathsf{info},\mathsf{audit}_\ell)$ as input. 
		\medskip
\item[]	and either one of the following two conditions: \medskip
 \item 
	\begin{enumerate}
		\item\label{e2e-cond5a}  if $\bot\neq \langle\mathsf{option}_{i_\ell}\rangle_{ V_\ell\notin\mathcal{V}_\mathsf{succ}} \leftarrow$ $\mathcal{E}(\mathsf{info}, \{\mathsf{audit}_\ell\}_{ V_\ell\in\mathcal{V}_\mathsf{succ}} )$ then
		 \[\mathrm{d}_1\big(\mbf{Result}(\mathsf{info}),F(\mathsf{option}_{i_1}\ldots,\mathsf{option}_{i_n})\big) \geq d\;.\] 
		\item\label{e2e-cond5b} $\bot \leftarrow \mathcal{E}(\mathsf{info}, \{\mathsf{audit}_\ell\}_{ V_\ell\in\mathcal{V}_\mathsf{succ}} )$.  
	\end{enumerate}
\end{enumerate}
	\end{boxfig}
We adopt the end-to-end (E2E) verifiability definition in~\cite{DEMOS}, modified accordingly to our setting. Namely, we encode the options set $\{\mathsf{option}_1,\ldots,\mathsf{option}_m\}$, where the encoding of $\mathsf{option}_i$ is an $m$-bit string which is $1$ only in the $i$-th position. Let $F$ be the \emph{election evaluation function} such that $F(\mathsf{option}_{i_1}\ldots,\mathsf{option}_{i_n})$ is equal to an $m$-vector whose $i$-th location is equal to the number of times $\mathsf{option}_i$ was voted. 
Then, we use the metric $\mr{d}_1$ derived by the L1-norm scaled to half, i.e., $\mr{d}_1(R,R')=\frac{1}{2}\cdot\sum_{i=1}^n|R_i-R'_i|$, 
where $R_i,R'_i$ is the $i$-th coordinate of $R,R'$ respectively, to measure the success probability of the adversary with respect to the amount of tally deviation $d$ and the number of voters that perform audit $\theta$.
In addition, we make use of a \emph{vote extractor} algorithm $\mathcal{E}$ (not necessarily running in polynomial-time) that extracts the non-honestly cast votes.
\par We define E2E verifiability  via an attack game between a challenger and an adversary  specified in detail in Figure~\ref{fig:int.game}. 
\begin{definition}[\textsc{\textbf{E2E Verifiability}}]\label{def:int.dfnt}
	Let $0<\epsilon<1$ and $n,m,N_v,N_b,N_t\in\mathbb{N}$ polynomial in the security parameter $\lambda$ with $\theta\leq n$. 
	Let $\Pi$ be an e-voting system with $n$ voters, $N_v$ VC nodes, $N_b$ BB nodes and $N_t$ trustees.
	We say that $\Pi$  achieves \emph{end-to-end verifiability} 
		with error $\epsilon$,
	w.r.t. the election  function $F$,
	a number of $\theta$ honest successful voters and tally deviation $d$
	if there exists a 
	(not necessarily polynomial-time) vote extractor $\mathcal{E}$  such that 
	for any PPT adversary $\A$ it holds that
	\[\Pr[G_{\mathsf{e2e\textrm{-}ver}}^{\mathcal{A},\mathcal{E},d,\theta}(1^\la,m,n,N_v,N_b,N_t)=1] \leq \epsilon.\]
\end{definition}
\vspace{5pt}
\ifextended
To prove E2E verifiability of D-DEMOS, we need a min-entropy variant of the Schwartz-Zippel lemma, to check the equality of two univariate polynomials $p_1,p_2$, i.e., test $p_1(x)-p_2(x)=0$ for random $x\pick{D} \mathbb{Z}_q$, where $q$ is prime. The probability that the test passes is at most $\frac{\max(d_1,d_2)}{2^\kappa}$ if $p_1\neq p_2$, where $d_i$ is the degree of $p_i$ for $i\in\set{1,2}$. 
We leverage Lemma~\ref{lem:SZ} from \cite{DEMOS}. 

\begin{lemma}[\textsc{\textbf{Min-entropy Schwartz-Zippel~\cite{DEMOS}}}]\label{lem:SZ}
Let $q$ be a prime and $p(x)$ be a non-zero univariate polynomial of degree $d$ over $\mathbb{Z}_q$. Let $D$ be a probability distribution on $\mathbb{Z}_q$ such that $H_{\infty}(D)\geq \kappa$. The probability of $p(x)=0$ for a randomly chosen $x\pick{D} \mathbb{Z}_q$ is at most $\frac{d}{2^\kappa}$. 
\end{lemma}
\medskip
\par We now analyse the soundness of the zero knowledge proof for each option encoding commitment. Note that a correct option encoding is an $m$-vector, where one of the $m$ elements is $1$ and the rest elements are $0$ (a.k.a. unit vector).  Our zero knowledge proof utilizes the Chaum-Pedersen  DDH-tuple proofs~\cite{CP} in conjunction with the Sigma OR-composition technique~\cite{CDS94} to show each (lifted) ElGamal ciphertext encrypts either $0$ or $1$ and the product of all the $m$ ElGamal ciphertexts encrypts $1$. 
We adopt the soundness amplification technique from~\cite{DEMOS}; namely, if the voters' coins $\textbf{c}$ are longer than $\lfloor\log q\rfloor$ then we divide it into $\kappa$ blocks, $(\textbf{c}_1,\textbf{c}_2,\ldots,\textbf{c}_\kappa)$ such that each block has less than $\lfloor\log q\rfloor$ coins,  where $q$ is the order of the underlying group used in the ElGamal encryption. Given a statement $x$, for each $\textbf{c}_i$, $i\in[\kappa]$, the prover needs to produce the zero knowledge transcript $(x, \phi_{1,i},\textbf{c}_i,\phi_{2,i})$ in order. The verifier accepts the proof if and only if for all $i\in[\kappa]$, $\mathsf{Verify}(x, \phi_{1,i},\textbf{c}_i,\phi_{2,i}) = \mathsf{accept}$. Hence, we have the following Lemma~\ref{lem:zk}.\\[3pt]
\begin{lemma}\label{lem:zk}
Denote $\textbf{c} =(\textbf{c}_1,\textbf{c}_2,\ldots,\textbf{c}_\kappa)$. If $H_{\infty}(\textbf{c})=\theta$, we have for all adversaries $\mathcal{A}$:
$$
\varepsilon(m,n,\theta,\kappa) = \Pr\left[ \begin{array}{l}
(x,\{\phi_{1,i}\}_{i\in[\kappa]})\leftarrow\mc{A}(1^{\lambda}); \\
\{\phi_{2,i}\}_{i\in[\kappa]}\leftarrow\mc{A}(\textbf{c}_1,\textbf{c}_2,\ldots,\textbf{c}_\kappa):\\
x \textrm{ is not a valid option encoding commitment }\\ 
\wedge \forall i \in[\kappa],  \mathsf{Verify}(x, \phi_{1,i},\textbf{c}_i,\phi_{2,i}) = \mathsf{accept}
\end{array}\right] \leq 2^{-\theta}\;.
$$
\end{lemma}
\begin{proof}
For $i\in{\kappa}$, denote $H_{\infty}(\textbf{c}_i)=\theta_i$, and $\sum^{\kappa}_{i=1} \theta_i = \theta$.
Chaum-Pedersen  DDH-tuple proof~\cite{CP} internally constructs and checks a degree-$1$ polynomial; therefore according to Lemma~\ref{lem:SZ}, the probability that the adversary $\mc{A}$ to cheat a single DDH-tuple zero knowledge proof is at most $2^{-\theta'}$, where $\theta'$ is the min-entropy of the challenge. Moreover, Sigma OR-composition technique~\cite{CDS94} perfectly maintains the soundness, so the probability that the adversary $\mc{A}$ to cheat the zero knowledge proofs for each (lifted) ElGamal ciphertext encrypts $0/1$ is at most $2^{-\theta'}$. Note that the zero knowledge proofs of the option encoding commitment is AND-composition of all the elementary zero knowledge proofs, the probability that $x$ is invalid and $\mathsf{Verify}(x, \phi_{1,i},\textbf{c}_i,\phi_{2,i}) = \mathsf{accept}$ is at most $2^{-\theta_i}$. Hence, the probability that $\forall i\in[\kappa]$, $\mathsf{Verify}(x, \phi_{1,i},\textbf{c}_i,\phi_{2,i}) = \mathsf{accept}$ is
$\varepsilon(m,n,\theta,\kappa) = \prod_{i=1}^{\kappa} 2^{-\theta_i} = 2^{- \sum_{i=1}^{\kappa}\theta_i} = 2^{-\theta}.$
\end{proof}

Applying Lemma~\ref{lem:zk}, we prove that D-DEMOS (both the IC and the Async version) achieves E2E verifiability according to Definition~\ref{def:int.dfnt}.
\fi
\ifextended
\begin{proof}
  Without loss of generality, we can assume that every party can read consistently the data published in the majority of the BB nodes, as otherwise the adversary fails to satisfy either condition~\ref{e2e-cond1} or~\ref{e2e-cond2} of the E2E verifiability game.
 \par We first construct a vote extractor $\mc{E}$ for D-DEMOS as follows:
 \begin{small}
 \begin{framed}
\begin{itemize}
\renewcommand{\labelitemi}{$\bullet$}
 \item  $\mc{E}$ takes input as the election transcript, $\mathsf{info}$ and a set of audit information $\set{\mathsf{audit}_\ell}_{V_\ell\in\mathcal{V}_\mathsf{succ}}$. If $\mathsf{info}$ is not meaningful, then $\mc{E}$ outputs $\bot$. 
 \item  Let $B\leq|\tc{V}|$ be the number of different serial numbers that appear in $\set{\mathsf{audit}_\ell}_{V_\ell\in\tc{V}}$. $\mc{E}$ (arbitrarily) arranges the voters in $V_\ell\in\mathcal{V}_\mathsf{succ}$ and the serial numbers not included in
  $\set{\mathsf{audit}_\ell}_{V_\ell\in\mathcal{V}_\mathsf{succ}}$ as $\langle V^{\mc{E}}_{\ell}\rangle_{\ell\in[n-|\mathcal{V}_\mathsf{succ}|]}$ and $\langle\mr{tag}^{\mc{E}}_{\ell}\rangle_{\ell\in[n-B]}$ respectively.
  \item For every  $\ell\in[n-|\mathcal{V}_\mathsf{succ}|]$, $\mc{E}$ extracts $\mathsf{option}_{i_\ell}$ by brute force opening and decrypting (in superpolynomial time) all the committed and encrypted BB data, or sets $\mathsf{option}_{i_\ell}$ as the zero vector, in case $V_\ell$'s vote is not published in the BB. 
  \item If there is an invalid option-commitment (i.e., it is not a commitment to some candidate encoding) ,then $\mc{E}$ outputs $\bot$. Otherwise, it outputs $\langle\mathsf{option}_{i_\ell}\rangle_{ V_\ell\notin\mathcal{V}_\mathsf{succ}}$.
\end{itemize}
\end{framed}
\end{small}
\par We will prove the E2E verifiability of D-DEMOS based on $\mc{E}$. Assume an adversary $\A$ that wins the game $G_{\mathsf{e2e\textrm{-}ver}}^{\mathcal{A},\mathcal{E},d,\theta}(1^\la,m,n,N_v,N_b,N_t)$. Namely, $\A$ breaks E2E verifiability by allowing at least $\theta$ honest successful voters and achieving tally deviation $d$. \vspace{2pt}
\par Let $Z$ be the event that $\A$ attacks by making at least one of the option-encoding commitments associated with some cast vote code invalid (i.e., it is in tally set $\mathbf{E}_{\mathsf{tally}}$ but it is not a commitment to some candidate encoding). By condition~\ref{e2e-cond3}, there are at least 
$\theta$ honest and successful voters, hence the min-entropy of the collected voters' coins is at least $\theta$. By Lemma~\ref{lem:zk}, the zero-knowledge proofs used in D-DEMOS for committed ballot correctness in the BB is sound except for some probability error $2^{-\theta}$.
Since $\theta\geq1$ and condition~\ref{e2e-cond4} holds, there is at least one honest voter that verifies, thus we have that
$\Pr[G_{\mathsf{e2e\textrm{-}ver}}^{\mathcal{A},\mathcal{E},d,\theta}(1^\la,m,n,N_v,N_b,N_t)=1\wedge Z]\leq 2^{-\theta}\;.$\\
%
\par Now assume that $Z$ does not occur. In this case, the vote extractor $\mc{E}$ will output the intended adversarial votes up to permutation. Thus, the deviation from the intended result that $\A$ achieves, derives only by miscounting the honest votes. This may be achieved by $\A$ in two different possible ways:
\begin{itemize}
 \item \textbf{Modification attacks. } When committing to the information of some honest voter's ballot part $\A$ changes the vote code and option
correspondence that is printed in the ballot. This attack will be detected if the voter does chooses to audit with the modified ballot part (it uses the other part to vote). The maximum deviation achieved by this attack is $1$ (the vote will count for another candidate).
\item \textbf{Clash attacks. } $\A$ provides $y$ honest voters with ballots that have the same serial number, so that the adversary can inject $y-1$ votes of his preference in the $y-1$ ``empty'' audit locations in the BB. This attack is successful only if all the $y$ voters verify the same ballot on the BB and hence miss the injected votes that produce the tally deviation. The maximum deviation achieved by this attack is $y-1$.
\end{itemize}
We stress that if $Z$ does not occur, then the above two attacks are the only meaningful\footnote{By meaningful we mean that the attack is not trivially detected. For example, the adversary may post malformed information in the BB nodes but if so, it will certainly fail at verification.} for $\A$ to follow. Indeed, if (i) all zero knowledge proofs are valid, (ii) all the honest voters are pointed to a unique audit BB location indexed by the serial number on their ballots, and (iii) the information committed in this BB location matches the vote code and option association in the voters' unused ballot parts, then by the binding property of the commitments, all the tally computed by the commitments included in $\mathbf{E}_{\mathsf{tally}}$ will decrypt to the actual intended result.
\par Since the honest voters choose the used ballot parts at random, the success probability of $x$ deviation via the modification attack is $(1/2)^x$. In addition,  the success probability to clash $y$ honest voters is $(1/2)^{y-1}$ (all $y$ honest voters choose the same version to vote). As a result, by combinations of modification and clash attacks, $\A$'s success probability reduces by a factor $1/2$ for every unit increase of tally deviation. Therefore, the upper bound of the success probability of $\A$ when $Z$ does not occur is 
$ \Pr[G_{\mathsf{e2e\textrm{-}ver}}^{\mathcal{A},\mathcal{E},d,\theta}(1^\la,m,n,N_v,N_b,N_t)=1\mid\neg Z]\leq 2^{-d}\;.$\\
%
\par Hence, we conclude that  $\Pr[G_{\mathsf{e2e\textrm{-}ver}}^{\mathcal{A},\mathcal{E},d,\theta}(1^\la,m,n,N_v,N_b,N_t)=1]\leq 2^{-\theta}+2^{-d}\;.$
\end{proof}	
%
%
\else
To prove the E2E verifiability of D-DEMOS, we provide the following lemma.
The lemma proof, that we omit due to space limitations, is based on the min-entropy variant of the Schwartz-Zippel lemma utilised in~\cite{DEMOS}.
\begin{lemma}\label{lem:zk}
If the zero-knowledge challenge $\mathbf{c}$ has min-entropy $H_{\infty}(\mathbf{c})=\theta$, then the probability that an adversary $\mathcal{A}$ generates a fake zero knowledge proof for the validity of an option-encoding commitment that verifies succsefully is no more than $2^{-\theta}$.
\end{lemma}

\fi
Applying Lemma~\ref{lem:zk}, the following theorem states that D-DEMOS (both the IC and the Async version) achieves E2E verifiability according to Definition~\ref{def:int.dfnt}.
\begin{theorem}[\textsc{\textbf{E2E Verifiability of D-Demos}}] Let $n,m,N_v,N_b,N_t,\theta, d\in\mathbb{N}$ where $1\leq\theta\leq n$.
Then, D-DEMOS run with $n$ voters, $m$ options, $N_v$ VC nodes, $N_b$ BB nodes and $N_t$ trustees achieves end-to-end  
		with error $2^{-\theta}+2^{-d}$,
	w.r.t. the election  function $F$,
	a number of $\theta$ honest successful voters and tally deviation $d$.
\end{theorem}
\begin{proof}(\emph{Sketch}). Without loss of generality, we can assume that every party can read consistently the data published in the majority of the BB nodes, as otherwise the adversary fails to satisfy condition \ref{e2e-cond2} of the E2E verifiability game. Via brute force search, the vote extractor $\mc{E}$ for D-DEMOS either (i) decrypts the adversarial votes (up to permutation) if all respective option-encoding commitments are valid, or (ii) aborts otherwise. We analyze the two cases
\par 
(i) If all option-encoding commitments are valid, then the output of $\mc{E}$ implies that the tally deviation that the adversary $\mc{A}$ can achieve may derive only by attacking the honest voter. Namely, by pointing the honest voter to audit in a BB location where the audit data is inconsistent with the respective information in at least one part of the voter's ballot. As in~\cite[Theorem 4]{DEMOS}, we can show that every such single attack has 1/2 success probability (the voter had chosen to vote with the inconsistent ballot part) and in case of success, adds 1 to the tally deviation. Thus, in this case, the probability that  $\mathcal{A}$ causes tally deviation $d$ is no more than $2^{-d}$.
\par
(ii) If there is an invalid option-encoding commitment ($\mc{E}$ aborts), then the min entropy provided by at least $\theta$ honest succesful voters is at least $\theta$. Thus, by Lemma~\ref{lem:zk}, the Sigma protocol verification will fail except from some soundness error $2^{-\theta}$.
\par The proof is completed by taking the union bound on the two cases.
\end{proof}

%% file: privacy.tex
\begin{boxfig}{\label{fig:priv.game} The Voter privacy Game between the adversary $\mathcal{A}$ and the challenger $\mathsf{Ch}$ using the simuator $\mathcal{S}$.}{}
  {\it \underline{Voter Privacy Game $G_{\mathsf{priv}}^{\mathcal{A},\mathcal{S},\phi}(1^{\la},n,m,N_v,N_b,N_t)$}}
  \begin{enumerate}[(i).]
  \item $\mathcal{A}$ on input $1^\lambda, n, m,N_v,N_b,N_t$, chooses a list of options $\mathcal{P}=\{P_1,\ldots,P_m\}$, a set of voters $\mathcal{V}=\{V_1,\ldots,V_n\}$, a set of trustees $\mathcal{T}=\{T_1,\ldots,V_{N_t}\}$, a set of VC nodes $\{\mathsf{VC}_1,\ldots,\mathsf{VC}_{N_v}\}$ a set of 
  BB nodes $\{\mathsf{BB}_1,\ldots,\mathsf{BB}_{N_b}\}$.
 It provides $\mathsf{Ch}$ with all the above sets. 
 \indent Throughout the game, $\mathcal{A}$ corrupts all the VC nodes a fixed set of $f_b<N_b/3$ BB nodes and a fixed set of $f_t<N_t/3$ trustees. On the other hand, $\mathsf{Ch}$ plays the role of the EA and all the non-corrupted nodes.
  \item $\mathsf{Ch}$ engages with $\mathcal{A}$ in an election preparation interaction following the \emph{Election Authority} protocol. 
  \item $\mathsf{Ch}$ chooses a bit value $b\in\{0,1\}$.
\item  The adversary $\mathcal{A}$ and the 
challenger $\mathsf{Ch}$ engage in an interaction
where $\mathcal{A}$ schedules the voters
which may run concurrently. For each voter $V_\ell\in \mathcal{V}$, the adversary chooses whether $V_\ell$ is corrupted:
\begin{itemize}
\item 
If $V_\ell$ is corrupted, then $\mathsf{Ch}$ provides the credential $s_\ell$ to  $\mathcal{A}$, who will play the
role of $V_\ell$ to cast the ballot. 
\item If $V_\ell$ is not corrupted, then $\mathcal{A}$ provides two
option selections $\langle \mathsf{option}_{\ell}^0, \mathsf{option}_{\ell}^1 \rangle$ to 
the challenger $\mathsf{Ch}$ which
operates on $V_\ell$'s behalf, voting for option $\mathsf{option}_{\ell}^b$.
 The adversary
$\mathcal{A}$ is allowed to observe the network trace.
After a ballot cast,
the challenger $\mathsf{Ch}$ provides to  $\mathcal{A}$: (a) the audit information $\alpha_\ell$ that $V_\ell$ obtains from the protocol, and 
(b) \underline{if $b=0$}, the current view of the internal state of the voter $V_\ell$, $view_{\ell}$, that the challenger obtains during voting, or  \underline{if $b=1$}, a simulated view of the internal state of $V_\ell$ produced by $\mathcal{S}(view_{\ell})$.
\end{itemize}

\item   The adversary $\mathcal{A}$ and the 
challenger $\mathsf{Ch}$ produce the election tally, running the \emph{Trustee} protocol. $\mathcal{A}$ is allowed to observe the network trace of that protocol. 

\item  Finally, $\mathcal{A}$ using all information collected above (including the contents of the BB) outputs a bit $b^*$.
\end{enumerate}   

Denote the set of corrupted voters as $\mathcal{V}_\mathsf{corr}$ and the set of honest
voters as $\tilde{\mathcal{V}} = \mathcal{V} \setminus \mathcal{V}_\mathsf{corr}$.
The game returns a bit which is $1$ if and
    only if the following hold true:
    \begin{enumerate}[(1)]
     \item\label{priv-cond1} $b=b^*$ (i.e., the adversary guesses $b$ correctly).
     \item\label{priv-cond2} $|\mathcal{V}_\mathsf{corr}| \leq \phi$ (i.e., the number of corrupted voters is bounded by $\phi$).
     \item\label{priv-cond3} $f( \langle \mathsf{option}_{\ell}^0 \rangle_{V_\ell \in \tilde{\mathcal{V}}} ) = f( \langle \mathsf{option}_{\ell}^1 \rangle_{V_\ell \in \tilde{\mathcal{V}}} )$ (i.e., the election result w.r.t. the set of voters in $\tilde{\mathcal{V}}$ does not leak $b$). 
    \end{enumerate}
\end{boxfig}
Our privacy definition extends the one used in~\cite{DEMOS} (there referred as Voter Privacy/Receipt-Freeness) to the distributed setting of D-DEMOS. Similarly, voter privacy is defined via a \emph{Voter Privacy} indistinguishability game as depicted in Figure~\ref{fig:priv.game}.
Note that, our system achieves computational weak unlinkability among the privacy classes modeled by \cite{Bohli:2011:RPN:1952982.1952986}.
%
\begin{definition}[\textsc{\textbf{Voter Privacy}}]\label{def:priv}
Let $0<\epsilon<1$ and $n,m,N_v,N_b,N_t\in \mathbb{N}$. Let $\Pi$ be an e-voting system with $n$ voters, $m$ options awith $n$ voters, $N_v$ VC nodes, $N_b$ BB nodes and $N_t$ trustees w.r.t. the election
 function $f$. We say that $\Pi$ achieves \emph{voter privacy} with error $\epsilon$
 for at most $\phi$ corrupted voters, if there is a PPT voter simulator $\mathcal{S}$ such that
 for any PPT adversary $\mathcal{A}$: 
 \[\big|\Pr[G_{\mathsf{priv}}^{\mathcal{A},\mathcal{S},\phi}(1^{\la},n,m,N_v,N_b,N_t)=1] - 1/2\big| = \negl(\la).\]
\end{definition}
In the following theorem, we prove that D-DEMOS (both the IC and the Async version) achieves voter privacy according to Definition~\ref{def:priv}.
\begin{theorem}[\textsc{\textbf{Voter Privacy of D-DEMOS}}]\label{thm:sec_safety}
Assume there is a constant $c\in(0,1)$ such that for any $2^{\lambda^c}$-time adversary $\mathcal{A}$, the advantage of breaking the hiding property of the underlying commitment scheme is $\mathsf{Adv}_{\mathsf{hide}}(\mathcal{A}) = \mathsf{negl}(\lambda)$. Let $c'<c$ be a constant and set $\phi=\lambda^{c'}$. 
Then, D-DEMOS run with $n$ voters, $m$ options, $N_v$ VC nodes, $N_b$ BB nodes and $N_t$ trustees achieves voter privacy for at most $\phi$ corrupted voters.
\end{theorem}

\begin{proof}
To prove voter privacy, we explicitly construct a simulator $\mathcal{S}$ such that we can convert any adversary $\mathcal{A}$ who can win the privacy game $G_{\mathsf{priv}}^{\mathcal{A},\mathcal{S},\phi}(1^{\la},n,m,N_v,N_b,N_t)$ with a non-negligible probability into an adversary $\mathcal{B}$ who can break the hiding assumption of the underlying commitment scheme within $poly(\lambda)\cdot 2^{\lambda^{c'}} << 2^{\lambda^{c}}$ time.

Note that the challenger $\mathsf{Ch}$ is maintaining a coin $b\in\{0,1\}$ and always uses the option $\mathsf{option}_{\ell}^b$ to cast the honest voters' ballots. When $n-\phi <2$, the simulator $\mathcal{S}$ simply outputs the real voters' views. When $n-\phi \geq 2$, consider the following simulator $\mathcal{S}$: At the beginning of the experiment, $\mathcal{S}$ flips a coin $b'\leftarrow\{0,1\}$. Then, for each honest voter $V_\ell$, $\mathcal{S}$ switches the vote codes for option $\mathsf{option}_{\ell}^b$ and $\mathsf{option}_{\ell}^{b'}$.
%
%
\par Due to full VC corruption, $\mathcal{A}$ learns all the vote codes. However, it does not help the adversary to distinguish the simulated view from real view as the simulator only permutes vote codes. We now can show that if $\mathcal{A}$ can win $G_{\mathsf{priv}}^{\mathcal{A},\mathcal{S},\phi}(1^{\la},n,m,N_v,N_b,N_t)$, then we can construct an adversary $\mathcal{B}$ that invokes $\mathcal{A}$ to win the IND-CPA game of the underlying ElGamal encryption.  In the IND-CPA  game, $\mc{B}$ receives as input a public key $\mathsf{pk}$ and executes the following steps:
%
\begin{enumerate}[1.]
\item It submits challenge messages $M_0=0,M_1=1$ and receives challenge ciphertext $C=\mathsf{Enc}_{\mathsf{pk}}(M_{b^*})$, where $b^*$ is the IND-CPA challenge bit for $\mc{B}$.
\item It invokes $\mc{A}$ and simulates $G_{\mathsf{priv}}^{\mathcal{A},\mathcal{S},\phi}(1^{\la},n,m,N_v,N_b,N_t)$, itself being the challenger.
\item $\mc{B}$ flips a coin $b\in\{0,1\}$ and uses the received public key $\mathsf{pk}$ as the election commitment key.
\item At the beginning, $\mc{B}$ generates/guesses all the voters coins, $\textbf{c}=(c_1,c_2,\ldots,c_n)$, and uses the coin $c_\ell$ for all the uncorrupted voter $V_\ell$; if some corrupted voters' coins do not match the guessed ones, start over again. This requires $2^\phi$ expected attempts to guess all the coins correctly.
\item $\mc{B}$ guesses the election tally $T=(T_1,T_2,\ldots,T_m)$, and starts over again if the guess is incorrect. This requires less than ${(n+1)^{m}}$ expected attempts.
\item $\mc{B}$ simulates all the zero knowledge proofs using the guessed voters' coins.
\item $\mc{B}$ guesses/chooses an uncorrupted voter $V_{\ell'}$; the option encoding commitment of $V_{\ell'}$'s ballot for the $i$-th option is set as $
\big({\mathsf{Enc}_{\mathsf{pk}}(T_1)}\cdot{C^{-T_1}}, \ldots, {\mathsf{Enc}_{\mathsf{pk}}(T_{i-1})}\cdot{C^{-T_{i-1}}},\underline{\mathsf{Enc}_{\mathsf{pk}}(T_i)\cdot{C^{-(T_i-1)}}},$ ${\mathsf{Enc}_{\mathsf{pk}}(T_{i+1})}\cdot{C^{-T_{i+1}}},  \ldots,{\mathsf{Enc}_{\mathsf{pk}}(T_m)}\cdot{C^{-T_m}} \big).
$
\par For the rest of the voters, it commits the $i$-th option as $\big(\mathsf{Enc}_{\mathsf{pk}}(0), \ldots, \underline{C\cdot \mathsf{Enc}_{\mathsf{pk}}(0)} , \ldots,  \mathsf{Enc}_{\mathsf{pk}}(0)\big)\;.$
%
\item If $V_\ell$ is corrupted, then $\mc{B}$ provides the credential $s_\ell$ to $\mc{A}$.
\item If $V_{\ell}$ is not corrupted, then $\mc{B}$ receives two option selections $\langle \mathsf{option}_{\ell}^0, \mathsf{option}_{\ell}^1 \rangle$ from $\mc{A}$. It then casts the vote by submitting the vote code corresponding to $\mathsf{option}_{\ell}^{b}$.
\item $\mc{B}$ finishes the election according to the protocol and returns $b^*=1$ if $\mc{A}$ guesses $b$ correctly. 
\end{enumerate}
%
Note that if $C$ encrypts $1$, the commitments on the BB are the same as the ones in a real election; whereas, if $C$ encrypts $0$, the commitments of all the voters are commitments of $0$'s except one honest voter's commitment is the tally results. In the latter case, the adversary $\mc{A}$'s winning probability is exactly $1/2$. Since the zero knowledge proofs are perfectly simulatable, it is easy to see that the advantage of $\mc{B}$ is the same as the advantage of $\mc{A}$.
Moreover, the running time of $\mc{B}$ is ${poly(\lambda)\cdot (n+1)^{m}}\cdot2^\phi=O(2^{\lambda^{c'}})$ steps. By exploiting the distinguishing advantage of $\mathcal{A}$, $\mathcal{B}$ can break the hiding property of the option-encoding commitment scheme in $O(2^{\lambda^{c'}})=o(2^{\lambda^c})$ steps, thus leading to contradiction.
\end{proof}